\documentclass[amsart,11pt,english]{article}
\usepackage{amssymb,amsmath,amsfonts, amsmath,mathtools,hyperref,
amsthm,euscript,mathrsfs
,MnSymbol
,verbatim
,enumerate,multirow
,bbding,slashed
,multicol,color,array, 
esint,babel, tikz,tikz-cd,tikz-3dplot,tkz-graph,pgfplots,ytableau,graphicx,float,rotating,hyperref,geometry,mathdots,savesym,cite, wasysym,amscd,graphicx,pifont,float,setspace,multirow,wrapfig,picture,subfigure, amsthm,
hepth,
 enumitem, enumerate, booktabs
}
\usepackage{adjustbox}
\usepackage[utf8]{inputenc}
\usepackage{thm-restate}
\usepackage{thmtools}
\usepackage{prettyref}
\usepackage[T1]{fontenc}
\usepackage{datetime}
\numberwithin{equation}{section}
\usetikzlibrary{arrows,positioning,decorations.pathmorphing,   decorations.markings, matrix, patterns,shapes}
\hypersetup{colorlinks=true}
\hypersetup{linkcolor=black}
\hypersetup{citecolor=black}
\hypersetup{urlcolor=black}
\theoremstyle{plain}
\newtheorem*{thm*}{Theorem}
\theoremstyle{plain}
\newtheorem{thm}{Theorem}[section]

\newtheorem{lem}[thm]{Lemma}
\newtheorem{prop}[thm]{Proposition}

\theoremstyle{definition}
\newtheorem{defn}[thm]{Definition}
\newtheorem*{defn*}{Definition}

\newtheorem{rem}[thm]{Remark}

\usepackage[normalem]{ulem}
\tikzset{
  big arrow/.style={
    decoration={markings,mark=at position 1 with {\arrow[scale=1.5,#1]{>}}},
    postaction={decorate},
    shorten >=0.4pt},
  big arrow/.default=black}

\numberwithin{equation}{section}

\begin{document}
\thispagestyle{empty}
\begin{titlepage}
\begin{center}
\vspace{4cm}
{\Huge\bfseries  
Mordell-Weil Torsion, Anomalies, and  Phase Transitions. \\
  }
\vspace{2cm}
{%
\LARGE  Mboyo Esole$^{\heartsuit}$, Monica Jinwoo Kang$^\clubsuit$, Shing-Tung Yau$^{\spadesuit,\clubsuit}$ \\}
\vspace{1cm}

{\large $^{\heartsuit}$ Department of Mathematics, Northeastern University, Boston, MA, 02115, USA}\par
{\large $^\clubsuit$ Department of Physics, Harvard University, Cambridge, MA 02138, USA}\par
{\large ${}^{\spadesuit}$Department of Mathematics, Harvard University, Cambridge, MA 02138, USA}\par
 \scalebox{.95}{\tt  j.esole@northeastern.edu,   jkang@fas.harvard.edu, yau@math.harvard.edu }\par
\vspace{3cm}
{ \bf{Abstract:}}\\
\end{center}

We explore how introducing a non-trivial Mordell-Weil group changes the structure of the Coulomb phases 
of a five-dimensional gauge theory from an M-theory compactified on  an elliptically fibered Calabi-Yau threefold with 
an I$_2$+I$_4$ collision of singularities. The resulting gauge theory has a  semi-simple Lie algebra  $\mathfrak{su}(2)\oplus \mathfrak{sp}(4)$ or $\mathfrak{su}(2)\oplus \mathfrak{su}(4)$. 
We compute topological invariants relevant for the physics, such as the Euler characteristic, Hodge numbers, and triple 
intersection numbers. 
We determine the matter representation  geometrically by computing weights via intersection of curves and fibral divisors. 
We fix the number of charged hypermultiplets transforming in each representation by comparing the triple intersection numbers and the one-loop prepotential. 
This condition is enough to fix the number of representation when the Mordell-Weil group is $\mathbb{Z}_2$ but not when it is trivial. 
The vanishing of the fourth power of the curvature forms in the anomaly polynomial is enough to fix the number of representations.
We discuss anomaly cancellations of the six-dimensional uplift. In particular, the gravitational anomaly is also considered as the Hodge numbers 
are computed explicitly without counting the degrees of freedom of the Weierstrass equation. 

\vfill 

\end{titlepage}

\thispagestyle{empty}

\tableofcontents
\thispagestyle{empty}
\clearpage

\setcounter{page}{1}

\section{Introduction}\label{sec:intro}

The study of elliptic fibrations started with Kodaira's seminal papers on minimal elliptic surfaces \cite{Kodaira}. 
 Kodaira classified the possible geometric singular fibers and the class of monodromies around them  for minimal elliptic surfaces. He further derived a formula for the Euler characteristics of minimal elliptic surfaces. 
 N\'eron also classified singular fibers of elliptic surfaces defined by  Weierstrass models, though from an arithmetic point of view \cite{Neron}. 
 Following a theorem of Deligne \cite{Formulaire},  an elliptic fibration with a rational section has a birational Weierstrass model defined over the same base. 
 Tate developed an algorithm to determine the type of singular fibers of a Weierstrass model by manipulating  its coefficients \cite{Tate}. 
   There are new phenomena for higher dimensional elliptic fibrations with direct applications  to string geometry and supergravity theories. 
  First, components of the discriminant locus can intersect each other. This is what Miranda calls ``collisions of singularities'' in his study of  regularizations of elliptic threefolds defined by Weierstrass models \cite{Miranda.smooth}. 
 Secondly, if we start with Weierstrass models, crepant resolutions (when they exist) are not unique: 
 different crepant resolutions of the same Weierstrass models are connected by a network of flops \cite{Matsuki.Weyl,Witten,IMS}, 
  see \cite{EY,ESY1,ESY2,G2,F4} for explicit constructions.  
 These two phenomena are closely related to each other. The collision of singularities are responsible for attaching a representation $\mathbf{R}$ of a Lie algebra $\mathfrak{g}$ to an elliptic fibration \cite{Bershadsky:1996nu}. 
One can then determine a hyperplane arrangement I($\mathfrak{g},\mathbf{R}$)   whose chamber structures have an incidence graph isomorphic to the network of flops between different crepant resolutions of the underlying Weierstrass model.  
 Each crepant resolution corresponds to a different chamber of the extended relative Mori cone of the elliptic fibration.\footnote{
   The use of the hyperplane arrangement I($\mathfrak{g},\mathbf{R}$) as a combinatorial invariant of an elliptic fibration \cite{ESY1,ESY2,Grimm:2011fx,Hayashi:2014kca}  is directly inspired by the study of Coulomb phases of five dimensional gauge theories \cite{IMS} and illustrates how ideas from string 
geometry improves our understanding of the topology and birational geometry of elliptic fibrations. 
These hyperplanes arrangements are also interesting independently of their connections to elliptic fibrations and gauge theories, as they have beautiful combinatorial properties that can be captured by generating functions \cite{EJJN1,EJJN2}.
}

Elliptic fibrations are commonly used in physics to geometrically engineer gauge theories via compactifications of M-theory and F-theory. 
In fact, certain theories are known only through their construction by elliptic fibrations. 
Although many aspects are well understood, the algorithm relating the geometry and physics is still a work in progress.
 The Lie algebra $\mathfrak{g}$ is uniquely determined by the dual graphs of the fibers over the generic points of irreducible components of the reduced discriminant locus. 
The representation $\mathbf{R}$ is captured at the intersections of components of the discriminant locus or in singularities.
The weights of the representations are obtained by intersection numbers of rational curves forming the singular fibers with fibral divisors produced by rational curves  moving over irreducible components of the discriminant locus
 The low energy theories derived by compactification of M-theory or F-theory on Calabi-Yau threefolds are respectively five-dimensional and six-dimensional supergravity theories with eight supercharges \cite{Cadavid:1995bk,Ferrara:1996wv}. 
 We will refer to them as five-dimensional ($5d$) $\mathcal{N}=1$ \cite{IMS} and  six-dimensional ($6d$) $\mathcal{N}=(1,0)$ supergravity theories \cite{Romans:1986er,Sadov:1996zm,Grimm:2015zea}. 

In compactification of M-theory \cite{Cremmer:1978km} on an elliptically fibered threefolds to five-dimensional gauge supergravity theories, the different crepant resolutions of the underlying Weierstrass model are understood as  different Coulomb phases of the same gauge theory, and 
 flops become phase transitions between different Coulomb phases \cite{IMS}. Triple intersection numbers are understood as Chern-Simons levels of the gauge theory and determine the couplings of vector multiplets and the graviphoton of the five-dimensional 
 supergravity theory.    In a five-dimensional $\mathcal{N}=1$ supergravity theory, the Chern-Simons levels and the kinetic terms of the vector multiplets and the graviphoton  are controlled by a cubic prepotential, which admits a   one-loop quantum correction but  is protected by supersymmetry from additional corrections \cite{IMS}. The one-loop quantum correction depends on the number of multiplets $n_{\mathbf{R}_i}$ charged under the irreducible representation $\mathbf{R}_i$ \cite{IMS}. 
In an M-theory compactified on a Calabi-Yau threefold, the full prepotential (including the quantum correction) is given geometrically by the  triple intersection numbers of the divisors \cite{Cadavid:1995bk,Ferrara:1996wv}.

 An important aspect of the dictionary  between elliptic fibrations and gauge theories  is that the elliptic fibration also captures global aspects of the gauge theory: the fundamental group $\pi_1(G)$  of the gauge group is isomorphic to the Mordell-Weil group of the elliptic fibration \cite{Aspinwall:1998xj,Mayrhofer:2014opa,Morrison:2012ei},  see also  \cite[and refs. within]{Baume:2017hxm}. 
One natural question is how the Mordell-Weil group of the elliptic fibration affects these supergravity theories. 
These questions have their mathematical counterparts that are also interesting for their own sake. 
 For example, what is the effect on the Coulomb branch of a five-dimensional gauge theory when a semi-simple group is quotiented by a subgroup of its center? 
This physics question translates in mathematics to
the following: what happens to the extended Mori cone of an elliptically-fibered Calabi-Yau threefold when the Mordell-Weil group is purely torsion? 
Moreover, would this five-dimensional theory with such a Mordell-Weil group still have a six-dimensional uplift with cancellation of anomalies?

  In this paper, we explore  non-trivial models of semi-simple Lie algebra with Mordell-Weil group $\mathbb{Z}_2:=\mathbb{Z}/2\mathbb{Z}$. Specifically, we study the geometry and physics of elliptic fibrations corresponding to the following collisions\footnote{
   Given two Kodaira types $T_1$ and $T_2$, a model of  type $T_1+T_2$ is an elliptic fibration such that the discriminant locus contains two intersecting divisors  $\Delta_1$ and $\Delta_2$, where the generic fiber of  
 $\Delta_i$ is $T_i$ and the generic fiber of any other component of the discriminant locus is an irreducible fiber (such as Kodaira type I$_1$ or II). 
  } 
\begin{equation}\text{I}_2^{\text{\text{ns}}}+\text{I}_4^{\text{\text{ns}}}\quad \text{and}\quad 
\text{I}_2^{\text{\text{ns}}}+\text{I}_4^{\text{\text{s}}},
\end{equation}
 with  a Mordell-Weil group that is either trivial or $\mathbb{Z}_2$. 
 The corresponding Lie algebras are 
 \begin{equation}
 \text{A}_1\oplus\text{C}_2, \quad  \text{A}_1\oplus\text{A}_3.  
 \end{equation}
Such collisions correspond to semi-simple gauge groups  (see Table \ref{Table.Models}):
$$
(\text{SU($2$)}\times \text{Sp($4$)})/\mathbb{Z}_2, \quad \text{SU($2$)}\times \text{Sp($4$)}, \quad  \text{SU($2$)}\times \text{SU($4$)})/\mathbb{Z}_2, \quad \text{SU($2$)}\times \text{SU($4$)}.
$$
These models  are uniquely defined for the  following reasons. 
Since a $\mathbb{Z}_2$ Mordell-Weil group always induces at least an I$_2$ fiber, the collision I$_2+$I$_4$ is the  simplest case  with Mordell-Weil group $\mathbb{Z}_2$ and a semi-simple Lie algebra that is not $\text{A}_1\oplus\text{A}_1$. 
We consider the two possible generic fibers of Kodaira type I$_4$, namely I$_4^{\text{s}}$ and I$_4^{\text{ns}}$. 
 We use I$_2^{\text{ns}}$ because it is the generic reducible fiber induced by the $\mathbb{Z}_2$ Mordell-Weil group (see section \ref{Sec:MW}). 
For comparison, we also study the same collisions with  a trivial Mordell--Weil group. 
These considerations completely fix our models once we use  Weierstrass models \cite{Esole.Elliptic} and assume minimal valuations for all Weierstrass coefficients with respect to the divisors supporting the simple components of the gauge group. 
The SO($4$)-model is studied along the same lines in \cite{SO4}. The SO($3$), SO($5$), and SO($6$) are studied in \cite{MP}. 

The $(\text{SU($2$)}\times \text{SU($4$)})/\mathbb{Z}_2$-model has been studied in \cite{Mayrhofer:2014opa} using toric ambient spaces for the fiber, the $\text{SU($2$)}$ and the $\text{SU($4$)}$ models are individually studied in \cite{ESY1,ES}. 
A specialization of the  $\text{SU($2$)}\times \text{SU($4$)}$-model is studied in \cite{Anderson:2017rpr} in relation to T-branes and also in \cite{Anderson:2015cqy}.
For each of these models, the representation $\mathbf{R}$ is uniquely determined by the geometry of the corresponding elliptic fibrations and are listed on Table \ref{Table.Models}.
For the collisions we consider in this paper, the Mordell-Weil group $\mathbb{Z}_2$ is an obstruction for the presence of fundamental representations. 
This change of matter content has also consequences for the cancellations of anomalies in the  six-dimensional theory. 

 Given a complex Lie algebra $\mathfrak{g}$,  there is a unique simply connected compact Lie group $\widetilde{G}=\exp(\mathfrak{g})$. We denote the center of $\widetilde{G}$ by $Z(\widetilde{G})$. 
All Lie groups sharing the same Lie algebra have the same universal cover $\widetilde{G}=\exp(\mathfrak{g})$ and are quotient of $\widetilde{G}$ by a subgroup  $H$ of its center  $Z(\widetilde{G})$. 
The center of the group $\widetilde{G}/H$ is then $Z(\widetilde{G})/H$  and depends not only on the isomorphic class of $H$ but also on the embedding of $H$ in $Z(\widetilde{G})$. It follows that gauge groups with the same Lie algebra can have different centers and first homotopy groups.   In a gauge theory, the center and the first homotopy group of the group play a crucial role in the description of  non-local operators such as Wilson lines and `tHooft operators. 
Since not all representations of the Lie algebra $\mathfrak{g}$ are coming from a representation of the Lie group,  a non-trivial Mordell-Weil group places restrictions on the representation $\mathbf{R}$. 
The choice of the correct group depends on the Mordell-Weil and is constrained by $\mathbf{R}$.

A representation $\mathbf{R}$ is sometimes enough to completely identity the group $G$ once we know its Lie algebra $\mathfrak{g}$ and its fundamental group. 
 We illustrate this point with two examples that will be the focus of this paper. In both cases, the Lie algebra is derived from an elliptic fibration with collisions of the type I$_2$+I$_4$ and a Mordell-Weil group  $\mathbb{Z}_2$. 
 In the case of the Lie algebra $\frak{g}=\text{A}_1\oplus \text{A}_3$, which corresponds to the simply connected compact group $\text{SU($2$)}\times \text{SU($4$)}$, the center is  $\mathbb{Z}_2\times \mathbb{Z}_4$. Since there is a unique $\mathbb{Z}_2$ subgroup in $\mathbb{Z}_4$, there are three possibilities for embedding  $\mathbb{Z}_2$ in the center $\mathbb{Z}_2\times \mathbb{Z}_4$, namely
\begin{equation}\label{eq:embZ2}
(\mathbb{Z}_2, 1),\quad (1,\ \mathbb{Z}_2), \quad \text{diagonal} \  \mathbb{Z}_2 .
\end{equation}
Hence, the possible quotient groups are 
\begin{equation}
\text{SO($3$)}\times \text{SU($4$)}, \quad \text{SU($2$)}\times \text{SO($5$)}, \quad (\text{SU($2$)}\times \text{SU($4$)})/\mathbb{Z}_2.
\end{equation}
These three groups have the same Lie algebra $\text{A}_1\oplus \text{A}_3$, the same universal cover SU($2$)$\times$ SU($4$), the same first homotopy group, but different centers 
\begin{equation}
\mathbb{Z}_4,  \quad \mathbb{Z}_2\times \mathbb{Z}_2, \quad   \mathbb{Z}_2.
\end{equation} 
The bifundamental representation $(\bf{2},\bf{4})$ of $\text{A}_1\oplus \text{A}_3$ is only compatible with the group $(\text{SU($2$)}\times \text{SU($4$)})/\mathbb{Z}_2$.

For the case of the Lie algebra $\frak{g}=\text{A}_1\oplus \text{C}_2$, which corresponds to the simply connected compact group $\text{SU($2$)}\times \text{Sp($4$)}$, the  center is $\mathbb{Z}_2\times \mathbb{Z}_2$. Equation \eqref{eq:embZ2} gives  the three possible ways to embed a $\mathbb{Z}_2$ in  $\mathbb{Z}_2\times \mathbb{Z}_2$. In this case, the possible quotient groups are 
\begin{equation}
\text{SO($3$)}\times \text{Sp($4$)}, \quad \text{SU($2$)}\times \text{SO($6$)}, \quad (\text{SU($2$)}\times \text{Sp($4$)})/\mathbb{Z}_2 .
\end{equation}
These three groups have the same universal cover $\text{SU($2$)}\times \text{Sp($4$)}$, the same fundamental group $\mathbb{Z}_2$, and the same center  $\mathbb{Z}_2$. 
But the center $\mathbb{Z}_2$ is given by a different embedding in $\mathbb{Z}_2\times \mathbb{Z}_2$. 
Once again, the bifundamental $(\bf{2},\bf{4})$ representation of $\text{A}_1\oplus \text{C}_2$ is only compatible with the last one.

While our  geometric computations are done relative to a base of arbitrary dimension and without assuming the Calabi-Yau condition, to discuss anomaly cancellations, we specifically require the elliptic fibration to be a Calabi-Yau threefold. 
In particular, this requires the base of the fibration $B$ to be a rational surface. We denote its canonical class by $K_B$ or just $K$ when the context is clear. 
Our approach  is rooted in geometry and can be summarized as follows \cite{EY, ESY1,ESY2,ES,F4,G2,Anderson:2017zfm}. 

\begin{enumerate}
\item We work with Weierstrass models. We first  determine a crepant resolution  (see Table \ref{Table.Blowups}). 
\item We  study the fibral structure of the resolved elliptic fibration (see Figure \ref{FS_Nonsplit_Z2_CY}, Figure \ref{FS_Nonsplit_NoZ2_CY}, Figure \ref{FS_Split_Z2_CY}, and Figure \ref{FS_Split_NoZ2_CY}). In particular, we  identify the singular fibers appearing at collisions of singularities. They are composed of rational curves  carrying weights of the representation $\mathbf{R}$ attached to the elliptic fibration.  
\item Using intersection theory, we find the weights of the vertical curves at collisions of singularities  and derive the representation $\bf{R}$. 
Each weight is minus the intersection number of the vertical curve with the fibral divisor not touching the section of the elliptic fibration. 
 One interesting property of our approach is that for complex representations such as the $(\mathbf{2},\mathbf{4})$, or the ($\bf{1},\bf{4}$) of $A_1\oplus A_3$, the weights of the dual representation appear naturally as weights of some vertical curves. Hence,  quaternionic representations are not forced by hand (to respect the  CPT invariance of an underlying physical theory) but are imposed by the geometry itself. 
 
\item We determine the network of flops of each model by studying the hyperplane arrangement I$(\mathfrak{g}, \mathbf{R})$ (See Figure \ref{I2nsI4nsCham}, Figure \ref{ChambersNoZ2}, and Figure \ref{Chambers}). 
 Crepant resolutions and flops are related to the structure of the extended Mori cone.  
The hyperplane arrangement I$(\mathfrak{g}, \mathbf{R})$ is a combinatorial invariant that controls the behavior of the extended Mori cone. 
The chamber structure of   I$(\mathfrak{g}, \mathbf{R})$ corresponds to the $5d$ and $6d$ Coulomb branches.  
The $6d$ Coulomb branch is related to the $5d$ Coulomb branch after a circle compactification and dualizing tensor multiplets into vectors. 
\item The crepant resolution allows us to compute the Euler characteristic of each models  (see Table \ref{tb:EulerGen}, Table \ref{tb:CY3Euler}, Table \ref{tb:CY4Euler}, Table \ref{tb:Elliptic3Euler}, and  Table  \ref{tb:Elliptic4Euler}) and the triple intersection polynomial of the fibral divisors (see section \ref{sec:Triple}). 
 We determine the Euler characteristic by computing the pushforward of the total Chern class defined in homology \cite{Euler}. 
Using pushforward theorems for blowups and projective bundles, we express all topological invariants in terms of the topology of the base.  
In the case of a Calabi-Yau threefold, we compute the Hodge numbers (see Table \ref{hodge}) using the Euler characteristic and the Shioda-Tate-Wazir theorem  \cite[Corollary 4.1]{Wazir}.
 Even though the fiber structure of   the models with a trivial Mordell--Weil group is much more complicated than the one with a $\mathbb{Z}_2$ Mordell--Weil group, 
the Euler characteristic of the former specializes to that of the latter (after imposing $S=-4K-2T$), since both are defined by the same sequence of blowups \cite{Euler}.
\item By comparing the triple intersection numbers with the one-loop quantum correction, we determine constraints on the number of hyperpmultiplets charged under irreducible components of $\mathbf{R}$. In the case of I$_2+$I$_4$-models, the fundamental representations are absent when the  Mordell-Weil group is $\mathbb{Z}_2$.
The possible gauge group and the corresponding representation are presented in Table \ref{Table.Models}. 
\item We check explicitly that the constraints obtained by identifying the triple intersection numbers with the one-loop prepotential are compatible with an uplift to an anomaly-free six-dimensional $\mathcal{N}=(1,0)$ supergravity theory. 
In the theories with a trivial Mordell-Weil group, the triple intersection numbers did not completely fix the number of hypermultiplets transforming in a given irreducible representation. 
But all numbers are fixed by considering the constraints from the vanishing of the coefficient of $\tr R^4$ and $\tr F^4_a$ where $R$ is the Riemann curvature two-form and $F_a$ is the field strength  of the vector field 
$A_a$. 
\end{enumerate}

A representation  $\mathbf{R}$ can be  complex, pseudo-real (quaternionic) or real.\footnote{
When a hypermultiplet is charged under an irreducible complex representation $\mathbf{R}$, its CPT dual is charged under the complex conjugate representation $\overline{R}$. 
 }
A hypermultiplet transforming in a reducible quaternionic representation $\mathbf{R}\oplus\overline{\mathbf{R}}$, such that  $\mathbf{R}$ an irreducible complex representation of the gauge group, is called a {\em full hypermultiplet}. 
The completion of the complex representation $\mathbf{R}$ to the reducible quaternionic representation  $\mathbf{R}\oplus\overline{\mathbf{R}}$  is required by  the CPT theorem.
A  hypermultiplet transforming in an irreducible quaternionic representation is  called a {\em half-hypermultiplet}. 
On a half-hypermultiplet, the pseudo-real representation of the gauge group requires the existence of an anti-linear map that squares to minus the identity. In six (resp. five) dimensional Lorentzian spacetime, this anti-linear map allows the introduction of a  symplectic Majorana Weyl  (resp. symplectic Majorana) condition that reduces by half the number of degrees of freedom. 
In our convention, we avoid the double counting that consists of writing a reducible representation $\mathbf{R}\oplus\overline{\mathbf{R}}$ even when $\mathbf{R}$ is an irreducible pseudo-real representation. A full hypermultiplet has multiplicity $n_{\mathbf{R}}+n_{\overline{\mathbf{R}}}=2 n_{\mathbf{R}}$ while a half-hypermultiplet has multiplicity $n_{\mathbf{R}}$. 
Since we do not double count half-hypermultiplets, we do not introduce factors of 1/2 in the anomaly polynomial when $\mathbf{R}$ is pseudo-real. 
Still, a hypermultiplet in a pseudo-real representation contributes to the anomaly polynomial half as much as a full hypermultiplet.

 Comparing the triple intersection numbers with the one-loop prepotential gives constraints linear in  $n_{\mathbf{R}_i}$ that are sometimes enough to completely determine all the numbers $n_{\mathbf{R}_i}$. 
This is for example the case for many models with a simple group \cite{ES,F4,G2} and also here for the models with a $\mathbb{Z}_2$ Mordell-Weil group, namely the $(\text{SU($2$)}\times \text{Sp($4$)})/\mathbb{Z}_2$-model and the 
$(\text{SU($2$)}\times \text{SU($4$)})/\mathbb{Z}_2$-model. The $\text{SU($2$)}\times \text{Sp($4$)}$ and $\text{SU($2$)}\times \text{SU($4$)}$-models have in addition hypermultiplets transforming in  fundamental representations and equating the one-loop prepotential and the triple intersection numbers does not provide enough constraints  to fix all $n_{\mathbf{R}_i}$.   We can  fix all  $n_{\mathbf{R}_i}$ in different ways. 
\begin{itemize} 
\item Firstly, by using Witten's formula that asserts that the number of hypermultiplets transforming in the adjoint representation of an irreducible component $G_a$ of the gauge group is the (arithmetic) genus of the curve $S_a$ supporting that group. 
Witten's formula uses M2-branes to  study the quantization of the curve $S_a$ seen as a moduli space. Each (anti)-holomorphic 1-form on $S_a$ is responsible for a hypermultiplet. 
\item Secondly, by using  techniques of intersecting branes to directly count $n_{\mathbf{R}_i}$ as intersection numbers between components of the discriminant locus.  
This method assumes that the components intersect relatively well. For certain representations (such as the traceless second antisymmetric  representation of Sp($2n$), this method has to be completed with a generalization of Witten's formula. 
\item Finally, we can use the existence of an anomaly-free six-dimensional supergravity theory. Matching the triple intersection numbers with the one-loop prepotential and asking for the   vanishing of the coefficients of $\tr R^4$ and $\tr F^4_a$ in the anomaly polynomial of the $6d$ $\mathcal{N}=1$ supergravity theory are enough to fix all the $n_{\mathbf{R}_i}$. Here $R$ is the Riemann curvature two-form and $F_a$ are the  field-strength of the vector fields. 
The vanishing of these quartic terms is a necessary condition to cancel anomalies by the Green-Schwarz mechanism. 
\end{itemize}
We find that the last method is the most satisfying and general as it only relies on basic aspects of supergravity theories in five and six dimensions: exactness of the one-loop quantum correction and the cancellation of anomalies of the six-dimensional theory.  
We checked that all three methods give the same result.  
 In a six-dimensional uplifted theory, the number of hypermultiplets are the same, but we have to adjust the number of vector and  tensor multiplets. 
Tensor multiplets and spinors require cancellations of anomalies.  Following Sadov, there is a geometric formulation of the Green-Schwarz mechanism for elliptically fibered Calabi-Yau threefolds \cite{Sadov:1996zm}. Hence, the number of multiplets, which are computed from the triple intersection numbers, can be used to check if the Green-Schwarz mechanism cancels the anomalies in the $6d$ theory uplift.
This method was introduced in \cite{Bonetti:2011mw} and implemented explicitly in \cite{ES} for SU($n$)-models, in \cite{F4} for F$_4$-models, and \cite{G2} for G$_2$, Spin($7$), and Spin($8$)-models.

Our convention for counting the multiplicity of representations is inspired by the geometry and ease the  comparison with the intersecting brane picture. 
In the case of  the SU($2$)$\times$ SU($4$)-model, the matter representation contains the reducible quarternionic representation $(\mathbf{1},\mathbf{4})\oplus(\mathbf{1},\bar{\mathbf{4}})$ with the multiplicity $n_{\bf{1},\bf{4}}=n_{\bf{1},\bf{\bar{4}}}=-T(4K+S+2T)$. 
The total multiplicity  corresponds to the number of intersection points between the divisor $T$ supporting SU($4$) and the irreducible component  $\Delta'$ of the discriminant not supporting any gauge group. At collision the fiber I$_4^s$ enhances to an I$_5^s$, producing the weight $\boxed{0,1,-1}$ of the fundamental representation $\mathbf{4}$ of $\mathfrak{su}(4)$ and the weight  $\boxed{-1,1,0}$ of the anti-fundamental representation $\mathbf{4}$ of $\mathfrak{su}(4)$. See Tables 
 \ref{Table.Matter}, \ref{Rep.noZ2model}, and \ref{I2sI4sNoZ2}.
 The representation $(\bf{2},\bf{1})$ is pseudo-real and $n_{\bf{2},\bf{1}}=-2S(4K+S+2T)$, which is geometrically the number of intersection $S\cdot V(\tilde{b}_8)$. 
\begin{table}[htb]
 \begin{center}
 \def\arraystretch{1.1}
\begin{tabular}{|c|c|c|}
\hline 
F-theory on $Y$ & M-theory on $Y$ &  F-theory on $Y\times S^1$  \\
$\downarrow$ & $\downarrow$& $\downarrow$\\
$6d$ $\mathcal{N}=(1,0)$ sugra  &  $5d$ $\mathcal{N}=1$ sugra &  $5d$ $\mathcal{N}=1$ sugra \\
 \hline 
$n_V^{(6)}=h^{1,1}(Y)-h^{1,1}(B)-1$ & \multicolumn{2}{c|}{$n_V^{(5)}=n_V^{(6)}+n_T+1=h^{1,1}(Y)-1$ }\\
 $n_H^0=h^{2,1}(Y)+1$& \multicolumn{2}{c|}{$n_H^0=h^{2,1}(Y)+1$}\\
 $n_T=h^{1,1}(B)-1$ &  \multicolumn{2}{c|}{}   \\
 \hline 
 \end{tabular}
 \end{center}
 \caption{\label{Table:FM} Compactification of F-theory and M-theory  on a Calabi-Yau threefold $Y$.  We assume that all tensor multiplets  in the five dimensional theory are dualized to vector multiplets. The number of neutral hypermultiplets are the same in five and six dimensions, but  $n_V^{(5)}=n_V^{(6)}+n_T+1$.  }
 \end{table}

\clearpage
\begin{table}[hbt]
\begin{center}
\begin{tabular}{| c |  l |  c|}
\hline
Models & \hspace{4cm}{Algebraic data}  &\#   Flops\\
\hline
\hline
& $F=y^2z-(x^3+a_2x^2z+st^2xz^2)$  & \multirow{3}{*}{3}\\
I$_2^{\text{\text{ns}}}$+I$_4^{\text{\text{ns}}}$&  $\Delta=s^2 t^4 (a_2^2-4 s t^2)$ &\\
\cline{2-2}
$\text{MW}=\mathbb{Z}_2$&
 $G=(\text{SU($2$)}\times \text{Sp($4$)})/\mathbb{Z}_2$ &\\
& $\mathbf{R}=(\bf{3},\bf{1})\oplus (\bf{1},\bf{10})\oplus (\bf{2},\bf{4})\oplus (\bf{1},{5})$ &  \\
& $\chi=-4 (9 K^2+8 K\cdot T+3 T^2)$& \\
\hline
\hline
 & $F=y^2z-(x^3+a_2x^2z+\widetilde{a}_4st^2xz^2+\widetilde{a}_6s^2t^4z^3)$ &  \multirow{3}{*}{3} \\
I$_2^{\text{\text{ns}}}$+I$_4^{\text{\text{ns}}}$&  $\Delta= s^2 t^4 (4 a_2^3\widetilde{a}_6-a_2^2\widetilde{a}_4^2-18 a_2 \widetilde{a}_4 \widetilde{a}_6 s t^2+4 a_4^3 s t^2+27 \widetilde{a}_6^2 s^2 t^4)$ & \\
\cline{2-2}
$\text{MW}=\{1\}$& $G=\text{SU($2$)}\times \text{Sp($4$)}$ & \\
& $\mathbf{R}=(\bf{3},\bf{1})\oplus (\bf{1},\bf{10})\oplus (\bf{2},\bf{4})\oplus (\bf{1},{5})\oplus (\bf{2},\bf{1})\oplus (\bf{1},{4})$ &    \\
& $\chi=-2 (30 K^2+15 K\cdot S+30 K\cdot T+3 S^2+8 S\cdot T+10 T^2)$& \\
\hline
\hline
  & $F=y^2z+a_1xyz-(x^3+\widetilde{a}_2tx^2z+st^2xz^2)$ &  \multirow{3}{*}{12}\\
I$_2^{\text{\text{ns}}}$+I$_4^{\text{\text{s}}}$ & $\Delta=s^2 t^4 \left(a_1^4+8 a_1^2 \widetilde{a}_2 t+16 \widetilde{a}_2^2 t^2-64 s t^2\right)$ & \\
\cline{2-2}
$\text{MW}=\mathbb{Z}_2$& $G=(\text{SU($2$)}\times \text{SU($4$)})/\mathbb{Z}_2$ & \\
& $\mathbf{R}=(\bf{3},\bf{1})\oplus(\bf{1},\bf{15})\oplus (\bf{2},\bf{4})\oplus (\bf{2},\bf{\bar{4}})\oplus(\bf{1},\bf{6})$& \\
&$\chi=-12 \left(3 K^2+3 K\cdot T+T^2\right)$ & \\
\hline 
\hline
 & $F=y^2z+a_1xyz-(x^3+\widetilde{a}_2tx^2z+\widetilde{a}_4st^2xz^2+\widetilde{a}_6s^2t^4 z^3)$&  \multirow{3}{*}{20} \\
I$_2^{\text{\text{ns}}}+$I$_4^{\text{\text{s}}}$& $\Delta=s^2 t^4 \left(a_1^4+8 a_1^2 \widetilde{a}_2 t+16 \widetilde{a}_2^2 t^2-64 s t^2\right)$& \\
\cline{2-2}
$\text{MW}=\{1\}$ & $G=\text{SU($2$)}\times \text{SU($4$)}$ & \\
& $\mathbf{R}=(\bf{3},\bf{1})\oplus(\bf{1},\bf{15})\oplus(\bf{2},\bf{4})\oplus(\bf{2},\bf{\bar{4}})\oplus(\bf{1},\bf{6})\oplus  (\bf{2},\bf{1})\oplus
(\bf{1},\bf{4})\oplus(\bf{1},\bf{\bar{4}})$ &  \\
& $\chi=-2 \left(30 K^2+15 K\cdot S+32 K\cdot T+3 S^2+8 S\cdot T+10 T^2\right)$ & \\
\hline 
\end{tabular}
\caption{Weierstrass models, discriminant loci,  gauge groups and representations. 
$F$ is the defining equation of the Weierstrass model, $\Delta$ is its discriminant, $G$ is the gauge group,    $Z(G)$  is the center of $G$  (isomorphic to the Mordell-Weil group MW of the elliptic fibration),  $\mathbf{R}$ is the matter representation, and $\chi$ is the Euler characteristic of a crepant resolution of a Calabi-Yau that is elliptically fibered with a  defining equation  $F=0$. 
The column ``\# Flops'' gives the number of distinct crepant resolutions , or equivalently, the number of chambers in the hyperplane arrangement I$(\mathfrak{g},\mathbf{R}$). The number of flops also  corresponds to the number of Coulomb phases of a 
five-dimensional  supergravity theory  with eight supercharges obtained by a compactification of M-theory on this elliptic fibration. 
 The Euler characteristic of the models with trivial  Mordell--Weil groups specializes to those of the models with  $\mathbb{Z}_2$ Mordell--Weil groups after imposing the relation $S=-4K-2T$.}
\label{Table.Models}
\end{center}
\end{table}

\begin{table}[H]
\begin{center}
\scalebox{.86}{
\begin{tabular}{|l| l l  |  l   |   l|}
\hline 
\hspace{1.5cm} $G$ &   \multicolumn{2}{c|}{Adjoint} & Bifundamental  & (Traceless) Antisymmetric, Fundamental  \\
\hline
$(\text{SU($2$)}\times \text{Sp($4$)})/\mathbb{Z}_2$ &
$n_{\bf{3},\bf{1}}=g_S$  &$ n_{\bf{1},\bf{10}}=g_T$ & $n_{\bf{2},\bf{4}}=S\cdot T$& 
$n_{\bf{1},\bf{5}}=g_T-1+\frac{1}{2}T\cdot V(a_2)$  \\
\hline
&$n_{\bf{3},\bf{1}}=g_S$ &$ n_{\bf{1},\bf{10}}=g_T$  &$n_{\bf{2},\bf{4}}=S\cdot T$ &$ 
n_{\bf{1},\bf{5}}=g_T-1+\frac{1}{2}T\cdot V(a_2)$  \\
$\text{SU($2$)}\times \text{Sp($4$)}$& && & $ n_{\bf{2},\bf{1}}=S\cdot V(\tilde{b}_8),\ n_{\bf{1},\bf{4}}= T \cdot V(\tilde{b}_8)$ \\
\hline
$(\text{SU($2$)}\times \text{SU($4$)})/\mathbb{Z}_2$ & 
$n_{\bf{3},\bf{1}}=g_S$ &$ n_{\bf{1},\bf{15}}=g_T$& { $n_{\bf{2},\bf{4}}+n_{\bf{2},\bf{\bar{4}}}=S\cdot  T$}&$ n_{\bf{1},\bf{6}}=T\cdot V(a_1)$\\
\hline
&  $n_{\bf{3},\bf{1}}=g_S$ &$ n_{\bf{1},\bf{15}}=g_T$   &{$n_{\bf{2},\bf{4}}+n_{\bf{2},\bf{\bar{4}}}=S\cdot  T$}&$ n_{\bf{1},\bf{6}}=T\cdot V(a_1)$\\
$\text{SU($2$)}\times \text{SU($4$)}$& &  & & $n_{\bf{2},\bf{1}}= S\cdot V(\tilde{b}_{8}),\ { n_{\bf{1},\bf{4}}+n_{\bf{1},\bf{\bar{4}}}= T\cdot V({\widetilde{b}_8})
}$\\
\hline
\end{tabular}}
\caption{  Geometrical interpretation of the number of representations of the matter content. 
In the last column,  $n_{\bf{1},\bf{5}}$ is the number of traceless antisymmetric matter in $\text{Sp($4$)}$, $n_{\bf{1},\bf{6}}$ is the number of antisymmetric matter in $\text{SU($4$)}$, $n_{\bf{1},\bf{4}}$ is the fundamental representation of 
$\text{SU($4$)}$, and $n_{\bf{2},\bf{1}}$ is the number of fundamental representation in $\text{SU($2$)}$.  
For SU($4$) theories, we have $[V(\tilde{b}_8)]=2(-4 T-S-2T)$.
 In presence of a $\mathbb{Z}_2$, with our choice of Weierstrass models, we have $S=-4K-2T$. 
\label{Table.Matter}
 }
\end{center}
\end{table}
\clearpage

\section{Preliminaries}
In this section, we review some basic notions used throughout the paper. 
We also introduce our notation and conventions. 

\subsection{Weierstrass models and Deligne's formulaire }
\label{sec:Wmodel}

We follow the notation of \cite{Formulaire}. 
Let  $\mathscr{L}$ be a line bundle over  a normal quasi-projective variety  $B$.  We define the projective bundle (of lines)
\begin{equation}
\pi: X_0=\mathbb{P}_B[\mathscr{O}_B\oplus \mathscr{L}^{\otimes 2}\oplus \mathscr{L}^{\otimes 3}]\longrightarrow B.
\end{equation} 
The relative projective coordinates of $X_0$ over $B$ are denoted $[z:x:y]$,  where $z$, $x$, and $y$ are defined  by the natural injection of 
 $\mathscr{O}_B$,   $\mathscr{L}^{\otimes 2}$, and $\mathscr{L}^{\otimes 3}$ into $\mathscr{O}_B\oplus \mathscr{L}^{\otimes 2}\oplus \mathscr{L}^{\otimes 3}$, respectively. Hence, 
  $z$ is a section of $\mathscr{O}_{X_0}(1)$, $x$ is a section of $\mathscr{O}_{X_0}(1)\otimes \pi^\ast \mathscr{L}^{\otimes 2}$, and
$y$ is a section of  $\mathscr{O}_{X_0}(1)\otimes \pi^\ast \mathscr{L}^{\otimes 3}$.

\begin{defn}
 A  Weierstrass model is an elliptic fibration $\varphi: Y\to B$  cut out by the zero locus of  a section of the  
line bundle $\mathscr{O}(3)\otimes \pi^\ast \mathscr{L}^{\otimes 6}$ in $X_0$. 
\end{defn}
The most general Weierstrass equation is written in the notation of Tate as
\begin{equation}
y^2z+ a_1 xy z + a_3  yz^2 -(x^3+ a_2 x^2 z + a_4 x z^2 + a_6 z^3) =0,
\end{equation} 
where $a_i$ is a section of $\pi^\ast \mathscr{L}^{\otimes i}$. 
The line bundle $\mathscr{L}$ is called the {\em fundamental line bundle} of the Weierstrass model $\varphi:Y\to B$. It can be defined directly from $Y$ as 
$\mathscr{L}=R^1 \varphi_\ast Y$. 
Following Tate and Deligne, we introduce the following quantities 
\begin{align}
\begin{cases}
b_2 &= a_1^2 + 4 a_2\\
b_4 &= a_1 a_3 + 2 a_4\\
b_6 &= a_3^2 + 4 a_6\\
b_8 &= a_1^2 a_6 - a_1 a_3 a_4 + 4 a_2 a_6 + a_2 a_3^2 - a_4^2\\
c_4 &= b_2^2 - 24 b_4\\
c_6 &= -b_2^3 + 36 b_2 b_4 - 216 b_6\\
\Delta &= -b_2^2 b_8 - 8 b_4^3 - 27 b_6^2 + 9 b_2 b_4 b_6\\
j& = {c_4^3}/{\Delta}
\end{cases}
\end{align}
The  $b_i$ ($i=2,3,4,6)$ and $c_i$   ($i=4,6$) are  sections of $\pi^\ast \mathscr{L}^{\otimes i}$. 
The discriminant $\Delta$ is a section of $\pi^\ast \mathscr{L}^{\otimes 12}$. 
They satisfy the two relations
\begin{align}
1728 \Delta=c_4^3-c_6^2, \quad 4b_8 = b_2 b_6 - b_4^2.
\end{align}
Completing the square in $y$ gives 
\begin{equation}
zy^2 =x^3 +\tfrac{1}{4}b_2 x^2 + \tfrac{1}{2} b_4 x + \tfrac{1}{4} b_6.
\end{equation}
Completing the cube in $x$ gives the short form of the Weierstrass equation
\begin{equation}
zy^2 =x^3 -\tfrac{1}{48} c_4 x z^2 -\tfrac{1}{864} c_6 z^3.
\end{equation}

\subsection{Elliptic fibrations with Mordell-Weil group $\mathbb{Z}_2$}\label{Sec:MW}
A generic  Weierstrass model with a Mordell-Weil torsion subgroup $\mathbb{Z}_2$ is given by the following theorem, which is a direct consequence of a classic result in the study of elliptic curves in number theory \cite[\S5 of Chap 4]{Husemoller}, and was first discussed in a string theoretic setting by Aspinwall and Morrison \cite{Aspinwall:1998xj}. 
\begin{thm}\label{Thm:WZ2}
An elliptic fibration over a smooth base $B$ and with  Mordell-Weil group $\mathbb{Z}_2$ is birational to the  following (singular) Weierstrass model:
\begin{equation} 
zy^2= x(x^2 + a_2 x z+ a_4z^2).\label{WZ2}
\end{equation}
\end{thm}
The section $x=y=0$ is the generator of the $\mathbb{Z}_2$ Mordell-Weil group  and $x=z=0$ is the neutral element of the  Mordell-Weil group. The discriminant of this Weierstrass model is
\begin{equation}
\Delta=16 a_4^2(a_2^2 -4 a_4).
\end{equation}

\subsection{$G$-models and representations}

 The locus of  points in the base that lie below singular fibers of a non-trivial elliptic fibration is a Cartier divisor $\Delta$ called the discriminant locus of the elliptic fibration. 
We denote the irreducible components of the reduced discriminant by $\Delta_i$. 
If the elliptic fibration is minimal, the type of the  fiber over the generic point of  $\Delta_i$ of $\Delta$ has a dual graph that is  an affine Dynkin diagram $\widetilde{\mathfrak{g}}_i^t$.  If the generic fiber over $\Delta_i$ is irreducible, $\mathfrak{g}_i$ is the trivial Lie algebra since $\widetilde{\mathfrak{g}}_i^t=\widetilde{A}_0$.
The Lie algebra   $\mathfrak{g}$ associated with the elliptic fibration $\varphi:Y\to B$ is then  the direct sum 
$\mathfrak{g}=\bigoplus_i \mathfrak{g}_i, $
where the Lie algebra $\mathfrak{g}_i$ is such that the affine Dynkin diagram  $\widetilde{\mathfrak{g}}^t_i$ is the dual graph of the fiber over the generic point of $\Delta_i$.

When an elliptic $\varphi: Y \to B$ has trivial Mordell-Weil group, the compact Lie group $G$ associated with the elliptic fibration  $\varphi$ is semi-simple,  simply connected,  and is given by the formula $
G:= \exp (\bigoplus_i \mathfrak{g}_i)$, 
$\exp(\mathfrak{g})$ is the unique compact simply connected Lie group whose Lie algebra is $\mathfrak{g}$, 
 the index $i$ runs over all the irreducible components of the reduced discriminant locus. The  Lie algebra $\mathfrak{g}_i$ is such that the affine Dynkin diagram  $\widetilde{\mathfrak{g}}^t_i$ is the dual graph of the fiber over the generic point of the irreducible component $\Delta_i$ of the reduced discriminant 
of the elliptic fibration. 

If the elliptic fibration has a Mordell-Weil group $T\times \mathbb{Z}^r$ with torsion subgroup $T$ and rank $r$, then the group $G$ is  the quotient  $\widetilde{G}/T \times U(1)^r$ where $\widetilde{G}=\exp (\bigoplus_i \mathfrak{g}_i)$. 
Defining the quotient $\widetilde{G}/T$  properly requires a choice of  embedding of $T$ in the center $Z(\widetilde{G})$ of the simply connected group $\widetilde{G}$. 
The center of $G$ is then $Z(\widetilde{G})/T$. 
The representation $\mathbf{R}$ attached to the elliptic fibration constrains the possibilities and is sometimes enough to completely determine the embedding of $T$ in $Z(\widetilde{G})$. 

\begin{defn}[$G$-model]\label{Defn:Gmodel}
Let $G$ be a compact, simply connected Lie group.
An elliptic fibration  $\varphi:Y\to B$ with an associated Lie group $G$  is called a \emph{$G$-model}.
\end{defn}

 \begin{defn}[$\mathcal{K}_1+\cdots+\mathcal{K}_n$-model]
 Let $\mathcal{ K}_1, \mathcal{K}_2, \cdots, \mathcal{K}_n$ be Kodaira types and  $S_1, \cdots, S_n$ be a smooth divisor of a projective variety $B$. 
 An elliptic fibration  $\varphi:Y\longrightarrow B$ over $B$  is said to be  a  \emph{$\mathcal{K}_1+\cdots+\mathcal{K}_n$-model} if
the reduced discriminant locus $\Delta(\varphi)$ contains components $S_i$ as an irreducible component a divisor  $S\subset B$ such that the generic fiber over $S_i$ is of type $\mathcal{ K}_i$ and any other  generic fiber  of a component of the discriminant locus different from the   $S_i$ is irreducible.
 \end{defn}
 
 The generic fibers degenerate into  fibers of different types  over points of codimension two in the base, usually intersections of irreducible components $\Delta_i$ of the reduced discriminant or codimension-one singularities of the reduced discriminant.

\begin{defn}[Weight vector of a vertical curve]\label{Def:WeightVerticalCurve}
Let $C$ be a vertical curve, i.e.  a curve contained in a fiber of the elliptic fibration.
Let $S$ be an irreducible component of the reduced discriminant of the elliptic fibration $\varphi: Y\to B$. 
The pullback of  $\varphi^* S$ has irreducible components $D_0, D_1, \ldots, D_n$, where $D_0$ is the component touching the section of the elliptic fibration.   
The {\em weight vector} of $C$ over $S$ is  by definition the vector ${\varpi}_S(C)=(-D_1\cdot C, \ldots, -D_n\cdot C)$ of intersection numbers $D_i\cdot C$ for $i=1,\ldots, n$. 
\end{defn}

The irreducible curves of the degenerations over codimension-two loci only give a subset of weights of a representation $\mathbf{R}$.  Hence, we need an algorithm that retrieves the full representation $\mathbf{R}$ given only a few of its weights.  
This problem can be addressed systematically using the notion of a saturated set of weights introduced by  Bourbaki  \cite[Chap.VIII.\S 7. Sect. 2]{Bourbaki.GLA79}.

\begin{defn}[Saturated set of weights]
A set $\Pi$ of integral weights is {\em saturated} if for any weight $\varpi\in\Pi$ and any simple root $\alpha$,  the weight $\varpi-i\alpha$ is also in $\Pi$ for any $i$ such that $0\leq i\leq \langle \varpi,\alpha\rangle$. 
A saturated set has {\em highest weight} $\lambda$ if $\lambda\in\Lambda^+$ and $\mu\prec\lambda$ for any $\mu\in \Pi$. 
\end{defn}
\begin{defn}[Saturation of a subset]
Any subsets $\Pi$ of weights is contained in a unique smallest saturated subset. We call it the saturation of $\Pi$. 
\end{defn}
\begin{prop}[{\cite[Chap. III \S 13.4]{Humphreys}}] \quad
\begin{enumerate}[label=(\alph*)]
\item A saturated set of weights is invariant under the action of the Weyl group.  
\item The saturation of a set of weights $\Pi$ is finite if and only if  the set $\Pi$ is finite.
\item A saturated set with highest weight $\lambda$  consists of all dominant weights lower than or equal to $\lambda$ and their conjugates under the Weyl group.  
\end{enumerate}
\end{prop}

\begin{thm}[{Bourbaki, \cite[Chap.VIII.\S 7. Sect. 2, Corollary to Prop. 5]{Bourbaki.GLA79}}]\label{Thm:R-Saturation} 
Let $\Pi$ be a finite saturated  set of weights. Then there exists a finite dimensional  $\mathfrak{g}$-module  whose set of weights is $\Pi$. 
\end{thm}

\begin{defn}[Representation of  a $G$-model]\label{Def.Rep}
To a $G$-model, we associate a representation $\mathbf{R}$ of the Lie algebra $\mathfrak{g}$ as follows. 
 The weight vectors of the irreducible vertical   rational curves of the fibers over codimension-two points form  a set $\Pi$  whose  saturation defines uniquely a representation  $\mathbf{R}$ by  Theorem \ref{Thm:R-Saturation}. 
  We call this representation  $\mathbf{R}$  the representation of the $G$-model. 
 \end{defn}
Definition \ref{Def.Rep}  is a formalization of the method of Aspinwall and Gross \cite[\S 4]{Aspinwall:1996nk}.
Note that we always get the adjoint representation as a summand of $\mathbf{R}$. See also \cite{Marsano,Morrison:2011mb}.

\subsection{Intersection theory}
All our intersection theory computations come down to  the following three theorems. 
The first one is a theorem of Aluffi which gives the Chern class after a blowup along a local complete intersection. 
The second theorem is a pushforward theorem that provides a user-friendly method to compute invariant of the blowup space in terms of the original space. 
The last theorem is a direct consequence of functorial properties of the Segre class and gives a simple method to pushforward analytic expressions in the Chow ring of a projective bundle to  the Chow ring of its base.

\begin{thm}[Aluffi, {
{\cite[Lemma 1.3]{Aluffi_CBU}}}]
\label{Thm:AluffiCBU}
Let $Z\subset X$ be the  complete intersection  of $d$ nonsingular hypersurfaces $Z_1$, \ldots, $Z_d$ meeting transversally in $X$.  Let  $f: \widetilde{X}\longrightarrow X$ be the blowup of $X$ centered at $Z$. We denote the exceptional divisor of $f$  by $E$. The total Chern class of $\widetilde{X}$ is then:
\begin{equation}
c( T{\widetilde{X}})=(1+E) \left(\prod_{i=1}^d  \frac{1+f^* Z_i-E}{1+ f^* Z_i}\right)  f^* c(TX).
\end{equation}
\end{thm}

\begin{thm}[Esole--Jefferson--Kang,  see  {\cite{Euler}}] \label{Thm:Push}
    Let the nonsingular variety $Z\subset X$ be a complete intersection of $d$ nonsingular hypersurfaces $Z_1$, \ldots, $Z_d$ meeting transversally in $X$. Let $E$ be the class of the exceptional divisor of the blowup $f:\widetilde{X}\longrightarrow X$ centered 
at $Z$.
 Let $\widetilde{Q}(t)=\sum_a f^* Q_a t^a$ be a formal power series with $Q_a\in A_*(X)$.
 We define the associated formal power series  ${Q}(t)=\sum_a Q_a t^a$, whose coefficients pullback to the coefficients of $\widetilde{Q}(t)$. 
 Then the pushforward $f_*\widetilde{Q}(E)$ is
 $$
  f_*  \widetilde{Q}(E) =  \sum_{\ell=1}^d {Q}(Z_\ell) M_\ell, \quad \text{where} \quad  M_\ell=\prod_{\substack{m=1\\
 m\neq \ell}}^d  \frac{Z_m}{ Z_m-Z_\ell }.
 $$ 
\end{thm}

\begin{thm}[{See  \cite{Euler} and  \cite{AE1,AE2,Fullwood:SVW,Esole:2014dea}}]\label{Thm:PushH}
Let $\mathscr{L}$ be a line bundle over a variety $B$ and $\pi: X_0=\mathbb{P}[\mathscr{O}_B\oplus\mathscr{L}^{\otimes 2} \oplus \mathscr{L}^{\otimes 3}]\longrightarrow B$ a projective bundle over $B$. 
 Let $\widetilde{Q}(t)=\sum_a \pi^* Q_a t^a$ be a formal power series in  $t$ such that $Q_a\in A_*(B)$. Define the auxiliary power series $Q(t)=\sum_a Q_a t^a$. 
Then 
$$
\pi_* \widetilde{Q}(H)=-2\left. \frac{{Q}(H)}{H^2}\right|_{H=-2L}+3\left. \frac{{Q}(H)}{H^2}\right|_{H=-3L}  +\frac{Q(0)}{6 L^2},
$$
 where  $L=c_1(\mathscr{L})$ and $H=c_1(\mathscr{O}_{X_0}(1))$ is the first Chern class of the dual of the tautological line bundle of  $ \pi:X_0=\mathbb{P}(\mathscr{O}_B \oplus\mathscr{L}^{\otimes 2} \oplus\mathscr{L}^{\otimes 3})\rightarrow B$.
\end{thm}

\subsection{Anomaly Cancellation} \label{sec:anomaly}

The matter content of the six-dimensional ${\cal N}=(1,0)$ supergravity theory are given by \cite{Green:1984bx}
\begin{center}
\begin{tabular}{r c l}
supergravity multiplets: & \hspace{2cm}& $(g_{\mu\nu}, B^-_{\mu\nu}, \psi_\mu{}^{A-})$\\
tensor multiplets: & & $(B_{\mu\nu}^+,\chi^{A+},\sigma)$\\
vector multiplets: & & $(A_\mu, \lambda^{A-})$\\
hypermultiplets: & & $(4\phi, \zeta^+)$ 
\end{tabular}\\
\end{center}
where $\mu,\nu=0,\ldots, 5$ label spacetime indices, $A=1,2$ labels the fundamental representation 
of the $R$-symmetry SU($2$), and $\pm$ denotes the chirality of Weyl spinors or the self-duality ($+$)  or anti-self-duality ($-$) of the field strength of antisymmetric two-forms. 
The gravitini $\psi_\mu{}^{A-}$, the tensorini $\chi^{A+}$, and gaugini $ \lambda^{A-}$ are symplectic Majorana Weyl spinors. 
The hyperino $\zeta^+$ is a Weyl spinor invariant under the $R$-symmetry group SU($2$)$_R$. The scalar manifold of the tensor multiplets is the symmetric space $\text{SO}(1, n_T)/\text{SO}(n_T)$ where $n_T$ is the number of tensor multiplets. 
The scalar manifold of the hypermultiplet is a quaternionic-K\"ahler manifold of quaternionic dimension $n_H$, where $n_H$ is the number of hypermultiplets.

Consider a gauged six-dimensional $\mathcal{N}=(1,0)$ supergravity theory  with a semi-simple gauge group $G=\prod_a G_a$, $n_V^{(6)}$ vector multiplets, $n_T$ tensor multiplets, and $n_H$ hypermultiplets consisting of $n_H^0$ neutral hypermultiplets and $n_H^{ch}$ charged under a representation $\bigoplus_i\mathbf{R}_i$ of the gauge group with $\mathbf{R}_i=\bigotimes_{a} \mathbf{R}_{i,a}$, where $\mathbf{R}_{i,a}$ is an irreducible representation of the simple component  $G_a$ of the semi-simple group $G$. The vector multiplets belongs to the adjoint of the gauge group (hence $n_V=\dim G$). 
As discussed in the introduction, CPT invariance requires the representation to be quaternionic.

By denoting the zero weights of a representation $\bf{R}_i$ as $\bf{R}^{(0)}_i$, the charged dimension of the hypermultiplets in representation $\bf{R}_i$ is given by $\dim{\bf{R}_i}-\dim{\bf{R}_{i}^{(0)}}$, as the hypermultiplets of zero weights are considered neutral. For a representation $\bf{R}_i$, $n_{\bf{R}_i}$ denotes the multiplicity of the representation $\bf{R}_i$. Then the number of charged hypermultiplets is simply
\begin{equation}
n_H^{ch}=\sum_{i} n_{\bf{R}_i} \left( \dim{\bf{R}_i}-\dim{\bf{R_{i}^{(0)}}} \right).
\end{equation}
The total number of hypermultiplets is the sum of the neutral hypermultiplets and the charged hypermultiplets. 
For a compactification on a Calabi-Yau threefold $Y$, the  number of  neutral hypermultiplets is  $h^{2,1}(Y)+1$ \cite{Cadavid:1995bk}. The number of each multiplet is
\begin{align}
n_V^{(6)}&=\dim{G}, \quad n_T=h^{1,1}(B)-1=9-K^2 , \\
n_H&=n_H^0+n_H^{ch}=h^{2,1}(Y)+1+\sum_{i} n_{\bf{R}_i} \left( \dim{\bf{R}_i}-\dim{\bf{R_{i}^{(0)}}} \right),
\end{align}
where the (elliptically fibered) base $B$ is a rational surface. From the Hodge number $h^{2,1}(Y)$ of the Calabi-Yau threefolds in Table \ref{hodge} of section \ref{eulerchar}, the number of hypernultiplets are computed for each model. 

The anomaly polynomial I$_8$ has a pure gravitational contribution of the form  $\tr R^4$ where $R$ is the Riemann tensor thought of as a $6\times 6$ real matrix of two-form values. 
To apply the Green-Schwarz mechanism, its coefficient is required to vanish, i.e.
\begin{equation}
n_H-n_V^{(6)}+29n_T-273=0.
\end{equation}
In order to define the remainder terms of the anomaly polynomial I$_8$, 
we define 
\begin{equation}
X^{(n)}_{a}=\tr_{\bf{adj}}F^n_a -\sum_{i}n_{\bf{R}_{i,a}}\tr_{\bf{R}_{i,a}}F^n_a, \quad 
Y_{ab}=\sum_{i,j} n_{\bf{R}_{i,a}, \bf{R}_{j,b}} \tr_{\bf{R}_{i,a}}F^2_a \tr_{\bf{R}_{j,b}}F^2_b,
\end{equation}
where $n_{\bf{R}_{i,a}, \bf{R}_{j,b}}$ is the number of hypermultiplets transforming in the representation $(\mathbf{R}_{i,a},\mathbf{R}_{j,b})$ of $G_a\times G_b$.
The trace identities for a representation $\mathbf{R}_{i,a}$ of a simple group $G_a$ are
\begin{equation}
\tr_{\bf{R}_{i,a}} F^2_a=A_{\bf{R}_{i,a}} \tr_{\bf{F}_a} F^2_a , \quad \tr_{\bf{R}_{i,a}} F^4_a=B_{\bf{R}_{i,a}} \tr_{\bf{F}_a} F^4_a+C_{\bf{R}_{i,a}} (\tr_{\bf{F}_a} F^2_a )^2
\end{equation}
with respect to a  reference representation $\bf{F}_a$ for each simple component $G_a$. To define the anomaly polynomial I$_8$, we introduce the following expressions: 
\begin{align}
X^{(2)}_a&=\left(A_{a,\bf{adj}}-\sum_{i}n_{\bf{R}_{i,a}}A_{\bf{R}_{i,a}}\right)\tr_{\bf{F}_a}F^2_a, \\
X^{(4)}_a&=\left(B_{a,\bf{adj}}-\sum_{i}n_{\bf{R}_{i,a}}B_{\bf{R}_{i,a}}\right)\tr_{\bf{F}_a}F^4_a
+\left(C_{a,\bf{adj}}-\sum_{i}n_{\bf{R}_{i,a}}C_{\bf{R}_{i,a}}\right)(\tr_{\bf{F}_a}F^2_a)^2 , \\
Y_{ab}&=\sum_{i,j} n_{\bf{R}_{i,a},\bf{R}_{j,b}} A_{{R}_{i,a}} A_{\bf{R}_{j,b}} \tr_{\bf{F}_a}F^2_a \tr_{\bf{F}_b}F^2_b.
\end{align}
For each simple component $G_a$, the anomaly polynomial  I$_8$ has a pure gauge contribution proportional to the quartic term $\tr F^4_a$ that is required to vanish in order to factorize I$_8$: 
$$
B_{a,\bf{adj}}-\sum_{i}n_{\bf{R}_{i,a}}B_{\bf{R}_{i,a}}=0.
$$
When the coefficients of all quartic terms  ($\tr R^4$ and $\tr F^4_a$) vanish,  the remaining part of the anomaly polynomial I$_8$ is
\begin{equation}
I_8=\frac{K^2}{8} (\tr R^2)^2 +\frac{1}{6}\sum_{a} X^{(2)}_{a} \tr R^2-\frac{2}{3}\sum_{a} X^{(4)}_{a}+4\sum_{a<b}Y_{ab}.
\end{equation}
The anomalies are canceled by the Green-Schwarz mechanism when I$_8$ factorizes \cite{Green:1984bx,Sagnotti:1992qw,Schwarz:1995zw}.

There is a subtlety on the representations that are charged on more than a simple component of the group, as it affects not only  $Y_{ab}$ but also $X^{(2)}_a$ and $X^{(4)}_a$. Consider a representation $(\bf{R_1},\bf{R_2})$ for of a semisimple group with two simple components $G=G_1\times G_2$, where $\bf{R_a}$ is a representation of $G_a$. Then this representation contributes to $n_{\bf{R_1}}$ $\dim{\bf{R_2}}$ times, and contributes to $n_{\bf{R_2}}$ $\dim{\bf{R_1}}$ times:
\begin{equation}
n_{\bf{R_1}}=\dim{\bf{R_2}} \ n_{\bf{R_1},\bf{R_2}}, \quad n_{\bf{R_2}}=\dim{\bf{R_1}} \ n_{\bf{R_1},\bf{R_2}}.
\end{equation}

If the coefficients of $\tr R^4$ and $\tr F_a^4$ vanishes and $G=G_1\times G_2$,  the remainder of the anomaly polynomial is given by
\begin{equation}
I_8 =\frac{K^2}{8} (\tr R^2)^2 +\frac{1}{6} (X^{(2)}_{1} +X^{(2)}_{2}) \tr R^2-\frac{2}{3} (X^{(4)}_{1}+X^{(4)}_{2})+4Y_{12} .
\end{equation}
 If I$_8$ factors as $\frac{1}{2}\Omega_{ij} X^{(4)}_i X^{(4)}_j$, then the anomaly is cancelled by adding the counter term 
$\Omega_{ij} B_i \wedge   X_j^{(4)}$ to the Lagrangian. 
The modification of the field strength $H^{(i)}$ of  the antisymmetric tensor $B^{(i)}$ are $H^{(i)}=dB^{(i)} +\omega^{(i)}$, where $\omega^{(i)}$ is a proper combination of Yang-Mills and gravitational Chern-Simons terms.

For $\text{SU($2$)}$ with the adjoint representation $\bf{3}$ and the fundamental representation $\bf{F}=\bf{2}$ as the reference representation,
\begin{align}
\begin{split}
\mathrm{tr}_{\bf{3}}\  F^2_1=4 \mathrm{tr}_{\bf{2}}\  F^2_1, \quad 
\mathrm{tr}_{\bf{3}}\  F^4_1=8 (\mathrm{tr}_{\bf{2}}\  F^2_1)^2 , \quad
\mathrm{tr}_{\bf{2}}\  F^4_1 =\frac{1}{2} (\mathrm{tr}_{\bf{2}}\  F^2_1)^2.
\label{eq:SU2trace}
\end{split}
\end{align}
The trace identities for $\text{Sp($4$)}$ is given by 
\begin{align}
\begin{split}
\mathrm{tr}_{\bf{10}}\  F^2_2=6 \mathrm{tr}_{\bf{4}}\  F^2_2, \quad 
\mathrm{tr}_{\bf{10}}\  F^4_2 =12 \mathrm{tr}_{\bf{4}}\  F^4_2+3 (\mathrm{tr}_{\bf{4}}\  F^2_2)^2 ,\\
\mathrm{tr}_{\bf{5}}\  F^2_2=2 \mathrm{tr}_{\bf{4}}\  F^2_2, \quad 
\mathrm{tr}_{\bf{5}}\  F^4_2 =-4 \mathrm{tr}_{\bf{4}}\  F^4_2+3 (\mathrm{tr}_{\bf{4}}\  F^2_2)^2 .
\label{eq:Sp4trace}
\end{split}
\end{align}
For $\text{SU($4$)}$, the trace identities are given by
\begin{align}
\begin{split}
\mathrm{tr}_{\bf{15}}\  F^2_2=8 \mathrm{tr}_{\bf{4}}\  F^2_2, \quad 
\mathrm{tr}_{\bf{15}}\  F^4_2 =8 \mathrm{tr}_{\bf{4}}\  F^4_2+6 (\mathrm{tr}_{\bf{4}}\  F^2_2)^2 ,\\
\mathrm{tr}_{\bf{6}}\  F^2_2=2 \mathrm{tr}_{\bf{4}}\  F^2_2, \quad 
\mathrm{tr}_{\bf{6}}\  F^4_2 =-4 \mathrm{tr}_{\bf{4}}\  F^4_2+3 (\mathrm{tr}_{\bf{4}}\  F^2_2)^2 .
\label{eq:SU4trace}
\end{split}
\end{align}
\section{Summary of results}

In this section, we summarize the results of this paper. 

\subsection{Crepant resolutions}

The models we consider in this paper are given by the following Weierstrass models. 

\begin{align}
(\text{SU($2$)}\times \text{Sp($4$)})/\mathbb{Z}_2: \qquad &y^2 z-(x^3+a_2x^2z+st^2 xz^2)=0\\
\text{SU($2$)}\times \text{Sp($4$)}:\qquad&y^2z-(x^3+a_2x^2z+\widetilde{a}_4st^2xz^2+\widetilde{a}_6s^2t^4z^3)=0\\
(\text{SU($2$)}\times \text{SU($4$)})/\mathbb{Z}_2: \qquad &y^2z+a_1xyz-(x^3+\widetilde{a}_2tx^2z+st^2xz^2)=0\\
\text{SU($2$)}\times \text{SU($4$)}: \qquad &y^2z+a_1xyz-(x^3+\widetilde{a}_2tx^2z+\widetilde{a}_4st^2xz^2+\widetilde{a}_6s^2t^4 z^3)=0
\end{align}

Given a complete intersection $Z$ of hypersurfaces  $Z_i=V(z_i)$ in a variety $X$, we denote the blowup of $\widetilde{X}\to X$ along $Z$ with exceptional divisor $E=V(e)$ as 
$$
\begin{tikzpicture}
	\node(X0) at (0,0){$\widetilde{X}$};
	\node(X1) at (3,0){$X.$};
	\draw[big arrow] (X1) -- node[above,midway]{$(z_1,\ldots,z_n|e)$} (X0);	
 \end{tikzpicture}
$$

We use the following sequence of blowups to determine a crepant resolution of each models. 

\begin{table}[htb]
\begin{center}
\begin{tabular}{|c|l|}
\hline
$\begin{matrix}
    \\
(\text{SU($2$)}\times \text{Sp($4$)})/\mathbb{Z}_2\\ 
\text{SU($2$)}\times \text{Sp($4$)}
\end{matrix}
$ & 
{$\begin{matrix}
\\
{\begin{tikzcd}[column sep=huge]  X_0  \arrow[leftarrow]{r} {(x,y,s|e_1)} & \arrow[leftarrow]{r} {(x,y,t|w_1)}  X_1 &X_2  \arrow[leftarrow]{r}{(x,y,w_1|w_2)} &X_3. \end{tikzcd}}
\end{matrix}$}
\\
 &\\
\hline
$\begin{matrix}
    \\
(\text{SU($2$)}\times \text{SU($4$)})/\mathbb{Z}_2\\ 
\text{SU($2$)}\times \text{SU($4$)}
\end{matrix}
$
& 
$\begin{matrix}
\\
\begin{tikzcd}[column sep=huge]  X_0  \arrow[leftarrow]{r} {(x,y,s|e_1)} & \arrow[leftarrow]{r} {(x,y,t|w_1)}  X_1 &X_2  \arrow[leftarrow]{r}{(y,w_1|w_2)} &X_3 \arrow[leftarrow]{r}{(x,w_2|w_3)} & X_4 \end{tikzcd}
\end{matrix}
$
\\
&  \\
\hline
\end{tabular}
\end{center}
\caption{Sequence of blowups for crepant resolutions used in the paper.\label{Table.Blowups} }
\end{table}

The class of the fibral divisors are given as follows. 

\begin{table}[htb]
\begin{center}
\begin{tabular}{|c||c|c||c|c|c|c|}
\hline
&  $D_0^{\text{s}}$ &$D_1^{\text{s}}$ & $D_0^t$ & $D_1^t$ & $D_2^t$& $D_3^t$   \\
\hline
\hline
$(\text{SU($2$)}\times \text{Sp($4$)})/\mathbb{Z}_2$ &\multirow{2}{*}{$S-E_1$} & \multirow{2}{*}{$E_1$} &\multirow{2}{*}{$W_0-W_1$} &\multirow{2}{*}{$W_1-W_2$} & \multirow{2}{*}{$W_2$} &  \\
$\text{SU($2$)}\times \text{Sp($4$)}$ & & & & & & \\ 
\hline
$(\text{SU($2$)}\times \text{SU($4$)})/\mathbb{Z}_2$ & \multirow{2}{*}{$S-E_1$}& \multirow{2}{*}{$E_1$} & \multirow{2}{*}{$W_0-W_1$}& \multirow{2}{*}{$W_1-W_2$}& \multirow{2}{*}{$W_2-W_3$}& \multirow{2}{*}{$W_3$} \\
$\text{SU($2$)}\times \text{SU($4$)}$ & & & & & & \\
\hline
\end{tabular}
\end{center}
\caption{Class of the fibral divisors}
\label{tb:FibralDiv}
\end{table}

\begin{figure}
\centering{\includegraphics[scale=.85]{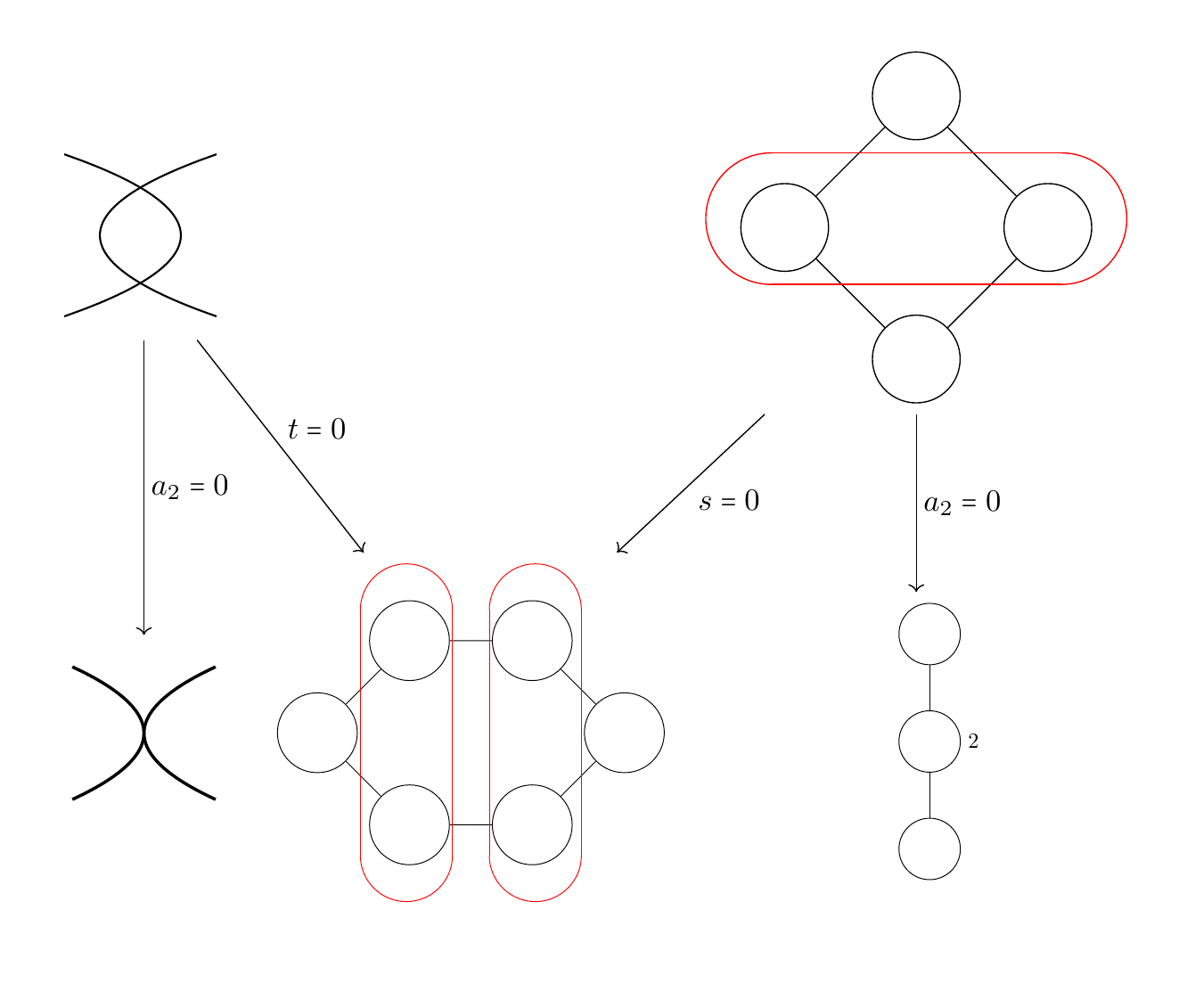}}
\caption{This is the fiber structure of $G=(\text{SU($2$)}\times \text{Sp($4$)})/\mathbb{Z}_2$ uptil codimension two for the Calabi-Yau threefolds. \label{FS_Nonsplit_Z2_CY}}		
\vspace{1cm}				
\centering{\includegraphics[scale=1]{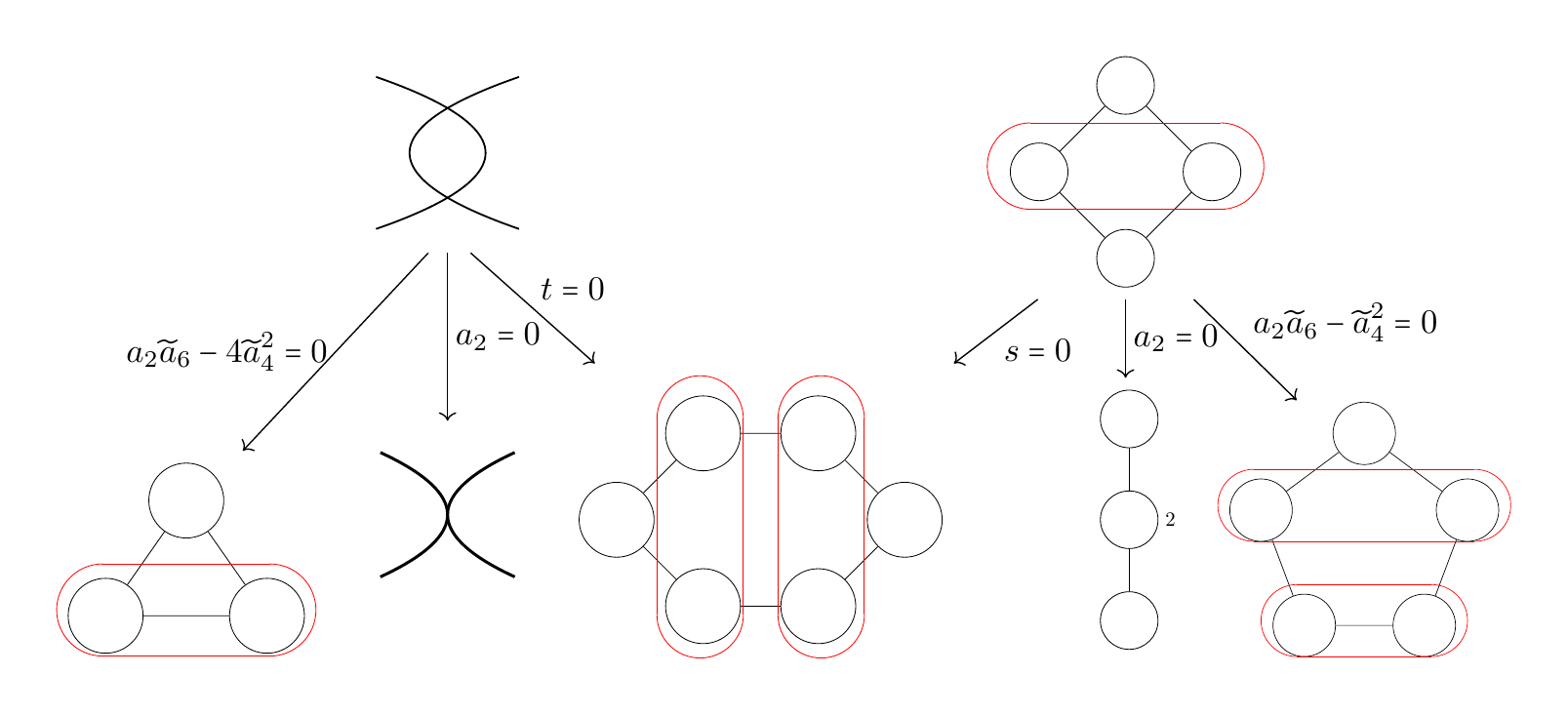}}				
\caption{This is the fiber structure of $G=\text{SU($2$)}\times \text{Sp($4$)}$  uptil codimension two for the Calabi-Yau threefolds. \label{FS_Nonsplit_NoZ2_CY} }
\end{figure}

\begin{figure}
\centering{\includegraphics[scale=.85]{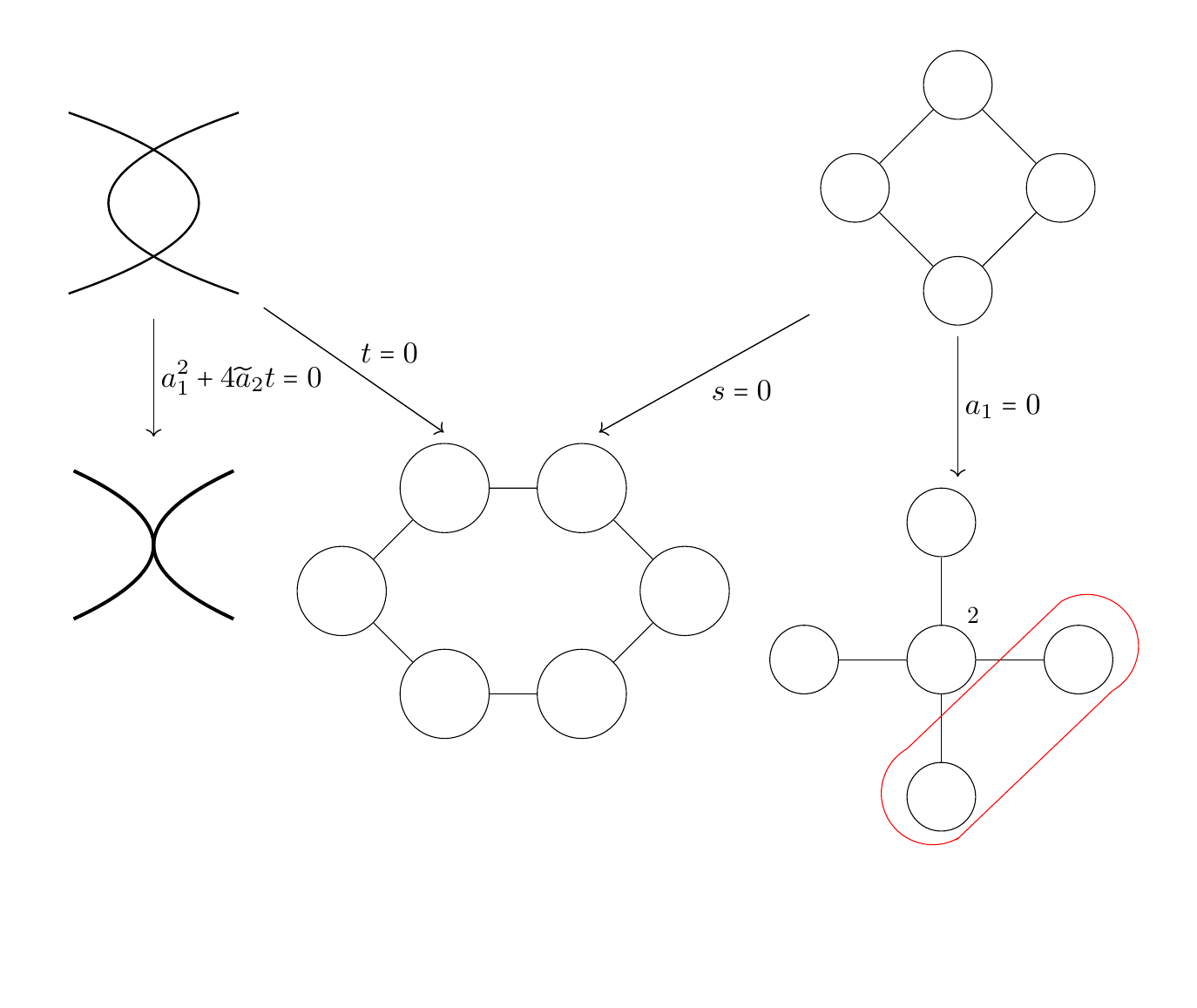}}
\caption{This is the fiber structure of $G=(\text{SU($2$)}\times \text{SU($4$)})/\mathbb{Z}_2$  uptil codimension two for the Calabi-Yau threefolds.   \label{FS_Split_Z2_CY}}
\vspace{.5cm}
\centering{\includegraphics[scale=1]{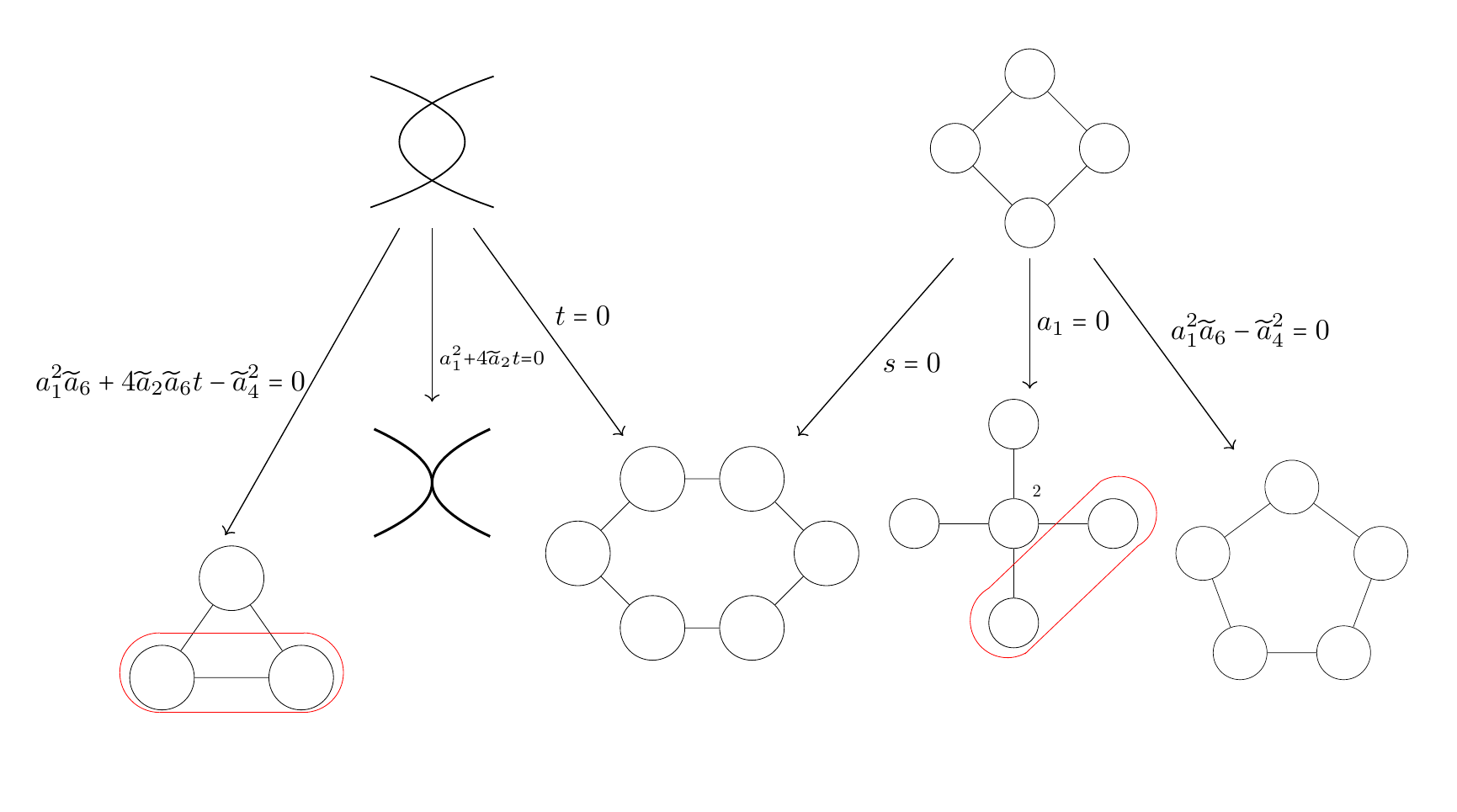}}
\caption{This is the fiber structure of $G=\text{SU($2$)}\times \text{SU($4$)}$  uptil codimension two for the Calabi-Yau threefolds. \label{FS_Split_NoZ2_CY}}
\end{figure}

\subsection{Euler characteristics} \label{eulerchar}

Using $p$-adic integration and the Weil conjecture, Batyrev proved the following theorem.

\begin{thm}[Batyrev, \cite{Batyrev.Betti}]
\label{thm:Batyrev}
Let $X$ and $Y$ be irreducible birational smooth $n$-dimensional projective algebraic varieties 
over $\mathbb{C}$. Assume that there exists a birational rational map $\varphi: X   - \rightarrow Y$ that does not 
change the canonical class. Then $X$ and $Y$ have the same Betti numbers. 
\end{thm}
 Batyrev's result was strongly inspired by string dualities, in particular by the work of Dixon, Harvey, Vafa, and Witten \cite{Dixon:1986jc}. 
As a direct consequence of Batyrev's theorem, the Euler characteristic of a crepant resolution of a variety with Gorenstein canonical singularities is independent on the choice of resolution. 
We identify the Euler characteristic as the degree  of the total  (homological) Chern class of a crepant resolution $f: \widetilde{Y }\longrightarrow Y$ of a Weierstrass model $Y\longrightarrow B$:
$$
\chi(\widetilde{Y})=\int c(\widetilde{Y}).
$$
We then use the birational invariance of the degree under the pushfoward to express the Euler characteristic as a class in the Chow ring of the projective bundle $X_0$. We subsequently push this class forward to the base to obtain a rational function depending upon only the total Chern class of the base $c(B)$, the first Chern class $c_1(\mathscr L)$, and the class $S$ of the divisor in $B$:
$$
\chi(\widetilde{Y})=\int_B \pi_* f_* c(\widetilde{Y}).
$$
In view of Theorem \ref{thm:Batyrev}, this Euler characteristic is independent of the choice of a crepant resolution. 

 We compute the Euler characteristic for each model by considering a particular crepant resolution as listed in Table \ref{Table.Blowups}. For the models with  Mordell-Weil group $\mathbb{Z}_2$, the divisors $S$ and $T$ satisfy the following linear relation since $a_4= s t^2$: 
\begin{equation}
4L=S+2T.
\label{eq:ST}
\end{equation}
The generating function of the Euler characteristics are presented in Table \ref{tb:EulerGen}, which produce the Euler characteristics for elliptic threefolds and fourfolds as listed in \ref{tb:Elliptic3Euler} and \ref{tb:Elliptic4Euler}. The Calabi-Yau condition imposes $L=-K$. For each model, we present the Euler characteristic of Calabi-Yau threefolds and fourfolds respectively in Table \ref{tb:CY3Euler} and \ref{tb:CY4Euler}.

\begin{table}[H]
\begin{center}
\begin{tabular}{|c|c|c|}
\hline
Models & Generating Function \\
 \hline
$(\text{SU($2$)}\times \text{Sp($4$)})/\mathbb{Z}_2$ 
 & \Large{$\frac{4 \left(2 L^2 (5 T+3)+L (3-5 (T-1) T)-3 T^2\right)}{(2 L+1) (T+1) (4 L-2 T+1)}c[B]$} \\
\hline
$\text{SU($2$)}\times \text{Sp($4$)}$  & \Large{$\frac{2 \left(T (-6 L^2 (5 S+4)+L (2 S-3) (5 S+4)+S (7 S+8))\right)}{(2 L+1) (S+1) (T+1) (-6 L+2 S+4 T-1)} c[B]$} \\
& \Large{$+\frac{2 \left(3 (2 L+1) (S^2-L (3 S+2))+2 T^2 (2 L (5 S+4)+7 S+5)\right)}{(2 L+1) (S+1) (T+1) (-6 L+2 S+4 T-1)} c[B]$} \\
\hline
$(\text{SU($2$)}\times \text{SU($4$)})/\mathbb{Z}_2$ & \Large{$\frac{12 \left(2 L T+L-T^2\right)}{(T+1) (4 L-2 T+1)}c[B]$} \\
\hline
 & \Large{$\frac{4T^3 (7 S+5) +4T^2 \left(14 S^2-7 (7 S+5) L+6 S-5\right)}{(S+1) (T+1) (S-4 L+2 T-1) (2 S-6 L+4 T-1)} c[B]$} \\
$\text{SU($2$)}\times \text{SU($4$)}$  & \Large{$+\frac{2T \left(12 (7 S+5) L^2+(2-S (49 S+43)) L+(7 S (S+1)-8) S\right)}{(S+1) (T+1) (S-4 L+2 T-1) (2 S-6 L+4 T-1)} c[B]$} \\
 & \Large{$+\frac{6 (S-4 L-1) \left(S^2-3 S L-2 L\right)}{(S+1) (T+1) (S-4 L+2 T-1) (2 S-6 L+4 T-1)} c[B]$} \\
\hline
\end{tabular}
\caption{Generating function for Euler Characteristics}
\label{tb:EulerGen}
\end{center}
\end{table}

\vspace{-0.5cm}
\begin{table}[H]
\begin{center}
\begin{tabular}{|c|c|}
\hline
Models & Euler Characteristics \\
 \hline
$(\text{SU($2$)}\times \text{Sp($4$)})/\mathbb{Z}_2$ & $-4 (9 K^2+8 K\cdot T+3 T^2)$ \\
\hline
$\text{SU($2$)}\times \text{Sp($4$)}$  & $-2 (30 K^2+15 K\cdot S+30 K\cdot T+3 S^2+8 S\cdot T+10 T^2)$ \\
\hline
$(\text{SU($2$)}\times \text{SU($4$)})/\mathbb{Z}_2$  & $-12 \left(3 K^2+3 K\cdot T+T^2\right)$ \\
\hline
$\text{SU($2$)}\times \text{SU($4$)}$  & $-2 \left(30 K^2+15 K\cdot S+32 K\cdot T+3 S^2+8 S\cdot T+10 T^2\right)$ \\
\hline
\end{tabular}
\caption{Calabi-Yau threefolds}
\label{tb:CY3Euler}
\end{center}
\end{table}
\vspace{-0.5cm}
\begin{table}[H]
\begin{center}
\begin{tabular}{|c|c|}
\hline
Models & Euler Characteristics \\
\hline
$(\text{SU($2$)}\times \text{Sp($4$)})/\mathbb{Z}_2$  & $-12 \left(c_2 K+12 K^3+16 K^2 T+8 K T^2+T^3\right)$ \\
\hline
\multirow{ 2}{*}{$\text{SU($2$)}\times \text{Sp($4$)}$}  & $-6 \left(2 c_2 K+60 K^3+49 K^2 S+98 K^2 T+14 K S^2\right.$ \\
 & $\left.+56 K S T +56 K T^2+S^3+8 S^2 T+16 S T^2+10 T^3\right)$ \\
\hline
$(\text{SU($2$)}\times \text{SU($4$)})/\mathbb{Z}_2$  & $-12 \left(c_2 K+12 K^3+17 K^2 T+8 K T^2+T^3\right)$ \\
\hline
$\text{SU($2$)}\times \text{SU($4$)}$  & $-6 \left(2 c_2 K+60 K^3+49 K^2 S+100 K^2 T+14 K S^2 \right. $ \\
 & $\left. +56 K S T+56 K T^2+S^3+8 S^2 T+16 S T^2+10 T^3\right) $ \\
\hline
\end{tabular}
\caption{Calabi-Yau fourfolds}
\label{tb:CY4Euler}
\end{center}
\end{table}
\vspace{-0.5cm}

\begin{table}[H]
\begin{center}
\begin{tabular}{|c|c|}
\hline
Models & Euler Characteristics \\
\hline
$(\text{SU($2$)}\times \text{Sp($4$)})/\mathbb{Z}_2$  & $4 \left(3 c_1 L-12 L^2+8 L T-3 T^2\right)$ \\
\hline
$\text{SU($2$)}\times \text{Sp($4$)}$  & $2 \left(6 c_1 L-36 L^2+15 L S+30 L T-3 S^2-8 S T-10 T^2\right)$ \\
\hline
$(\text{SU($2$)}\times \text{SU($4$)})/\mathbb{Z}_2$  & $12 \left(c_1 L-4 L^2+3 L T-T^2\right)$ \\
\hline
$\text{SU($2$)}\times \text{SU($4$)}$  & $2 \left(6 c_1 L-36 L^2+15 L S+32 L T-3 S^2-8 S T-10 T^2\right)$ \\
\hline
\end{tabular}
\caption{Elliptic threefolds}
\label{tb:Elliptic3Euler}
\end{center}
\end{table}
\vspace{-0.5cm}
\begin{table}[H]
\begin{center}
\begin{tabular}{|c|c|}
\hline
Models & Euler Characteristics \\
\hline
$(\text{SU($2$)}\times \text{Sp($4$)})/\mathbb{Z}_2$  & $4 \left(-12 c_1 L^2+8 c_1 L T-3 c_1 T^2+3 c_2 L+48 L^3 \right.$\\
& $\left. -56 L^2 T+27 L T^2-3 T^3\right)$ \\
\hline
& $2 \left(-36 c_1 L^2+15 c_1 L S+30 c_1 L T-3 c_1 S^2-8 c_1 S T \right.$ \\
$\text{SU($2$)}\times \text{Sp($4$)}$ & $\left.-10 c_1 T^2+6 c_2 L+216 L^3-162 L^2 S-324 L^2 T+45 L S^2T\right.$\\
& $\left.+176 L S +178 L T^2-3 S^3-24 S^2 T-48 S T^2-30 T^3\right)$ \\
\hline
$(\text{SU($2$)}\times \text{SU($4$)})/\mathbb{Z}_2$  & $12 \left(-4 c_1 L^2+3 c_1 L T-c_1 T^2+c_2 L+16 L^3 \right.$\\
& $\left.-20 L^2 T+9 L T^2-T^3\right)$ \\
\hline
  & $2 \left(-36 c_1 L^2+15 c_1 L S+32 c_1 L T-3 c_1 S^2-8 c_1 S T\right.$ \\
$\text{SU($2$)}\times \text{SU($4$)}$ & $\left.-10 c_1 T^2+6 c_2 L+216 L^3-162 L^2 S-332 L^2 T+45 L S^2\right.$ \\
 & $\left. +176 L S T+178 L T^2-3 S^3-24 S^2 T-48 S T^2-30 T^3\right) $ \\
\hline
\end{tabular}
\caption{Elliptic fourfolds}
\label{tb:Elliptic4Euler}
\end{center}
\end{table}
\subsection{Hodge numbers for Calabi-Yau elliptic threefolds}
 Using motivitic integration, Kontsevich shows in his famous ``String Cohomology'' Lecture at Orsay that birational equivalent Calabi-Yau varieties have the same class in the completed Grothendieck ring \cite{Kontsevich.Orsay}. 
Hence, birational equivalent Calabi-Yau varieties have the same  Hodge-Deligne polynomial, Hodge numbers, and Euler characteristic. 
 In this section, we compute the Hodge numbers of crepant resolutions of Weierstrass models in the case of  Calabi-Yau threefolds.

\begin{thm}[Kontsevich, (see \cite{Kontsevich.Orsay})]
Let $X$ and $Y$ be birational equivalent Calabi-Yau varieties over the complex numbers. Then $X$ and $Y$ have the same Hodge numbers. 
\end{thm}
\begin{rem}
In Kontsevich's theorem, a Calabi-Yau variety is a nonsingular complete projective variety of dimension $d$ with a trivial canonical divisor. 
To compute Hodge numbers in this section, we use the following stronger definition of a Calabi-Yau variety.
\end{rem}

\begin{defn}\label{defn:CY}
A \emph{Calabi-Yau variety} is a smooth compact projective variety $Y$ of dimension $n$ with a trivial canonical class and such that  $H^i(Y,\mathscr{O}_X)=0$ for $1\leq i\leq n-1$.
\end{defn}

 We first recall some basic definitions and relevant classical theorems.  
\begin{thm}[Noether's formula]
 If $B$ is a smooth compact, connected, complex surface with canonical class $K_B$ and Euler number  $c_2$, then
$$
\chi(\mathscr{O}_B)=1-h^{0,1}(B)+h^{0,2}(B),\quad    \chi(\mathscr{O}_B)=\frac{1}{12}(K^2+c_2).
$$
\end{thm}
When $B$ is a smooth compact rational surface, we have a simple  expression of   $h^{1,1}(B)$ as a function of $K^2$ using the following lemma. 
\begin{lem}\label{lem:NoetherRational}
Let $B$ be a smooth compact rational surface with  canonical class $K$. Then 
\begin{equation}
h^{1,1}(B)=10-K^2.
\end{equation}
\end{lem}
\begin{proof}
Since $B$ is a rational surface, $h^{0,1}(B)=h^{0,2}(B)=0$. Hence $c_2=2+h^{1,1}(B)$ and the  lemma follows from  Noether's formula. 
\end{proof}

We now compute  $h^{1,1}(Y)$ using the Shioda-Tate-Wazir theorem  \cite[Corollary 4.1]{Wazir}. 
\begin{thm}\label{Thm:STW2} Let $Y$ be a smooth Calabi-Yau threefold  elliptically fibered over a smooth variety $B$ with Mordell-Weil group of rank zero. Then,
\begin{equation}\nonumber
h^{1,1}(Y)=h^{1,1}(B)+f+1, \quad h^{2,1}(Y)=h^{1,1}(Y)-\frac{1}{2}\chi(Y),
\end{equation}
where $f$ is the number of geometrically irreducible fibral divisors not touching the zero section. In particular, if $Y$ is a $G$-model with $G$ being a semi-simple group, then $f$ is the rank of $G$. 
\end{thm}

\begin{table}[H]
\begin{center}
\begin{tabular}{|c|c|c|}
\hline
Models & $h^{1,1}(Y)$ & $h^{2,1}(Y)$ \\
\hline
$(\text{SU($2$)}\times \text{Sp($4$)})/\mathbb{Z}_2$  & $14-K^2$ & $17 K^2+16 K\cdot T+6 T^2+14$ \\
\hline
$\text{SU($2$)}\times \text{Sp($4$)}$ & $14-K^2$ & $29 K^2+15 K\cdot S+30 K\cdot T+3 S^2+8 S\cdot T+10 T^2+14$ \\
\hline
$(\text{SU($2$)}\times \text{SU($4$)})/\mathbb{Z}_2$  & $15-K^2$ & $17 K^2+18 K\cdot T+6 T^2+15$ \\
\hline
$\text{SU($2$)}\times \text{SU($4$)}$  & $15-K^2$ & $29 K^2+8 T\cdot (4 K+S)+15 K\cdot S+3 S^2+10 T^2+15$ \\
\hline
\end{tabular}
\label{hodge}
\caption{Hodge Numbers}
\end{center}
\end{table}

\subsection{Hyperplane arrangements}

Let  $\mathfrak{g}$ be a semi-simple Lie algebra and $ \mathbf{R}$ a representation of $\mathfrak{g}$.
The kernel of each  weight $\varpi$ of $\mathbf{R}$ defines a  hyperplane $\varpi^\perp$ through the origin of the Cartan sub-algebra of $\mathfrak{g}$.

\begin{defn}
The hyperplane arrangement I($\mathfrak{g},\mathbf{R}$) is defined inside the dual fundamental Weyl chamber of $\mathfrak{g}$, i.e. the  dual cone of the fundamental Weyl chamber of $\mathfrak{g}$, and its hyperplanes are the set of kernels of  the weights of $\mathbf{R}$. 
\end{defn}
For each $G$-model, we associate the hyperplane arrangement $\mathrm{I}(\mathfrak{g},  \mathbf{R})$ using the representation $\mathbf{R}$ induced by the weights of vertical rational curves produced by degenerations of the generic fiber over codimension-two points of the base. We then study the incidence structure of the hyperplane arrangement I$(\mathfrak{g}, \mathbf{R})$ \cite{EJJN1,EJJN2,G2,F4,Hayashi:2014kca}.

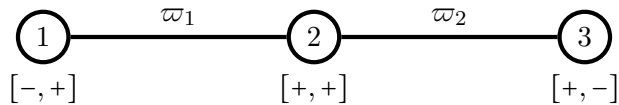
\begin{figure}[H]
\begin{center}
\begin{tikzpicture}[scale=1.2]
\node[draw,ultra thick,circle,label=below:{$[-,+]$}](1) at  (-3,0)   {$1$};
\node[draw,ultra thick,circle,label=below:{$[+,+]$}](2) at (0,0) {$2$};
\node[draw,ultra thick,circle,label=below:{$[+,-]$}](3) at (3,0) {$3$};
\draw[ultra thick] (1) --node[above] {$\varpi_1$} (2) --node[above] {$\varpi_2$} (3);
\end{tikzpicture}\\
\vspace{0.2cm}
\label{I2nsI4nsCham}
\caption{There are three chambers in I$_2^{\text{ns}}+$I$_4^{\text{ns}}$-model with a Mordell-Weil group $\mathbb{Z}_2$. Each chamber is noted as the signs of $[\varpi_1,\varpi_2]$ where $\varpi_1=\phi _1-\psi _1$ and $\varpi_2=\psi _1+\phi _1-\phi _2$. For chamber $1$, $\varpi_1<0 ,\ \varpi_2>0$; for chamber $2$, $\varpi_1>0 ,\ \varpi_2>0$; and for chamber $3$, $\varpi_1>0 ,\ \varpi_2<0$.}
\end{center}
\end{figure}

\begin{figure}[htb]
\begin{center}
\scalebox{1}{
\begin{tikzpicture}[scale=1.2]
\node[draw,ultra thick,circle,label=above:{}](1a) at  (-5.5, -6)   {$1a^-$};
\node[draw,ultra thick,circle,label=above:{}](2a) at (-5, -4) {$2a^-$};
\node[draw,ultra thick,circle,label=above:{}](3a) at (-5.5, -2) {$3a^-$};
\node[draw,ultra thick,circle,label=above:{}](4a) at (-5, 0) {$4a^-$};
\node[draw,ultra thick,circle,label=above:{}](5a) at (-5.5, 2) {$5a^-$};
\node[draw,ultra thick,circle,label=above:{}](1b) at (-2, -6) {$1b^-$};
\node[draw,ultra thick,circle,label=above:{}](2b) at (-2.5,- 4) {$2b^-$};
\node[draw,ultra thick,circle,label=above:{}](3b) at (-2, -2) {$3b^-$};
\node[draw,ultra thick,circle,label=above:{}](4b) at (-2.5, 0)  {$4b^-$};
\node[draw,ultra thick,circle,label=above:{}](5b) at (-2, 2) {$5b^-$};
\node[draw,ultra thick,circle,label=above:{}](1c) at (1, -6) {$1b^+$};
\node[draw,ultra thick,circle,label=above:{}](2c) at (1.5,- 4) {$2b^+$};
\node[draw,ultra thick,circle,label=above:{}](3c) at (1, -2) {$3b^+$};
\node[draw,ultra thick,circle,label=above:{}](4c) at (1.5,0) {$4b^+$};
\node[draw,ultra thick,circle,label=above:{}](5c) at (1, 2) {$5b^+$};
\node[draw,ultra thick,circle,label=above:{}](1d) at (4.5,- 6) {$1a^+$};
\node[draw,ultra thick,circle,label=above:{}](2d) at (4,-4) {$2a^+$}; 
\node[draw,ultra thick,circle,label=above:{}](3d) at (4.5, -2) {$3a^+$};
\node[draw,ultra thick,circle,label=above:{}](4d) at (4, 0) {$4a^+$};
\node[draw,ultra thick,circle,label=above:{}](5d) at (4.5,2) {$5a^+$};

\draw[ultra thick] (1b) --node[right] {$\varpi_7$} (2b) --node[right] {$\varpi_5$} (3b)--node[right] {$\varpi_9$}(4b)--node[right] {$\varpi_6$}(5b);
\draw[ultra thick] (1a) --node[left] {$\varpi_4$} (2a) --node[left] {$\varpi_5$} (3a)--node[left] {$\varpi_9$}(4a)--node[left] {$\varpi_6$}(5a);
\draw[ultra thick] (1c) --node[right] {$\varpi_5$} (2c) --node[right] {$\varpi_7$} (3c)--node[right] {$\varpi_6$}(4c)--node[right] {$\varpi_9$}(5c);
\draw[ultra thick] (1d) --node[right] {$\varpi_8$} (2d) --node[right] {$\varpi_7$} (3d)--node[right] {$\varpi_6$}(4d)--node[right] {$\varpi_9$}(5d);

\draw[ultra thick]  (1b) --node[above] {$\varpi_3$} (1c);
\draw[ultra thick] (2a) --node[above] {$\varpi_1$} (2b); 
\draw[ultra thick] (2c)--node[above] {$\varpi_2$}(2d);
\draw[ultra thick] (3a) --node[above] {$\varpi_1$} (3b) --node[above] {$\varpi_3$} (3c)--node[above] {$\varpi_2$}(3d);
\draw[ultra thick] (4a) -- node[above] {$\varpi_1$}(4b); \draw[ultra thick] (4c)--node[above] {$\varpi_2$}(4d);
\draw[ultra thick] (5a) --node[above] {$\varpi_1$} (5b) --node[above] {$\varpi_3$} (5c)--node[above] {$\varpi_2$}(5d);
\end{tikzpicture}}\\
\vspace{0.2cm}
\begin{align}\nonumber
\begin{split}
&\text{The weights corresponding to the 9 entries are} \ v=(\varpi_1, \varpi_2, \varpi_3, \varpi_4, \varpi_5, \varpi_6, \varpi_7, \varpi_8, \varpi_9), \\
&\text{where} \ \varpi_1=[0;-1,1,0], \varpi_2=[0;0,1,-1], \varpi_3=[0;-1,0,1], \varpi_4=[1;-1,1,0], \\
&\varpi_5=[1;0,-1,1], \varpi_6=[-1;1,0,0], \varpi_7=[-1;-1,1,0], \varpi_8=[-1;0,-1,1], \varpi_9=[-1;0,0,1] . 
\end{split}
\end{align}
\vspace{0.2cm}

\begin{tabular}{|c|c|c|c|}
\hline
 $5a^-\  (010110000)$ & $5b^-\  (110110000)$ & $5b^+\  (111110000)$ &$5a^+\  (101110000)$   \\
 $4a^-\  (010111000)$ & $4b^-\  (110111000)$& $4b^+\  (111110001)$ &$4a^+\  (101110001)$  \\
 $3a^-\  (010111001)$ & $3b^-\  (110111001)$ &$3b^+\  (111111001)$ &$3a^+\  (101111001)$  \\
 $2a^-\  (010101001)$& $2b^-\  (110101001)$& $2b^+\  (111111101)$& $2a^+\  (101111101)$  \\
 $1a^-\  (010001001)$ & $1b^-\  (110101101)$& $1b^+\  (111101101)$  &$1a^+\  (101111111)$   \\ 
 \hline
\end{tabular}
\vspace{0.3cm}
\end{center}
\caption{Chambers of the hyperplane arrangement I($A_1\oplus A_2,\mathbf{R}$) with $\mathbf{R}=(\bf{3},\bf{1})\oplus(\bf{1},\bf{15})\oplus(\mathbf{1},\mathbf{6})\oplus(\mathbf{2},\mathbf{4})\oplus(\mathbf{2},\mathbf{\bar{4}})\oplus(\mathbf{1},\mathbf{4})\oplus(\mathbf{1},\mathbf{\bar{4}})$.  Each circle corresponds to a chamber. The label on the edge connecting two chambers is the wall separating them.  
In the sign vector, an entry $s$  means a sign $(-1)^{s+1}$ for the corresponding form, that is, $s=0$ (resp. $s=1$) means that the corresponding linear form is negative (resp. positive). 
For example, the chamber $1a^-$ corresponds to $(010001001)$, which gives the sign vector $(-1,1,-1,-1,-1,1,-1,-1)$. }
\label{ChambersNoZ2}
\end{figure}
\clearpage

\begin{figure}[H]
\begin{center}
\begin{tikzpicture}[scale=1.2]
\node[draw,ultra thick,circle,label=above:{}](1a) at  (-5.5, -6)   {$1a^-$};
\node[draw,ultra thick,circle,label=above:{}](1b) at (-2, -6) {$1b^-$};
\node[draw,ultra thick,circle,label=above:{}](2b) at (-2.5,- 4) {$2ab^-$};
\node[draw,ultra thick,circle,label=above:{}](3b) at (-2, -2) {$3ab^-$};
\node[draw,ultra thick,circle,label=above:{}](4b) at (-2.5, 0)  {$4ab^-$};
\node[draw,ultra thick,circle,label=above:{}](5b) at (-2, 2) {$5ab^-$};
\node[draw,ultra thick,circle,label=above:{}](1c) at (1, -6) {$1b^+$};
\node[draw,ultra thick,circle,label=above:{}](2c) at (1.5,- 4) {$2ab^+$};
\node[draw,ultra thick,circle,label=above:{}](3c) at (1, -2) {$3ab^+$};
\node[draw,ultra thick,circle,label=above:{}](4c) at (1.5,0) {$4ab^+$};
\node[draw,ultra thick,circle,label=above:{}](5c) at (1, 2) {$5ab^+$};
\node[draw,ultra thick,circle,label=above:{}](1d) at (4.5,- 6) {$1a^+$};

\draw[ultra thick] (1b) --node[right] {$\varpi_7$} (2b) --node[right] {$\varpi_5$} (3b)--node[right] {$\varpi_9$}(4b)--node[right] {$\varpi_6$}(5b);
\draw[ultra thick] (1c) --node[right] {$\varpi_5$} (2c) --node[right] {$\varpi_7$} (3c)--node[right] {$\varpi_6$}(4c)--node[right] {$\varpi_9$}(5c);

\draw[ultra thick]  (1b) --node[above] {$\varpi_3$} (1c);
\draw[ultra thick] (1a) --node[above] {$\varpi_4$} (2b); 
\draw[ultra thick] (2c)--node[above] {$\varpi_8$} (1d);
\draw[ultra thick] (3b) --node[above] {$\varpi_3$} (3c);
\draw[ultra thick] (5b) --node[above] {$\varpi_3$} (5c);
\end{tikzpicture}\\
\vspace{0.2cm}
\begin{align}\nonumber
\begin{split}
&\text{The weights corresponding to the 7 entries are} \ v=(\varpi_3, \varpi_4, \varpi_5, \varpi_6, \varpi_7, \varpi_8, \varpi_9), \\
&\text{where} \ \varpi_3=[0;-1,0,1], \varpi_4=[1;-1,1,0], \varpi_5=[1;0,-1,1], \varpi_6=[-1;1,0,0], \\
&\varpi_7=[-1;-1,1,0], \varpi_8=[-1;0,-1,1], \varpi_9=[-1;0,0,1] . 
\end{split}
\end{align}
\vspace{0.2cm}

\begin{tabular}{|c|c|c|c|}
\hline
 & $5ab^-\  (0110000)$ & $5ab^+\  (1110000)$ &   \\
 & $4ab^-\  (0111000)$& $4ab^+\  (1110001)$ &  \\
 & $3ab^-\  (0111001)$ &$3ab^+\  (1111001)$ &  \\
 & $2ab^-\  (0101001)$& $2ab^+\  (1111101)$ &  \\
 $1a^-\  (0001001)$ & \ $1b^-\  (0101101)$& \ $1b^+\  (1101101)$  &$1a^+\  (1111111)$   \\
 \hline
\end{tabular}
\vspace{0.3cm}
\end{center}
\caption{Chambers of the hyperplane arrangement I($A_1\oplus A_2,\mathbf{R}$) with $\mathbf{R}=(\bf{3},\bf{1})\oplus(\bf{1},\bf{15})\oplus (\bf{2},\bf{4})\oplus (\bf{2},\bf{\bar{4}})\oplus(\bf{1},\bf{6})$.  Each circle corresponds to a chamber. The label on the edge connecting two chambers is the wall separating them.  
In the sign vector, an entry $s$  means a sign $(-1)^{s+1}$ for the corresponding form, that is, $s=0$ (resp. $s=1$) means that the corresponding linear form is negative (resp. positive). 
For example, the chamber $1a^-$ corresponds to $(010001001)$, which gives the sign vector $(-1,1,-1,-1,-1,1,-1,-1)$. }
\label{Chambers}
\end{figure}
\clearpage

\subsection{Triple intersection numbers}\label{sec:Triple}

In contrast to the Euler characteristic, the triple intersection  polynomial  do depend on the choice of a crepant resolution. Those presented here corresponds to the crepant resolutions given by the sequence of blowups listed in 
Table  \ref{Table.Blowups}. 
For each model, we compute the triple intersection numbers of the fibral divisors. 
We start with the Weierstrass models 
$\varphi:Y\to B$ listed in Table \ref{Table.Models} and consider the crepant resolution  $f:\widetilde{Y}\to Y$ induced by the sequence of blowups  in  Table \ref{Table.Blowups}.
The crepant resolution produces fibral divisors  $D_a^s$ and $D_a^t$ whose classes are  listed in Table \ref{tb:FibralDiv}. Here, the index $a$ runs through $\{1,2\}$ for $\text{Sp($4$)}$ and $\{1,2,3\}$ for $\text{SU($4$)}$. The class of the proper transform $\tilde{Y}$ of $Y$ is
\begin{align}
(\text{SU($2$)}\times \text{Sp($4$)})/\mathbb{Z}_2 , \ \text{SU($2$)}\times \text{Sp($4$)} &: \quad [\widetilde{Y}]=3H+6L-2E_1-2W_1-2W_2 , \\
(\text{SU($2$)}\times \text{SU($4$)})/\mathbb{Z}_2 , \ \text{SU($2$)}\times \text{SU($4$)} &: \quad [\widetilde{Y}]=3H+6L-2E_1-2W_1-W_2-W_3 .
\end{align}
The triple intersection numbers are then  given by 
\begin{equation}
\mathscr{F}_{trip}=\varphi_* f_* \left( \left( D_1^{\text{s}} \psi_1+\sum_a D_a^t \phi_a \right)^3 [\widetilde{Y}] \right).
\end{equation}
The pushforwards are computed using theorems \ref{Thm:Push} and \ref{Thm:PushH}. 
We specialize to the  Calabi-Yau case by imposing the condition $L=-K$ which  ensures that the canonical class of $Y$ is trivial. 
The triple intersection polynomial of the Calabi-Yau threefold obtained by the resolutions listed  in Table   \ref{Table.Blowups} are 
\begin{align}
\mathscr{F}^{(\text{SU($2$)}\times \text{Sp($4$)})/\mathbb{Z}_2}_{trip}=&-8 (T+3K) (T+2K) \psi _1^3 -8 T^2 \phi _1^3-4 T(K+T) \phi _2^3 -12KT \phi _1^2 \phi _2 \nonumber \\
&+6 T (T+2K) \phi _1\phi _2^2 +12 T(T+2K) \psi_1 (2\phi _1^2 -2\phi _2 \phi _1 +\phi _2^2 ) .
\end{align}
\begin{align}
\mathscr{F}^{\text{SU($2$)}\times \text{Sp($4$)}}_{trip}=& -2 S (S-2K)\psi _1^3-8T^2 \phi _1^3 +2T(2K+S)\phi _2^3  \\
&+6T(S+2 T+2K)\phi _1^2 \phi _2-6T(2K+S+T) \phi _1\phi _2^2 -6 S T \psi _1 \left(2 \phi _1^2-2 \phi _2 \phi _1+\phi _2^2\right).\nonumber
\end{align}
\begin{align}
\mathscr{F}^{(\text{SU($2$)}\times \text{SU($4$)})/\mathbb{Z}_2}_{trip}=&-8 \left(6K^2+5KT+T^2\right) \psi _1^3 -4T(K+T)\left(\phi _1^3+\phi _2^3\right) -2T(K+2 T) \phi _3^3 +6KT \phi _1^2 \phi _3 \nonumber \\
&-6KT \phi _1 \phi _2 \phi _3 +3T(T+2K)\phi _2^2 \left(\phi _1+\phi _3\right)-3KT\phi _2 \left(\phi _1^2+\phi _3^2\right) \\
&+12T(T+2K) \psi _1 \left(\phi _1^2-\phi _2 \phi _1+\phi _2^2+\phi _3^2-\phi _2 \phi _3\right) \nonumber.
\end{align}
\begin{align}
\mathscr{F}^{\text{SU($2$)}\times \text{SU($4$)}}_{trip}=&-2S(S-2K)\psi _1^3 -4T(K+T)\phi _1^3 +2T(S+2K)\phi _2^3 -2T(K+2T)\phi _3^3 +6KT \phi _1^2 \phi _3 \nonumber \\
&-6KT \phi _1 \phi _2 \phi _3 -3T(2K+S+T) \phi _2^2 \left(\phi _1+\phi _3\right)+3T(3K+S+2T) \phi _2 \left(\phi _1^2+\phi _3^2\right)\nonumber \\
&-6 S T \psi _1 \left(\phi _1^2-\phi _2 \phi _1+\phi _2^2+\phi _3^2-\phi _2 \phi _3\right) .
\end{align}

\subsection{The prepotential of the five-dimensional theories}
Following Intrilligator, Morrison, and Seiberg \cite{IMS}, we compute the quantum contribution to the prepotential of a five-dimensional gauge theory ($6\mathscr{F}_{IMS}$) with the matter fields in the representations $\mathbf{R}_i$ of the gauge group. Let $\phi$ be in the Cartan subalgebra of a Lie algebra $\mathfrak{g}$.  We denote by $\phi=\{\psi_1,\phi_1,\phi_2\}$ for the cases with $\frak{g}=\frak{su}(2)+\frak{sp}(4)$, and we denote by $\phi=\{\psi_1,\phi_1,\phi_2,\phi_3\}$ for the cases with $\frak{g}=\frak{su}(2)+\frak{su}(4)$. The  weights are in the dual space of the Cartan subalgebra.  We denote the evaluation of a  weight on $\phi$ as a scalar product $\langle \mu,\phi \rangle$.  We recall that the roots are the weights of the adjoint representation of  $\mathfrak{g}$.
Denoting the fundamental roots by $\alpha$ and the weights of $\mathbf{R}_i$ by $\varpi$ we have 
\begin{align}
\mathscr{F}_{\text{IMS}} =&\frac{1}{12} \left(
{
\sum_{\alpha} |\langle \alpha, \phi \rangle|^3-\sum_{\mathbf{R}_i} \sum_{\varpi\in W_i} n_{\mathbf{R}_i} |\langle \varpi, \phi\rangle|^3 
}
\right).
\end{align}
The representations $\bf{R}$ for each group are determined geometrically by using the splittings of the curves. The prepotential is computed in a particular chamber of the five-dimensional theory that matches with the crepant resolution in which the triple intersection polynomial is computed.

For the case of $G=(\text{SU($2$)}\times \text{Sp($4$)})/\mathbb{Z}_2$, we first determine that matching chamber is given by
\begin{equation}\nonumber
\text{Chamber} \ [-,+] : 2\phi _2>2\phi _1>\phi _2>0\land \psi _1>\phi _1,
\end{equation}
which is the left chamber in Figure \ref{I2nsI4nsCham}. The prepotential in this chamber $[-,+]$ is given by
\begin{align}
\begin{split}
6\mathscr{F}_{\text{IMS}}=&-4 (n_{\bf{2},\bf{4}}+2 n_{\bf{3},\bf{1}}-2) \psi _1^3 -8 (n_{\bf{1},\bf{10}}-1) \phi _1^3+(-8 n_{\bf{1},\bf{10}}-n_{\bf{1},\bf{5}}+8)\phi _2^3 \\
& -3 \phi _1^2 \phi _2 (4 n_{\bf{1},\bf{10}}+n_{\bf{1},\bf{5}}-4) +3(6 n_{\bf{1},\bf{10}}+n_{\bf{1},\bf{5}}-6) \phi _1\phi _2^2 \\
&+\psi _1 \left(-12 n_{\bf{2},\bf{4}} \phi _1^2+12 n_{\bf{2},\bf{4}} \phi _2 \phi _1-6 n_{\bf{2},\bf{4}} \phi _2^2\right).
\end{split}
\end{align}
For the case of $G=\text{SU($2$)}\times \text{Sp($4$)}$, we find the matching chamber to be
\begin{equation}\nonumber
\text{Chamber} \ [-,+]: 2\phi _2>2\phi _1>\phi _2>0\land \psi _1>\phi _1,
\end{equation}
which is the very same chamber with the one above with a trivial Mordell-Weil group. Due to its different representations, the prepotential is given by
\begin{align}
\begin{split}
6 \mathscr{F}_{\text{IMS}}=& - (n_{\bf{2},\bf{1}}-4n_{\bf{2},\bf{4}}+8 n_{\bf{3},\bf{1}}-8) \psi _1^3 -8 (n_{\bf{1},\bf{10}}+n_{\bf{1},\bf{5}}-1) \phi _1^3-(8 n_{\bf{1},\bf{10}}+n_{\bf{1},\bf{4}}-8)\phi _2^3 \\
& -3 \phi _1^2 \phi _2 (4 n_{\bf{1},\bf{10}}+n_{\bf{1},\bf{4}}-4 n_{\bf{1},\bf{5}}-4) +3(6 n_{\bf{1},\bf{10}}+n_{\bf{1},\bf{4}}-2n_{\bf{1},\bf{5}}-6) \phi _1\phi _2^2 \\
&+\psi _1 \left(-12 n_{\bf{2},\bf{4}} \phi _1^2+12 n_{\bf{2},\bf{4}} \phi _1\phi _2 -6 n_{\bf{2},\bf{4}} \phi _2^2\right) .
\end{split}
\end{align}
For the case of $G=(\text{SU($2$)}\times \text{SU($4$)})/\mathbb{Z}_2$, we find the matching chamber to be
\begin{equation}
\text{Chamber} \ 5ab+ : \phi _1>\phi _2>0\land \psi _1>\phi _1\land \psi _1>\phi _3>\phi _2-\psi _1 ,
\end{equation}
which is a chamber on the top right of Figure \ref{Chambers}. The prepotential in this chamber is given by
\begin{align}
\begin{split}
6\mathscr{F}_{\text{IMS}}=& \  -4 (n_{\bf{2},\bf{4}}+n_{\bf{2},\bf{\bar{4}}}+2 n_{\bf{3},\bf{1}}-2) \psi _1^3 -8 (n_{\bf{1},\bf{15}}-1) \phi _1^3-(8 n_{\bf{1},\bf{15}}-8)\phi _2^3  \\
&-2 (4 n_{\bf{1},\bf{15}}+n_{\bf{1},\bf{6}}-4) \phi _3^3 -6 n_{\bf{1},\bf{6}} \phi _1^2\phi _3 +6 n_{\bf{1},\bf{6}} \phi _1\phi _2 \phi _3 +\frac{3}{2} (4 n_{\bf{1},\bf{15}}-2 n_{\bf{1},\bf{6}}-4) \phi _2^2 (\phi _1+\phi _3) \\
& -\frac{3}{2} (-2 n_{\bf{1},\bf{6}}) \phi _2 (\phi _1^2 + \phi _3^2) -6 (n_{\bf{2},\bf{4}}+n_{\bf{2},\bf{\bar{4}}}) \psi _1 \left(\phi _1^2 -\phi _1 \phi _2 +\phi _2^2 +\phi _3^2 -\phi _2 \phi _3 \right) .
\end{split}
\end{align}
For the case of $G=\text{SU($2$)}\times \text{SU($4$)}$, we find the matching chamber to be
\begin{equation}
\text{Chamber} \ 5b+ : 2 \phi _1>\phi _2>0 \land \phi _2>\phi _3>\phi _1 \land \psi _1>\phi _3,
\end{equation}
which is represented in Figure \ref{ChambersNoZ2}. The prepotential in this chamber is given by
\begin{align}
\begin{split}
6 \mathscr{F}_{\text{IMS}}=& \ -(n_{\bf{3},\bf{1}}+4 (n_{\bf{2},\bf{4}}+n_{\bf{2},\bf{\bar{4}}}+2 n_{\bf{3},\bf{1}}-2)) \psi _1^3 -8 (n_{\bf{1},\bf{15}}-1) \phi _1^3 \\
& -(8 n_{\bf{1},\bf{15}}+n_{\bf{1},\bf{4}}+n_{\bf{1},\bf{\bar{4}}}-8)\phi _2^3 -2 (4 n_{\bf{1},\bf{15}}+n_{\bf{1},\bf{6}}-4) \phi _3^3  -6 n_{\bf{1},\bf{6}} \phi _1^2\phi _3 +6 n_{\bf{1},\bf{6}} \phi _1\phi _2 \phi _3 \\
& +\frac{3}{2} (4 n_{\bf{1},\bf{15}}+n_{\bf{1},\bf{4}}+n_{\bf{1},\bf{\bar{4}}}-2 n_{\bf{1},\bf{6}}-4) \phi _2^2 (\phi _1+\phi _3) -\frac{3}{2} (n_{\bf{1},\bf{4}}+n_{\bf{1},\bf{\bar{4}}}-2 n_{\bf{1},\bf{6}}) \phi _2 (\phi _1^2 + \phi _3^2) \\
&-6 (n_{\bf{2},\bf{4}}+n_{\bf{2},\bf{\bar{4}}}) \psi _1 \left(\phi _1^2 -\phi _1 \phi _2 +\phi _2^2 +\phi _3^2 -\phi _2 \phi _3 \right) .
\end{split}
\end{align}

\subsection{Number of charged hypermultiplets.}

The number of charged hypermultiplets under each representation is obtained by comparing the triple intersection numbers and the one-loop prepotential:
\begin{equation}
\mathscr{F}_{trip}=6\mathscr{F}_{IMS}.
\end{equation}
The comparison is enough to completely determine the number $n_{\mathbf{R}_i}$ for the models with a $\mathbb{Z}_2$ Mordell-Weil group, that is, 
$(\text{SU($2$)}\times \text{Sp($4$)})/\mathbb{Z}_2$ and $(\text{SU($2$)}\times \text{SU($4$)})/\mathbb{Z}_2$.  
We see that the introduction of a $\mathbb{Z}_2$ Mordell-Weil group removes the fundamental representation, but does not affect the (traceless) antisymmetric representation, the 
adjoint, or bifundamental matters since they are self-dual representations. 
However, for the models $\text{SU($2$)}\times \text{Sp($4$)}$ and 
$\text{SU($2$)}\times \text{SU($4$)}$, comparing the triple intersection numbers  and the one-loop prepotential is  not enough to fix all the multiplicities and we are left with some linear relations between the number of representations. 
This is because without the $\mathbb{Z}_2$, we get additional matter content but the non-zero triple intersection numbers are unchanged. 
The remaining linear relations can be solved in many ways. For example, we can use Witten's genus formula to count the number of adjoint matters as the genus of the curve supporting the gauge group \cite{Witten}. 
Another possibility is to direct count the number of bi-fundamental representations  as intersection numbers between the divisors $S$ and $T$. 
We can also use the vanishing of anomalies in the six dimensional uplift to fix the remaining linear equations. For example,  the gravitational anomaly or the vanishing of the terms $\tr F^4_a$ are enough in addition to the  linear relations left from the triple intersection numbers.  The result is spelled out in Table \ref{Table.Matter}.

\subsection{Cancellation of anomalies in the uplifted six-dimensional theories} 
 
The anomalies in six-dimensional theories can be formally summarized by an eight-form I$_8$ constructed from the Riemann curvature two-form $R$ and the gauge curvature two-forms $F_i$ 
(see \cite{Green:1984bx,Sagnotti:1992qw} and \cite{Schwarz:1995zw,Erler:1993zy,Avramis:2005hc}).  
The Green-Schwarz mechanism consists of adding a counterterm depending on the Yang-Mills and gravitational Chern-Simons three-forms and modifying the gauge transformation of the appropriate antisymmetric self-dual or anti-self-dual two-forms  simultaneously \cite{Schwarz:1995zw,Sagnotti:1992qw}. 
This requires I$_8$ to factorize into a product of two four-forms. 
We check these conditions in detail and factorize the anomaly polynomial I$_8$ explicitly for each cases.

Before factoring the anomaly polynomial I$_8$, a necessary condition is the vanishing of the coefficients of the terms $\tr R^4$ and $\tr F_a^4$ from the quartic contribution of the pure gravitational and the pure gauge anomalies. 
The condition on the pure gravitational anomaly requires knowing the Euler characteristic of the elliptic fibration due to the content of hypermultiplets.  Due to the recent results on the pushforward of blowups  \cite{Euler}, we can now easily compute such invariants using a crepant resolution of singularities with centers that are local complete intersections. In particular, we get results that are independent of the dimension of the base and provide the results for all $n$-folds. 

Using the number of multiplets, we show that  the pure gravitational and the pure gauge anomalies are canceled using the number of representations $n_{\mathbf{R}_i}$, which are restricted by matching triple intersection numbers with the one-loop quantum correction to the cubic prepotential  (see Section  \ref{sec:anomaly}).  Assuming that the coefficients of $\tr R^4$ and $\tr F_a^4$ vanish, for the semi-simple gauge group with two simple components $G=G_1\times G_2$, the anomaly polynomial is 
$$
I_8 =\frac{K^2}{8} (\tr R^2)^2 +\frac{1}{6} (X^{(2)}_{1} +X^{(2)}_{2}) \tr R^2-\frac{2}{3} (X^{(4)}_{1}+X^{(4)}_{2})+4Y_{12},
$$
which we prove to reduce to a perfect square for all case considered, as expected from Sadov's work \cite{Sadov:1996zm}. More precisely,
\begin{align}\nonumber
I_8= \frac{1}{2} \Big( \frac{K}{2}\tr R^2-2 S \tr_{\mathbf{2}} F_1^2 -2 T  \tr_{\mathbf{4}} F_2^2\Big)^2.
\end{align}
The term that is squared is used  as a magnetic source  for the antisymmetric two-form in the gravitational multiplet that cancel the anomaly in the Green-Schwarz-Sognatti mechanism.

\section{$(\text{SU($2$)}\times \text{Sp($4$)})/\mathbb{Z}_2$-model} \label{sec:nonsplitZ2}

The fiber geometry of the collision of I$_2 ^{\text{ns}}$ and I$_4 ^{\text{ns}}$ is described in detail. The Weierstrass equation of I$_2 ^{\text{ns}}+$I$_4 ^{\text{ns}}$ is given by
\begin{equation}
y^2z=x^3+a_2x^2z+st^2xz^2 ,
\end{equation}
where $S=V(s)$ is the divisor supporting I$_2^{\text{ns}}$ and $T=V(t)$ is the divisor supporting I$_4^{\text{ns}}$. The discriminant of this model is
\begin{equation}
\Delta=16 s^2 t^4 (a_2^2-4 s t^2).
\end{equation}
The corresponding simply connected group $G$ and the representation $\mathbf{R}$, which is computed geometrically in the next section, are 
\begin{equation}
G=(\text{SU($2$)}\times \text{Sp($4$)})/\mathbb{Z}_2, \quad \bf{R}=(\bf{3},\bf{1})\oplus (\bf{1},\bf{10})\oplus (\bf{2},\bf{4})\oplus (\bf{1},{5}).
\end{equation}
The following sequence of blowups gives a crepant resolution of the elliptic fibration:
\begin{equation}
  \begin{tikzcd}[column sep=huge] 
  X_0  \arrow[leftarrow]{r} {(x,y,s|e_1)} & \arrow[leftarrow]{r} {(x,y,t|w_1)}  X_1 &X_2  \arrow[leftarrow]{r}{(x,y,w_1|w_2)} &X_3.
  \end{tikzcd}
\end{equation}
The proper transform is 
\begin{equation}
y^2=e_1w_1w_2^2x^3+a_2x^2+st^2w_1x,
\end{equation}
and the relative projective ``coordinates'' are
\begin{equation}
[e_1w_1w_2^2x : e_1w_1w_2^2y : z=1][w_1w_2^2x : w_1w_2^2y : s][w_2x : w_2y : t][x : y : w_1].
\end{equation}
To show that we have  a resolution of singularities, it is enough to assume that $V(a_2)$, $V(s)$, and $V(t)$ are smooth divisors intersecting transversally. 
In particular, working in patches using $(x,y,a_2)$ as a part of the local coordinates, the absence of singularities follows from the Jacobian criterion. 
From applying the adjunction formula after each blowup, we conclude that the resolution is crepant. 
\subsection{Fiber structure and representations}
We denote by $D_i^{\text{s}}$ ($i=0,1$), $D_j^t$ ($j=0,1,2$) the fibral divisors; by $C_i^{\text{s}}$ ($i=0,1$) and $C_j^t$ ($j=0,1,2$) the generic fibers of D$_i^{\text{s}}$ over $S$ and D$_j^t$ over $T$, respectively. The fibral divisors for this model are
\begin{align}
\begin{cases}
 D_0^{\text{s}} &: s=y^2-e_1w_1w_2^2x^3-a_2x^2=0 , \\
 D_1^{\text{s}} &: e_1=y^2-a_2x^2-st^2w_1x=0 , \\
 D_0^t &: t=y^2-e_1w_1w_2^2x^3-a_2x^2=0 , \\
 D_1^t &: w_1=y^2-a_2x^2=0 , \text{ (two roots $C_1=C_{1+}^t+C_{1-}^t$),}\\
 D_2^t &: w_2=y^2-a_2x^2-st^2w_1x=0 .
\end{cases}
\end{align}
The only components that touch the generator of $\mathbb{Z}_2$ are $C_1^{\text{s}}$ and $C_2^t$. The only sections that touch the zero section are $D_0^{\text{s}}$ and $D_0^t$.

Over $S=V(s)$, we have a generic fiber of type I$_2^{\text{ns}}$ with two geometric components $C_0^{\text{s}}$ and $C_1^{\text{s}}$. The fiber I$_2^{\text{ns}}$ specializes to a fiber of  type III over $V(a_2)$. Over $T=V(t)$, on the other hand, we have a generic fiber of type I$_4^{\text{ns}}$, whose geometric components are $C_0^t$, $C_{1+}^t$, $C_{1-}^t$, and $C_2^t$. The fiber I$_4^{\text{ns}}$ further enhances over $V(a_2)$, where two non-split curves $C_{1\pm}^t$ degenerate, which is represented on the right side of Figure \ref{I2nsI4ns}.

At the collision of $S$ and $T$, we produce the following curves:
\begin{align}
\begin{cases}
C_0^{\text{s}}\cap C_0^t &: s=t=y^2-e_1w_1w_2^2x^3-a_2x^2=0 \rightarrow \eta^{00} , \\
C_1^{\text{s}}\cap C_0^t &: e_1=t=y^2-a_2x^2=0 \rightarrow \eta^{10 \pm} , \text{ (two roots for each curve,)}\\
C_1^{\text{s}}\cap C_1^t &: e_1=w_1=y^2-a_2x^2=0 \rightarrow \eta^{11 \pm} , \text{ (two roots for each curve,)}\\
C_1^{\text{s}}\cap C_2^t &: e_1=w_2=y^2-a_2x^2-st^2w_1x=0 \rightarrow \eta^{12} .
\end{cases}
\end{align}
The fiber structure is represented in  Figure \ref{I2nsI4ns}. As expected, the collision of the divisors of the two fibers (type I$_2^{\text{ns}}$ and I$_4^{\text{ns}}$) is naturally enhanced into an I$_6^{\text{ns}}$.

The fibers of the collisions can be described from the splitting of the curves $C_i^{\text{s}}$ ($i=0,1$) and  $C_i^t$ ($i=0,1,2$) from I$_2^{\text{ns}}$ and  I$_4^{\text{ns}}$, respectively:
\begin{align}
\begin{cases}
C_0^{\text{s}}  & \rightarrow \eta^{00} ,\quad C_1^{\text{s}}  \rightarrow \eta^{10+}+\eta^{10-}+\eta^{11+}+\eta^{11-}+\eta^{12} ,\\
C_0^t  & \rightarrow \eta^{00}+\eta^{10+}+\eta^{10-} ,\quad C_1^t \rightarrow \eta^{11+}+\eta^{11-} ,\quad C_2^t  \rightarrow \eta^{12} .
\end{cases}
\end{align}
From these splittings of the curves, we compute the intersection numbers between the curves and the fibral divisors of I$_2^{\text{ns}}$ and I$_4^{\text{ns}}$  on the collision to be

\begin{equation}
\begin{tabular}{c|ccccc}
 & $D_0^{\text{s}}$ & $D_1^{\text{s}}$ & $D_0^t$ & $D_1^t$ & $D_2^t$ \\
 \hline
$\eta^{00}$ & $-2$ & 2 & 0 & 0 & 0 \\
$\eta^{10+}+\eta^{10-}$ & 2 & $-2$ & $-2$ & 2 & 0 \\
$\eta^{10\pm}$ & 1 & $-1$ & $-1$ & 1 & 0 \\
$\eta^{11+}+\eta^{11-}$ & 0 & 0 & 2 & $-4$ & 2 \\
$\eta^{11\pm}$ & 0 & 0 & 1 & $-2$ & 1 \\
$\eta^{12}$ & 0 & 0 & 0 & 2 & $-2$ 
\end{tabular}
\quad \raisebox{-40pt}{.} 
\end{equation}
The physical weight are minus the intersection numbers. We recall that the curve $\eta^{10\pm}$ carries the weight $[1]$ on the $\mathfrak{su}(2)$ side and the weight $[-1,0]$ on the $\mathfrak{sp}(4)$ side. This gives $[1;-1,0]$, which yields the representation $(\bf{2},\bf{4})$. These non-split curves join together to produce $\eta^{10+}+\eta^{10-}$ with weight $[2;-2,0]$, and the corresponding representation is $(\bf{3},\bf{10})$. Hence, the representation for the I$_2^{\text{ns}}+$I$_4^{\text{ns}}$-model with Mordell-Weil group $\mathbb{Z}_2$ is $
\bf{R}=(\bf{3},\bf{1})\oplus (\bf{1},\bf{10})\oplus (\bf{3},\bf{10})\oplus (\bf{2},\bf{4})\oplus (\bf{1},{5}).
$
We note that for the case of threefolds, the curves $\eta^{10\pm}$ are always split since all curve can split over a codimension-two point. Hence, the bi-adjoint $(\bf{3},\bf{10})$ does not show up geometrically. Hence for the Calabi-Yau threefolds, the representation is then
\begin{equation}
\bf{R}=(\bf{3},\bf{1})\oplus (\bf{1},\bf{10})\oplus (\bf{2},\bf{4})\oplus (\bf{1},{5}).
\end{equation}
The only group that is consistent with this representation $R$ is
\begin{equation}
G=(\text{SU($2$)}\times \text{Sp($4$)})/\mathbb{Z}_2,
\end{equation}
where $\mathbb{Z}_2$ is minus the identity.

The representations with respect to $(\frak{su}(2),\frak{sp}(4))$ from this I$_2^{\text{ns}}+$I$_4^{\text{ns}}$-model with the Mordell-Weil group $\mathbb{Z}_2$ are summarized in Table \ref{Rep.NonsplitModel} below. Here we denoted weights as $[\psi ;\varphi_1,\varphi_2]$ where $[\psi]$ is the weight for the $\frak{su}(2)$ and $[\varphi_1,\varphi_2]$ is the weight for the $\frak{sp}(4)$.
\begin{table}[H]
\begin{center}
\begin{tabular}{|c|c|c|c|}
\hline
Locus & {$tw_1w_2=0$} &  \multicolumn{2}{c|}{$se_1=tw_1w_2=0$} \\
\hline
Curves & $C_1^{t'}$ & $\eta^{10\pm}$ & $\eta^{10+}+\eta^{10-}$ \\
\hline
Weights &  $[0;2,-1]$ & $[1;-1,0]$ & $[2;-2,0]$ \\
\hline
Representations & $(\bf{1},\bf{5})$ & $(\bf{2},\bf{4})$ & $(\bf{3},\bf{10})$ \\
\hline
\end{tabular}
\end{center}
\caption{Weights and representations for the I$_2^{\text{ns}}+$I$_4^{\text{ns}}$-model with a Mordell-Weil group $\mathbb{Z}_2$}
\label{Rep.NonsplitModel}
\end{table}
\begin{figure}[H]
\begin{center}
\includegraphics[scale=.73]{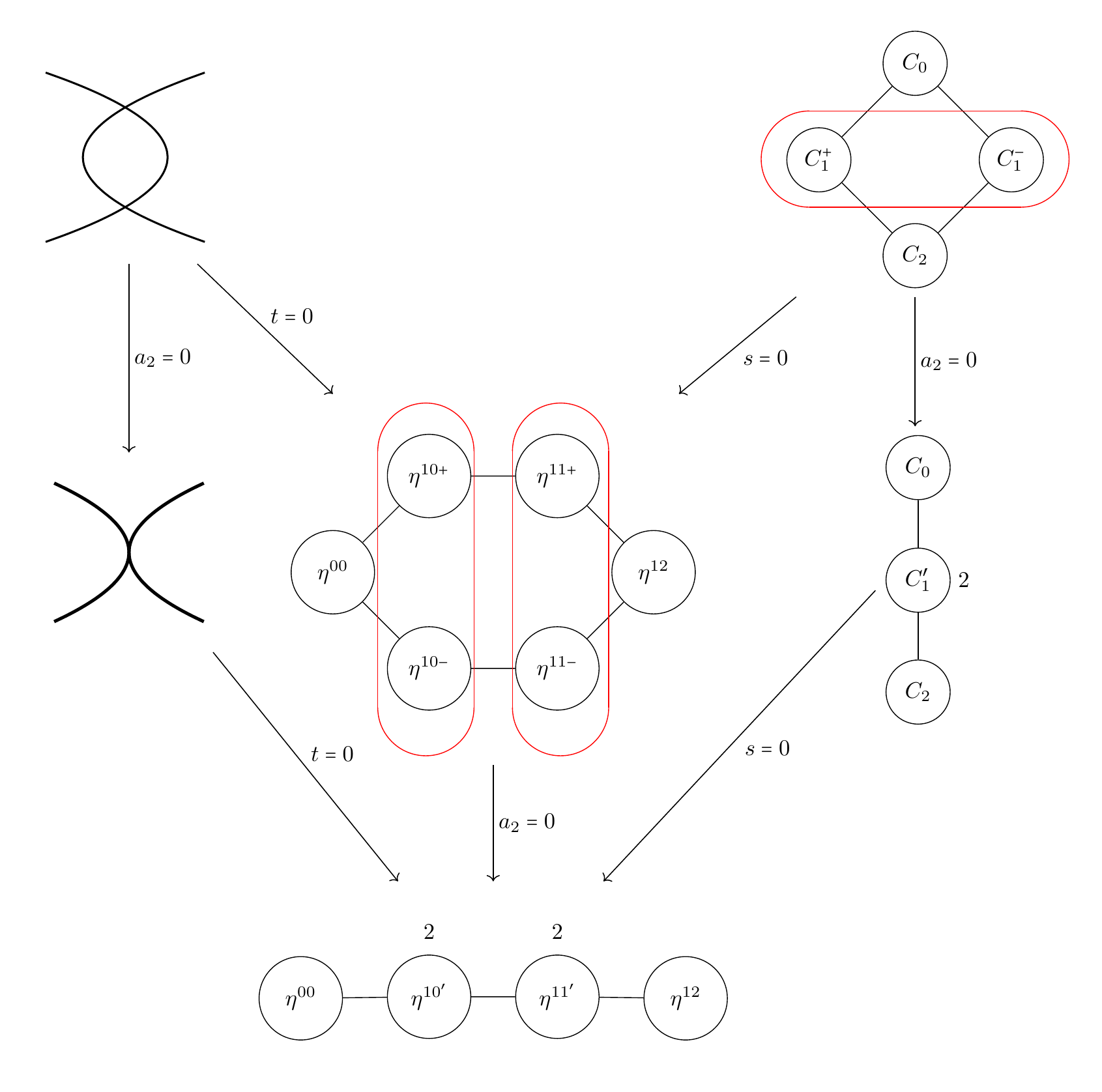}
\caption{Fiber structure of I$_2^{\text{ns}}+$I$_4^{\text{ns}}$ with a Mordell-Weil group $\mathbb{Z}_2$. The diagrams on the top are the Kodaira fibers of type I$_2^{\text{ns}}$ and I$_4^{\text{ns}}$. When $a_2=0$, I$_2^{\text{ns}}$ specializes to III, as seen in the middle left diagram, and I$_4^{\text{ns}}$ specializes to the diagram on the middle right. When I$_2^{\text{ns}}$ and I$_4^{\text{ns}}$ collide on the locus $se_1=tw_1w_2w_3=0$, we get the fiber structure of the collision of I$_2^{\text{ns}}$ and I$_4^{\text{s}}$ as the hexagon drawn in the middle. This enhancement has two newly split curves $\eta^{10+}$ and $\eta^{10-}$ that are of non-split type. When $a_2=0$, this hexagon further specializes to the diagram on the bottom.} 
\label{I2nsI4ns} 
\end{center}
\end{figure}


\subsection{Coulomb phases}
In this section, we show that the I$_2^{\text{ns}}+$I$_4^{\text{ns}}$-model with a Mordell-Weil group $\mathbb{Z}_2$ has three chambers. Denote $\frak{su}(2)$ by $[\psi]$  and $\frak{sp}(4)$ by $[\varphi_1,\varphi_2,\varphi_3]$. Then the Weyl chamber of the I$_2^{\text{ns}}+$I$_4^{\text{ns}}$-model with $\mathbb{Z}_2$ is defined with three hyperplanes given by
\begin{equation}
\psi_1>0 , \quad \phi _2-\phi _1>0, \quad 2 \phi _1-\phi _2>0.
\label{NonsplitRel}
\end{equation}
In addition, the subchambers are defined  by  the weights of the representation $(\bf{2},\bf{4})$,
\begin{equation}
\varpi_1=\phi _1-\psi _1, \quad \varpi_2=\psi _1+\phi _1-\phi _2,
\end{equation}
where the summation of these is positive from the first two hyperplanes in equation \eqref{NonsplitRel}:
\begin{equation}
\varpi_1+\varpi_2=(2\phi_1-\phi_2)+\psi_1>0.
\end{equation}
Hence, $\varpi_1$ and $\varpi_2$ cannot be both negative. It follows there are a total of three chambers, which are denoted by the signs of $[\varpi_1,\varpi_2]$:
\begin{align}
[-,+] &: \phi _2>0\land \frac{\phi _2}{2}<\phi _1<\phi _2\land \psi _1>\phi _1, \\
[+,+] &: \phi _2>0\land \frac{\phi _2}{2}<\phi _1<\phi _2\land \phi _2-\phi _1<\psi _1<\phi _1, \\
[+,-] &: \phi _1>0\land \phi _1<\phi _2<2 \phi _1\land 0<\psi _1<\phi _2-\phi _1.
\end{align}
The chambers are flop related by $\varpi_1$ and $\varpi_2$, as shown in Figure \ref{I2nsI4nsCham}.

\subsection{$5d$  $\mathcal{N}=1$ prepotentials and the triple intersection polynomials}
The triple intersection polynomial is computed for the I$_2^{\text{ns}}+$I$_4^{\text{ns}}$-model with the Mordell-Weil group $\mathbb{Z}_2$ in the crepant resolution:
\begin{align}
\begin{split}
\mathscr{F}_{trip}=&-8 (T-3 L) (T-2 L) \psi _1^3 -8 T^2 \phi _1^3+4 T(L-T) \phi _2^3+12 L T \phi _1^2 \phi _2 +6 T (T-2 L) \phi _1\phi _2^2  \\
&+\psi_1(24 T(T-2 L) \phi _1^2 -24 T(T-2 L) \phi _2 \phi _1 +12 T(T-2 L) \phi _2^2 ) \\
&-8 \psi _0^3 (3 L-2 T) (2 L-T)+6 T (2 L-T) \phi _0^2 \left(\phi _1-2 \psi _1\right) -4 L T \phi _0^3 \\
&+2 \phi _0 \left(6 T \psi _1^2 (T-2 L)+12 T \psi _1 \phi _1 (2 L-T)+6 T \phi _1^2 (T-L)\right)\\
&+\psi _0^2 \left(2 \left(8 \psi _1 (3 L-2 T) (2 L-T)-4 \psi _1 (2 L-T) (3 L-T)\right)+12 T \phi _0 (T-2 L)\right) \\
&+\psi _0 \left(2 \left(8 \psi _1^2 (2 L-T) (3 L-T)-4 \psi _1^2 (3 L-2 T) (2 L-T)\right)-24 T \psi _1 \phi _0 (T-2 L)\right) .
\end{split}
\end{align}
For Calabi-Yau threefolds, the representation for this I$_2^{\text{ns}}+$I$_4^{\text{ns}}$-model is geometrically computed as
\begin{align}
\bf{R}=(\bf{3},\bf{1})\oplus(\bf{1},\bf{10})\oplus(\bf{2},\bf{4})\oplus(\bf{1},\bf{5}).
\end{align}
Using these representations, the $5d$ prepotential in the  chamber $[-,+]$ is
\begin{align}
\begin{split}
6\mathscr{F}_{\text{IMS}}=&-4 (n_{\bf{2},\bf{4}}+2 n_{\bf{3},\bf{1}}-2) \psi _1^3 -8 (n_{\bf{1},\bf{10}}-1) \phi _1^3+(-8 n_{\bf{1},\bf{10}}-n_{\bf{1},\bf{5}}+8)\phi _2^3 \\
& -3 \phi _1^2 \phi _2 (4 n_{\bf{1},\bf{10}}+n_{\bf{1},\bf{5}}-4) +3(6 n_{\bf{1},\bf{10}}+n_{\bf{1},\bf{5}}-6) \phi _1\phi _2^2 \\
&+\psi _1 \left(-12 n_{\bf{2},\bf{4}} \phi _1^2+12 n_{\bf{2},\bf{4}} \phi _2 \phi _1-6 n_{\bf{2},\bf{4}} \phi _2^2\right).
\end{split}
\end{align}
Using the triple intersection polynomials that are independent from $\psi_0$ and $\phi_0$ to match the prepotential, the numbers of representations $n_R$ are computed to be
\begin{align}
\begin{split}
&n_{\bf{3},\bf{1}}=6 L^2-7 L T+2 T^2+1=g_S, \quad n_{\bf{2},\bf{4}}=-2 T (T-2 L)=2 (-4 g_T+T^2+4), \\
&n_{\bf{1},\bf{5}}=\frac{1}{2} \left(L T+T^2\right)=-g_T+T^2+1, \quad n_{\bf{1},\bf{10}}=\frac{1}{2} \left(-L T+T^2+2\right)=g_T .
\end{split}
\end{align}

\subsection{$6d$ $\mathcal{N}=(1,0)$ anomaly cancellation}
In this section, we consider an I$_2^{\text{ns}}+$I$_4^{\text{ns}}$-model with the Mordell-Weil group $\mathbb{Z}_2$. Then, the gauge algebra is given by
\begin{equation}
\frak{g}=A_1+C_2 ,
\end{equation}
and the representation is geometrically computed in section \ref{sec:nonsplitZ2} to be
\begin{equation}
\bf{R}=(\bf{3},\bf{1})\oplus (\bf{1},\bf{10})\oplus (\bf{2},\bf{4})\oplus (\bf{1},{5}).
\end{equation}
Then the number of vector multiplets $n_V^{(6)}$, tensor multiplets $n_T$, and hypermultiplets $n_H$ are computed to be
\begin{align}
\begin{split}
n_V^{(6)}&=13, \quad n_T=9-K^2, \\
n_H&=h^{2,1}(Y)+1+n_{\bf{3},\bf{1}}(3-1)+n_{\bf{2},\bf{4}} (10-2) + n_{\bf{1},\bf{5}} (5 - 1) + n_{\bf{1},\bf{10}} (10 - 2)\\
&=17 K^2+16 K T+6 T^2+14 .
\end{split}
\end{align}
 We recall  the number of representations from the earlier subsection:
\begin{align}
\begin{split}
&n_{\bf{3},\bf{1}}=(2 K+T) (3 K+2 T)+1, \ n_{\bf{2},\bf{4}}=-2 T (2 K+T), \\
&n_{\bf{1},\bf{5}}=\frac{1}{2} T (T - K), \quad n_{\bf{1},\bf{10}}=\frac{1}{2}(KT+T^2+2) .
\end{split}
\end{align}
Thus, we see that
\begin{equation}
n_H-n_V^{(6)}+29n_T-273=0,
\end{equation}
which means that the pure gravitational anomalies are canceled.

By using the trace identities for $\text{SU($2$)}$ given by equation \eqref{eq:SU2trace}, we first compute the $\text{SU($2$)}$ side of the anomaly polynomials. First, we can determine that
\begin{equation}
n_{\bf{3}}=n_{\bf{3},\bf{1}} , \quad n_{\bf{2}}=4n_{\bf{2},\bf{4}}.
\end{equation}
Hence, $X^{(2)}_1$ and $X^{(4)}_1$ are given by
\begin{align}
X^{(2)}_{1}&=\left(A_{\bf{3}}(1-n_{\bf{3}})-n_{\bf{2}}A_{\bf{2}}\right)\tr_{\bf{2}}F^2_1 =-12K(2K+T)\tr_{\bf{2}}F^2_1\\
\begin{split}
X^{(4)}_{1}&=\left(B_{\bf{3}}(1-n_{\bf{3}})-n_{\bf{2}}B_{\bf{2}}\right)\tr_{\bf{2}}F^4_1+\left(C_{\bf{3}}(1-n_{\bf{3}})-n_{\bf{2}}C_{\bf{2}}\right) (\tr_{\bf{2}}F^2_1)^2 \\
&=-12(2K+T)^2 (\tr_{\bf{2}}F^2_1)^2.
\end{split}
\end{align}
Now consider the $\text{Sp($4$)}$ side of the anomaly cancellation by using the trace identities for $\text{Sp($4$)}$ given by equation \eqref{eq:Sp4trace}. We first determine that
\begin{equation}
n_{\bf{10}}=n_{\bf{1},\bf{10}} , \quad n_{\bf{5}}=n_{\bf{1},\bf{5}} , \quad n_{\bf{4}}=2n_{\bf{2},\bf{4}}.
\end{equation}
Hence, $X^{(2)}_2$ and $X^{(4)}_2$ are given by
\begin{align}
X^{(2)}_{2}&=\left(A_{\bf{10}}(1-n_{\bf{10}})-n_{\bf{5}}A_{\bf{5}}-n_{\bf{4}}A_{\bf{4}}\right)\tr_{\bf{4}}F^2_2 =6K T \tr_{\bf{4}}F^2_2\\
\begin{split}
X^{(4)}_{2}&=\left(B_{\bf{10}}(1-n_{\bf{10}})-n_{\bf{5}}B_{\bf{5}}-n_{\bf{4}}B_{\bf{4}}\right)\tr_{\bf{4}}F^4_2+\left(C_{\bf{10}}(1-n_{\bf{10}})-n_{\bf{5}}C_{\bf{5}}-n_{\bf{4}}C_{\bf{4}}\right) (\tr_{\bf{4}}F^2_2)^2 \\ &=-3T^2 (\tr_{\bf{4}}F^2_2)^2.
\end{split}
\end{align}
Now we further include on both the $\text{SU($2$)}$ and $\text{Sp($4$)}$ sides the additional mixed term
\begin{equation}
Y_{12}=n_{\bf{2},\bf{4}} \tr_{\bf{2}}F^2_1 \tr_{\bf{4}}F^2_2;
\end{equation}
this is necessary to fully consider the bifundamental representation $(\bf{2},\bf{4})$. Then the full anomaly polynomial is given by
\begin{align}
\begin{split}
I_8&=\frac{9-n_T}{8} (\tr R^2)^2 +\frac{1}{6} (X^{(2)}_{1} +X^{(2)}_{2}) \tr R^2-\frac{2}{3} (X^{(4)}_{1}+X^{(4)}_{2})+4Y_{12} \\
&=\frac{1}{8} \left(K\tr R^2 -16 K \tr_{\bf{4}}F^2_2-8 T\tr_{\bf{4}}F^2_2+4 T \tr_{\bf{2}}F^2_1\right)^2,
\end{split}
\end{align}
which is a perfect square. This means that the total anomalies are canceled.

\section{$\text{SU($2$)}\times \text{Sp($4$)}$-model}  \label{sec:nonsplitNoZ2}

In this section we consider the I$_2 ^{\text{ns}}+$I$_4 ^{\text{ns}}$-model with a trivial Mordell-Weil group. The Weierstrass model is
\begin{equation}
y^2z=x^3+a_2x^2z+{\widetilde{a}_4}st^2xz^2+{\widetilde{a}_6}s^2t^4z^3.
\end{equation}
The discriminant is 
\begin{equation}
\Delta=-16 s^2 t^4 (4 a_2^3{\widetilde{a}_6}-a_2^2{\widetilde{a}_4}^2-18 a_2 {\widetilde{a}_4} {\widetilde{a}_6} s t^2+4 {\widetilde{a}_4}^3 s t^2+27 {\widetilde{a}_6}^2 s^2 t^4).
\end{equation}
The corresponding simply connected group $G$ and the representation $\mathbf{R}$, which is computed geometrically in the next section, are 
\begin{equation}
G=\text{SU($2$)}\times \text{Sp($4$)}, \quad \bf{R}=(\bf{3},\bf{1})\oplus(\bf{1},\bf{10})\oplus (\bf{2},\bf{1})\oplus (\bf{2},\bf{4})\oplus (\bf{1},{5}) \oplus (\bf{2},{1})\oplus (\bf{1},\bf{4}).
\end{equation}
The following sequence of blowups is a crepant resolution of the Weierstrass model:
\begin{equation}
  \begin{tikzcd}[column sep=huge] 
  X_0  \arrow[leftarrow]{r} {(x,y,s|e_1)} & \arrow[leftarrow]{r} {(x,y,t|w_1)}  X_1 &X_2  \arrow[leftarrow]{r}{(x,y,w_1|w_2)} &X_3.
  \end{tikzcd}
\end{equation}
The proper transform is
\begin{equation}
y^2=e_1w_1w_2^2x^3+a_2x^2+{\widetilde{a}_4}st^2w_1x+{\widetilde{a}_6} s^2t^4,
\end{equation}
and the relative projective coordinates are
\begin{equation}
[e_1w_1w_2^2x : e_1w_1w_2^2y : z=1][w_1w_2^2x : w_1w_2^2y : s][w_2x : w_2y : t][x : y : w_1].
\end{equation}
To prove that this is a crepant resolution, it is enough to assume that  $V(a_2)$, $V({\widetilde{a}_6})$, $S=V(s)$, and $T=V(t)$  are 
smooth divisors intersecting two by two transversally.  
\subsection{Fiber structure and representations}
The fibral divisors for this model are
\begin{align}
\begin{cases}
 D_0^{\text{s}} &: s=y^2-x^2(e_1w_1w_2^2x+a_2)=0 , \\
 D_1^{\text{s}} &: e_1=y^2-a_2x^2-{\widetilde{a}_4}st^2w_1x-{\widetilde{a}_6}s^2t^4w_1^2=0 , \\
 D_0^t &: t=y^2-x^2(e_1w_1w_2^2x+a_2)=0 , \\
 D_1^t &: w_1=y^2-a_2x^2=0 , \\
 D_2^t &: w_2=y^2-a_2x^2-{\widetilde{a}_4}st^2w_1x-{\widetilde{a}_6}s^2t^4w_1^2=0 .
\end{cases}
\end{align}
The fiber I$_2^{\text{ns}}$  specializes to a fiber of  type III over $V(a_2)$ and a fiber of type I$_3^{\text{ns}}$ over $V(a_2{\widetilde{a}_6}-4{\widetilde{a}_4}^2)$ as the generic fiber of $D_1^{\text{s}}$ degenerates into two lines $C_{1\pm}^{\text{s}}$. 
 The intersection numbers between the curves and the fibral divisors of I$_2^{\text{ns}}$ are
\begin{equation}
\begin{array}{c|ccccc}
 &D_0^{\text{s}}
 & D_1^{\text{s}} & D_0^t  & D_1^t   & D_2^t   \\
 \hline
C_0^{\text{s}} & -2 & 2  & 0 & 0 & 0  \\
C_{1\pm}^{\text{s}} & 1 & -1  & 0 & 0 & 0  \\
C_{1+}^{\text{s}}+C_{1-}^{\text{s}} & 2 & -2 & 0 & 0 & 0 
\end{array}
\end{equation}
We get the weight $[-1]$ from each copy of $C_{1\pm}^{\text{s}}$. This is in the representation $\bf{2}$ of $A_1$ and uncharged from $\frak{sp}(4)$ as it is away from its locus. Hence, the charged matter is in the representation  $(\bf{2},\bf{1})$.

Over $T=V(t)$, we have a generic fiber of type I$_4^{\text{ns}}$, whose geometric components are  $C_0^{\text{s}}$, $C_{1+}^t$, $C_{1-}^t$, and $C_2^t$. The curve $C_2^t$ is  a conic that splits into two lines over  $V(4a_2{\widetilde{a}_6}-{\widetilde{a}_4}^2)$, 
\begin{equation}
C_2^t \rightarrow C_{2+}^t+C_{2-}^t ,
\end{equation}
which results in the degeneration I$_4^{\text{ns}}\to$ I$_5^{\text{ns}}$. Then we get an enhancement of type I$_5^{\text{ns}}$, which is represented in Figure \ref{I2nsI4nsNoZ2}. Based on the splitting of the curve $C_2^t$, with $C_0^t$, $C_1^t$, and $C_3^t$ remaining  the same, the intersection numbers between these curves and the fibral divisors of I$_4^{\text{ns}}$ are computed to find the weights of the new curves:
\begin{equation}
\begin{array}{c|rrrrr}
 & D_0^{\text{s}} & D_1^{\text{s}}   & D_0^t & D_1^t & D_2^t  \\
 \hline
C_0^t &0 & 0 &  -2 & 2 & 0  \\
C_{1+}^t & 0 & 0 & 1 & -2 & 1  \\
C_{2+}^t & 0 & 0 & 0 & 1 & -1  \\
C_{2-}^t & 0 & 0 & 0 & 1 & -1  \\
C_{1-}^t & 0 & 0 & 1 & -2 & 1 
\end{array}
\end{equation}
From the new splittings of the curves, each produce the weight $[-1,1]$, which corresponds to the representation $\bf{4}$ of $\frak{sp}(4)$. Since this is away from the locus of $se_1=0$, the weight produced is simply $[0;-1,1]$. Thus, the charged matter is in the representation $(\bf{1},\bf{4})$ of $(\frak{su}(2),\frak{sp}(4))$.

At the collision of $S$ and $T$, we get  the following curves:
\begin{align}
\begin{cases}
C_0^{\text{s}}\cap C_0^t &: s=t=y^2-x^2(e_1w_1w_2^2x+a_2)=0 \rightarrow \eta^{00} , \\
C_1^{\text{s}}\cap C_0^t &: e_1=t=y^2-a_2x^2=0 \rightarrow \eta^{10 \pm} , \text{ (two roots for each curve,)}\\
C_1^{\text{s}}\cap C_1^t &: e_1=w_1=y^2-a_2x^2=0 \rightarrow \eta^{11 \pm} , \text{ (two roots for each curve,)}\\
C_1^{\text{s}}\cap C_2^t &: e_1=w_2=y^2-a_2x^2-{\widetilde{a}_4}st^2w_1x-{\widetilde{a}_6}s^2t^4w_1^2=0 \rightarrow \eta^{12} .
\end{cases}
\end{align}
The fiber structure is presented in Figure \ref{I2nsI4nsNoZ2}. As expected, we get a natural enhancement of an I$_6^{\text{ns}}$. 
The fibers of the collisions can be described from the splitting of the curves $C_i^{\text{s}}$ ($i=0,1$) from I$_2^{\text{ns}}$ and the curves $C_i^t$ ($i=0,1,2$) fom I$_4^{\text{ns}}$. 
From these splittings of the curves, we compute the intersection numbers between the curves and the fibral divisors of I$_2^{\text{ns}}$ and I$_4^{\text{ns}}$. The splitting of the curves $C_a^{\text{s}}$ and their intersection numbers with the fibral divisors 
are computed to be 
\begin{align}
\begin{array}{c|ccccc}
 & D_0^{\text{s}} & D_1^{\text{s}} & D_0^t & D_1^t & D_2^t \\
 \hline
\eta^{00} & -2 & 2 & 0 & 0 & 0 \\
\eta^{10+}+\eta^{10-} & 2 & -2 & -2 & 2 & 0 \\
\eta^{10\pm} & 1 & -1 & -1 & 1 & 0 \\
\eta^{11+}+\eta^{11-} & 0 & 0 & 2 & -4 & 2 \\
\eta^{11\pm} & 0 & 0 & 1 & -2 & 1 \\
\eta^{12} & 0 & 0 & 0 & 2 & -2 
\end{array} \quad\quad\quad
\begin{cases}
C_0^{\text{s}}  & \rightarrow \eta^{00}\\
C_1^{\text{s}}  & \rightarrow \eta^{10+}+\eta^{10-}+\eta^{11+}+\eta^{11-}+\eta^{12} \\
C_0^t  & \rightarrow \eta^{00}+\eta^{10+}+\eta^{10-} \\
C_1^t & \rightarrow \eta^{11+}+\eta^{11-} \\
C_2^t  & \rightarrow \eta^{12} 
\end{cases}
\end{align}
 The curve $\eta^{10\pm}$  yields the representation $(\bf{2},\bf{4})$. These nonsplit curves together $\eta^{10+}+\eta^{10-}$ produce the weight $[2;-2,0]$, the corresponding representation is $(\bf{3},\bf{10})$. Hence, the representation for the I$_2^{\text{ns}}+$I$_4^{\text{ns}}$-model with a trivial Mordell-Weil group is
$
\bf{R}=(\bf{3},\bf{1})\oplus (\bf{1},\bf{10})\oplus (\bf{3},\bf{10})\oplus (\bf{2},\bf{4})\oplus (\bf{1},{5})\oplus (\bf{2},\bf{1})\oplus (\bf{1},{4}).
$
We note that for the case of threefolds, the curves $\eta^{10\pm}$ are always split since all curve can split over a codimension-two point. Hence, the bi-adjoint $(\bf{3},\bf{10})$ does not show up geometrically, and the representation is then
\begin{equation}
\bf{R}=(\bf{3},\bf{1})\oplus (\bf{1},\bf{10})\oplus (\bf{2},\bf{4})\oplus (\bf{1},{5})\oplus (\bf{2},\bf{1})\oplus (\bf{1},{4}) .
\end{equation}
We summarize the representations and the weights of this model with their locus in Table \ref{Rep.NonsplitModel}.
\begin{table}[H]
\begin{center}
\begin{tabular}{|c|c|c|c|c|c|}
\hline
Locus &  \multicolumn{2}{c|}{$tw_1w_2=0$} &  \multicolumn{2}{c|}{$se_1=tw_1w_2=0$} & $se_1=0$\\
\hline
Curves & $C_1^{t'}$ & $C_{2\pm}^t$ & $\eta^{10\pm}$ & $\eta^{10+}+\eta^{10-}$ & $C_{1\pm}$\\
\hline
Weights &  $[0;2,-1]$ & $[0;-1,1]$ & $[1;-1,0]$ & $[2;-2,0]$ & $[-1;0,0,0]$ \\
\hline
Representations & $(\bf{1},\bf{5})$ & $(\bf{1},\bf{4})$ & $(\bf{2},\bf{4})$ & $(\bf{3},\bf{10})$ & $(\bf{2},\bf{1})$ \\
\hline
\end{tabular}
\end{center}
\caption{Weights and representations for the I$_2^{\text{ns}}+$I$_4^{\text{ns}}$-model with a trivial Mordell-Weil group.}
\label{Rep.NonsplitModel}
\end{table} 

{ When $a_2=0$, $\eta^{10\pm}$ becomes a degenerate node $\eta^{10'}$ with degeneracy two and $\eta^{11\pm}$ also degenerates into a single node $\eta^{11'}$ of degeneracy two:
\begin{align}
\eta^{10\pm}& \rightarrow \eta^{10'}, \quad \eta^{11\pm} \rightarrow \eta^{11'},
\end{align}
where the two new types of the curves are given by
\begin{align}
\eta^{10'} &: e_1=t=y=0 , \ \text{and} \ \eta^{11'} : e_1=w_1=y=0 .
\end{align}
This new fiber structure for $a_2=0$ is the bottom left diagram of Figure \ref{I2nsI4nsNoZ2}. 
When ${\widetilde{a}_4}^2-a_2{\widetilde{a}_6}=0$ instead, $\eta^{12}$ is geometrically reducible, and thus becomes two nonsplit curves 
\begin{equation}
\eta^{12}\rightarrow \eta^{12+'}+\eta^{12-'}.
\end{equation}
Therefore, we get the bottom right diagram in Figure \ref{I2nsI4nsNoZ2}.
When $a_2={\widetilde{a}_4}=0$, $\eta^{12}$ geometrically reduces into two nonsplit curves
\begin{equation}
\eta^{12} \rightarrow \eta^{12+''}+\eta^{12-''} : e_1=w_2=y^2-{\widetilde{a}_6}s^2t^4w_1^2=0,
\end{equation}
whereas from the other specialization when ${\widetilde{a}_4}^2-a_2{\widetilde{a}_6}={\widetilde{a}_4}=0$,
\begin{align}
\begin{cases}
\eta^{10} &\rightarrow \eta^{10'}, \\
\eta^{11} &\rightarrow \eta^{11'}, \\
\eta^{12} &\rightarrow \eta^{12+''}+\eta^{12-''} .
\end{cases}
\end{align}
When $a_2={\widetilde{a}_4}={\widetilde{a}_6}=0$, the nonsplit fibers become $\eta^{12'''}$ with degeneracy two:
\begin{equation}
\eta^{12\pm''}\rightarrow \eta^{12'''} : e=w_2=y=0,
\end{equation}
which gives the fiber structure to be the very bottom diagram in Figure \ref{I2nsI4nsNoZ2}.}

\begin{figure}[H]
\begin{center}
\vspace{-1.0cm}
\lapbox{-0.9cm}
{\includegraphics[scale=1]{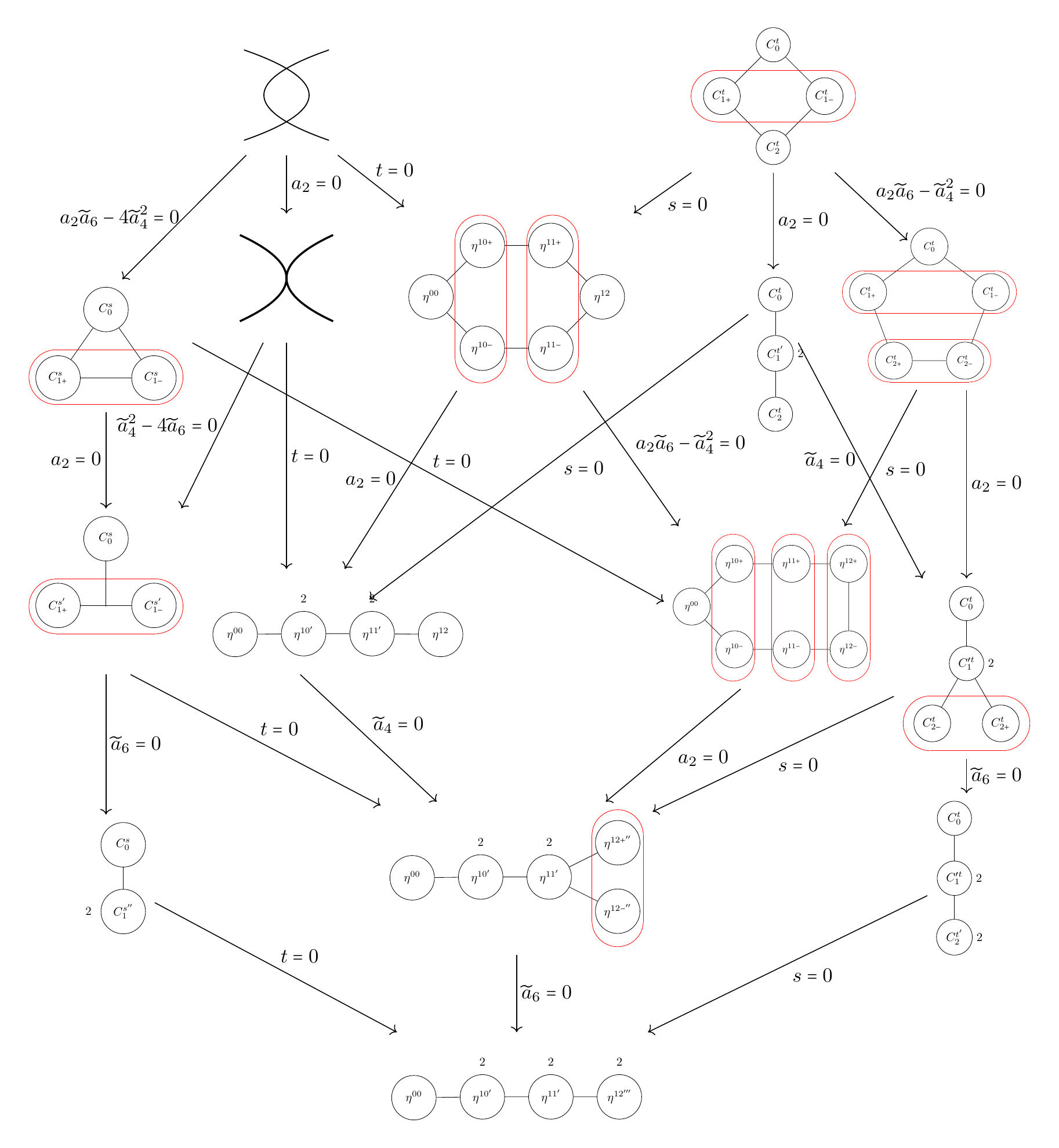}}
\caption{Fiber structure of I$_2^{\text{ns}}+$I$_4^{\text{ns}}$ with a trivial Mordell-Weil group. The diagrams on the top are the Kodaira fibers of type I$_2^{\text{ns}}$ and I$_4^{\text{ns}}$. When $a_2=0$, I$_2^{\text{ns}}$ specializes to III as seen in the middle left diagram, and I$_4^{\text{ns}}$ specializes to the diagram on the middle right. When I$_2^{\text{ns}}$ and I$_4^{\text{ns}}$ collide on the locus $se_1=tw_1w_2w_3=0$, we get the fiber structure as an hexagon drawn in the middle. This enhancement has two newly split curves $\eta^{10+}$ and $\eta^{10-}$ that are of non-split type. When $a_2=0$, this hexagon further specializes to the diagram on the bottom left, but when ${\widetilde{a}_4}^2-a_2{\widetilde{a}_6}=0$, the hexagon becomes a heptagon.} 
\label{I2nsI4nsNoZ2}
\end{center}
\end{figure}

\subsection{Coulomb phases} \label{IMSnonsplitNoZ2}
The groups $\text{SU($2$)}$ and $\text{Sp($4$)}$ individually have a unique Coulomb chamber. Their product with the bifundamental representation introduces an interior wall inside the Weyl chamber, which yields three chambers. Since we have the same Lie algebra with the bifundamental representation $(\bf{2},\bf{4})$, the chamber structure does not change under the change of the Mordell-Weil group.

\subsection{$5d$  $\mathcal{N}=1$ prepotentials and the triple intersection polynomials}
The I$_2^{\text{ns}}+$I$_4^{\text{ns}}$-model with a trivial Mordell-Weil group has the identical blowups with  the nonsplit model with the Mordell-Weil group $\mathbb{Z}_2$, and the fiber structure in codimension two is the same for the hexagon. Hence, the triple intersection polynomial and the prepotential are the same with the I$_2^{\text{ns}}+$I$_4^{\text{ns}}$-model with the $\mathbb{Z}_2$, without the relation between two divisors of class $S$ and $T$. Thus, the triple intersection polynomial is given by
\begin{align}
\begin{split}
\mathscr{F}_{trip}=& -2 S (2 L+S)\psi _1^3-8T^2 \phi _1^3-2T (2 L-S)\phi _2^3 \\
&+6T(S+2 T-2 L)\phi _1^2 \phi _2 -6T(-2 L+S+T) \phi _1\phi _2^2 -6 S T \psi _1 \left(2 \phi _1^2-2 \phi _2 \phi _1+\phi _2^2\right) \\
&+6 T \phi _0^2 \left((S+T-2 L) \phi _1 -S \psi _1\right) +6 T \phi _0 \left((2 L-S)\phi _1^2 -S \left(\psi _0-\psi _1\right)^2+2 S \psi _1 \phi _1\right) \\
&-2 T (-2 L+S+2 T)\phi _0^3 +4 S \psi _0^3 (L-S)+6 S(S-2 L) \psi _0^2\psi _1+12 L S\psi _0 \psi _1^2  .
\label{eq:nonsplit.tri.NoZ2}
\end{split}
\end{align}
The representations we achieve geometrically are
\begin{equation}
\bf{R}=(\bf{3},\bf{1})\oplus(\bf{2},\bf{1})\oplus(\bf{2},\bf{4})\oplus(\bf{1},\bf{4})\oplus(\bf{1},\bf{5})\oplus(\bf{1},\bf{10}).
\end{equation}
Using these representations, the $5d$ prepotential is computed for chamber $1$,  which is $[-,+]$, to be
\begin{align}
\begin{split}
6 \mathscr{F}_{\text{IMS}}=& - (n_{\bf{2},\bf{1}}-4n_{\bf{2},\bf{4}}+8 n_{\bf{3},\bf{1}}-8) \psi _1^3 -8 (n_{\bf{1},\bf{10}}+n_{\bf{1},\bf{5}}-1) \phi _1^3-(8 n_{\bf{1},\bf{10}}+n_{\bf{1},\bf{4}}-8)\phi _2^3 \\
& -3 \phi _1^2 \phi _2 (4 n_{\bf{1},\bf{10}}+n_{\bf{1},\bf{4}}-4 n_{\bf{1},\bf{5}}-4) +3(6 n_{\bf{1},\bf{10}}+n_{\bf{1},\bf{4}}-2n_{\bf{1},\bf{5}}-6) \phi _1\phi _2^2 \\
&+\psi _1 \left(-12 n_{\bf{2},\bf{4}} \phi _1^2+12 n_{\bf{2},\bf{4}} \phi _1\phi _2 -6 n_{\bf{2},\bf{4}} \phi _2^2\right)
\end{split}
\end{align}
The triple intersection numbers that are independent from $\psi_0$ and $\phi_0$, which are the first two lines of the equation \eqref{eq:nonsplit.tri.NoZ2}, are matched with the $5d$ prepotential term by term. This  fixes the linear combination of the number of representations $n_{\bf{R}}$:
\begin{align}
\begin{split}
&8n_{\bf{3},\bf{1}}+n_{\bf{2},\bf{1}}=2 S (2 L + S - 2 T) + 8, \\
&8n_{\bf{1},\bf{10}}+n_{\bf{1},\bf{4}}=2 (2 L T - S T + 4), \quad n_{\bf{1},\bf{10}}+n_{\bf{1},\bf{5}}=T^2+1.
\end{split}
\end{align}
However, we need further information    in order to fix the number of representations. 

Using the fact that the charged matter in the representation $(\bf{1},\bf{4})$ is from the splittings of the curve $C_2^t \rightarrow C_{2+}^t+C_{2-}^t$ when $4a_2{\widetilde{a}_6}-{\widetilde{a}_4}^2=0$, whose class is given by twice of $[{\widetilde{a}_4}]=(4L-S-2T)$, in the locus of $tw_1w_2=0$, we  can safely see that this gives the class $n_{\bf{1},\bf{4}}=2T(4L-S-2T)$. With this specified $n_{\bf{1},\bf{4}}$, we can fix the number of representations for $n_{\bf{1},\bf{5}}$ and $n_{\bf{1},\bf{10}}$. Moreover, $n_{\bf{2},\bf{1}}$ and $n_{\bf{3},\bf{1}}$ terms are both only in $\psi_1^3$ term. In order to fix these representations, we look into the specialization of the curve $C_1^{\text{s}} \rightarrow C_{1+}^{\text{s}}+C_{1-}^{\text{s}}$ that produces the charged matter in the representation $(\bf{2},\bf{1})$. This specialization was when $4a_2{\widetilde{a}_6}-{\widetilde{a}_4}^2=0$, whose class is then the same  as twice the class of $[{\widetilde{a}_4}]=4L-S-2T$. Since the splittings happen when $se_1=4a_2{\widetilde{a}_6}-{\widetilde{a}_4}^2=0$, we can safely see that this gives the class $n_{\bf{2},\bf{1}}=2S(4L-S-2T)$. Hence,  the number of representations $n_{\bf{R}}$ can all be computed, which are listed below:
\begin{align}
\begin{split}
&n_{\bf{3},\bf{1}}=\frac{1}{2} \left(-L S+S^2+2\right), \quad n_{\bf{2},\bf{1}}=2 S (4 L-S-2 T), \quad n_{\bf{2},\bf{4}}=ST,\\
&n_{\bf{1},\bf{4}}=2 T (4 L - S - 2 T) \quad n_{\bf{1},\bf{5}}=\frac{1}{2} T (L+T), \quad n_{\bf{1},\bf{10}}=\frac{1}{2} \left(-L T+T^2+2\right) .
\end{split}
\end{align}

When we impose the Calabi-Yau condition, $L=-K$,
$n_{\bf{3},\bf{1}}=g_S$ and  $n_{\bf{1},\bf{10}}=g_T$.

\subsection{$6d$ $\mathcal{N}=(1,0)$ anomaly cancellation}
To compute the number of hypermultiplets, we recall  the number of representations from section \ref{IMSnonsplitNoZ2}:
\begin{align}
\begin{split}
&n_{\bf{3},\bf{1}}=\frac{1}{2} \left(K S+S^2+2\right), \quad n_{\bf{2},\bf{1}}=-2 S (4K+S+2 T), \quad n_{\bf{2},\bf{4}}=ST, \\
&n_{\bf{1},\bf{4}}=-2 T (4K+S+2 T), \quad n_{\bf{1},\bf{5}}=-\frac{1}{2}(KT-T^2), \quad n_{\bf{1},\bf{10}}=\frac{1}{2} \left(K T+T^2+2\right) .
\end{split}
\end{align}
The number of vector multiplets $n_V^{(6)}$, tensor multiplets $n_T$, and hypermultiplets $n_H$ are
\begin{align}
\begin{split}
n_V^{(6)}&=13, \quad n_T=9-K^2, \\
n_H&=h^{2,1}(Y)+1+n_{\bf{2},\bf{1}}(2-0)+n_{\bf{3},\bf{1}}(3-1)+n_{\bf{2},\bf{4}} (8-0) + n_{\bf{1},\bf{4}} (4-0)\\
&\quad + n_{\bf{1},\bf{5}} (5 - 1) + n_{\bf{1},\bf{10}} (10 - 2)=29 K^2+25.
\end{split}
\end{align}
The coefficients of the $\tr R^4$ vanishes as $
n_H-n_V^{(6)}+29n_T-273=0$,
so we can conclude that the pure gravitational anomalies are canceled out.

The terms $X^{(2)}_1$, $X^{(2)}_2$, $X^{(4)}_1$, $X^{(4)}_2$, and $Y_{12}$ are obtained as\footnote{
A key point of the computation is that  along SU($2$), a hypermultiplet transforming in the representation $(\bf{2},\bf{4})$ is seen as 
$4$ hypermultiplets in the representation $\bf{2}$ of SU($2$). In the same way, the same hypermultiplet in the representation $(\bf{2},\bf{4})$ is seen from the group Sp($4$) as 
 $2$  hypermultiplets in the representation $\bf{4}$ of Sp($4$). It follows that we use
$$
\begin{aligned}
n_{\bf{3}}=n_{\bf{3},\bf{1}}, &\quad n_{\bf{2}}=n_{\bf{2},\bf{1}}+4n_{\bf{2},\bf{4}}, \quad
n_{\bf{10}}=n_{\bf{1},\bf{10}}, \quad n_{\bf{5}}=n_{\bf{1},\bf{5}},  \quad n_{\bf{4}}=n_{\bf{1},\bf{4}}+2n_{\bf{2},\bf{4}}.
\end{aligned}
$$
}
\begin{align}
X^{(2)}_{1}&=\left(A_{\bf{3}}(1-n_{\bf{3}})-n_{\bf{2}}A_{\bf{2}}\right)\tr_{\bf{2}}F^2_1 =6KS \tr_{\bf{2}}F^2_1 ,\\
X^{(2)}_{2}&=\left(A_{\bf{10}}(1-n_{\bf{10}})-n_{\bf{5}}A_{\bf{5}}-n_{\bf{4}}A_{\bf{4}}\right)\tr_{\bf{4}}F^2_2 =6KT\tr_{\bf{4}}F^2_2 ,\\
Y_{12} &=n_{\bf{2},\bf{4}} \tr_{\bf{2}}F^2_2 \tr_{\bf{4}}F^2_2=ST \tr_{\bf{2}}F^2_2 \tr_{\bf{4}}F^2_2 ,\\
X^{(4)}_{1}&=\left(B_{\bf{3}}(1-n_{\bf{3}})-n_{\bf{2}}B_{\bf{2}}\right)\tr_{\bf{2}}F^4_1+\left(C_{\bf{3}}(1-n_{\bf{3}})-n_{\bf{2}}C_{\bf{2}}\right) (\tr_{\bf{2}}F^2_1)^2 =-3S^2 (\tr_{\bf{2}}F^2_1)^2 ,\\
X^{(4)}_{2}&=\left(B_{\bf{10}}(1-n_{\bf{10}})-n_{\bf{5}}B_{\bf{5}}-n_{\bf{4}}B_{\bf{4}}\right)\tr_{\bf{4}}F^4_2\\
&\quad +\left(C_{\bf{10}}(1-n_{\bf{10}})-n_{\bf{5}}C_{\bf{5}}-n_{\bf{4}}C_{\bf{4}}\right) (\tr_{\bf{4}}F^2_2)^2 =-3 T^2(\tr_{\bf{4}}F^2_2)^2 ,
\end{align}
Since we have a theory with two quartic Casimirs, to satisfy the anomaly cancellation conditions, the coefficients of the $\tr_{\bf{2}}F^4_1$ and $\tr_{\bf{4}}F^4_2$ must vanish. These terms are coming from $X^{(4)}_1$ and $X^{(4)}_2$. We can observe from the equations above that the coefficients of $\tr_{\bf{2}}F^4_1$ and $\tr_{\bf{2}}F^4_1$ indeed vanish. 

Using the terms above, we compute the anomaly polynomial as
\begin{equation}
I_8=\frac{1}{2} \left(\frac{1}{2} K \tr R^2+2 S\tr_{\bf{4}}F^2_2+2 T\tr_{\bf{2}}F^2_1\right)^2 ,
\end{equation}
which is a perfect square. Hence we conclude that the anomalies are all canceled when lifted to the six-dimensional theories.

\section{$(\text{SU($2$)}\times \text{SU($4$)})/\mathbb{Z}_2$-model} \label{sec:splitZ2}

In this section, we  study the I$_2^{\text{ns}}$+I$_4^{\text{s}}$ with a $\mathbb{Z}_2$ Mordell-Weil group.  The Weierstrass model is
\begin{equation}
y^2z+a_1xyz=x^3+{\widetilde{a}_2}tx^2z+st^2xz^2 .
\end{equation}
The discriminant for this model is
\begin{equation}
\Delta=s^2 t^4 \left(a_1^4+8 a_1^2 {\widetilde{a}_2} t+16 {\widetilde{a}_2}^2 t^2-64 s t^2\right).
\end{equation}
The corresponding simply connected group $G$ and the representation $\bf{R}$, which is computed geometrically in the next section, are
\begin{equation}
G=(\text{SU($2$)}\times \text{SU($4$)})/\mathbb{Z}_2 , \quad \bf{R}= (\bf{3},\bf{1})\oplus(\bf{1},\bf{15})\oplus(\bf{1},\bf{6})\oplus (\bf{2},\bf{4})\oplus (\bf{2},\bf{\bar{4}}).
\end{equation}

We consider the following sequence of blowups for a crepant resolution:
\begin{equation}
\begin{tikzcd}[column sep=huge]  X_0  \arrow[leftarrow]{r} {(x,y,s|e_1)} & \arrow[leftarrow]{r} {(x,y,t|w_1)}  X_1 &X_2  \arrow[leftarrow]{r}{(y,w_1|w_2)} &X_3 \arrow[leftarrow]{r}{(x,w_2|w_3)} & X_4 \end{tikzcd}.
\end{equation}
The proper transform is
\begin{equation}
w_2y^2+a_1xy=w_1 x(e_1 w_3^2 x^2+{\widetilde{a}_2}w_3 tx+st^2),
\end{equation}
where the relative projective coordinates are given by
\begin{equation}
[e_1w_1w_2w_3^2x : e_1w_1w_2^2w_3^2y : z=1][w_1w_2w_3^2x : w_1w_2^2w_3^2y : s][w_3x : w_2w_3y : t][y : w_1][x : w_2].
\end{equation}

\subsection{Fiber structure}
This model has the following fibral divisors that corresponds to their curves:
\begin{align}
\begin{cases}
 D_0^{\text{s}} &: s=w_2y^2+a_1xy-w_1w_3x^2(e_1w_3x+{\widetilde{a}_2}t)=0 , \\
 D_1^{\text{s}} &: e_1=w_2y^2+a_1xy-w_1tx({\widetilde{a}_2}w_3x+st)=0 , \\
 D_0^t &: t=w_2y^2+a_1xy-e_1w_1w_3^2x^3=0 , \\
 D_1^t &: w_1=w_2y+a_1x=0\\
 D_2^t &: w_3=w_2y^2+ x(a_1y-w_1 st^2)=0\\
 D_3^t &: w_2=a_1y-w_1(e_1w_3^2x^2+{\widetilde{a}_2}w_3 t x+st^2)=0 .
\end{cases}
\end{align}
The fiber of type I$_2^{\text{ns}}$ consists of two curves $C_i^{\text{s}}$ ($i=0,1$) with their intersection point given by
\begin{equation}
C_0^{\text{s}} \cap C_1^{\text{s}} : s=e_1=w_2 y^2+a_1xy-{\widetilde{a}_2} tw_1w_3x^2=0.
\end{equation}
Hence it specializes to a type III fiber over $V(a_1^2+4{\widetilde{a}_2} t)$.

On the other hand, over $T=V(t)$, we have a generic fiber of type I$_4^{\text{s}}$ with its geometric components $C_0^t$, $C_1^t$, $C_2^t$, and $C_3^t$. This fiber I$_4^{\text{s}}$ enhances into a fiber of type I$_0^{*ss}$ over $V(a_1)$, which is presented in Figure \ref{I2sI4s}. The curve $C_1^t$ specializes into the central node $C_{13}^t$ where
\begin{equation}
C_1^t \rightarrow C_{13}^t : w_1=w_2=0 .
\end{equation}
The curve $C_3^t$ splits into three curves $C_{12}^t$, which is the central node, and $C_{3 \pm}^{t'}$, which is given by the two roots of the curve:
\begin{align}
& C_3^t  \rightarrow  C_{13}^t+C_{3+}^{t'}+C_{3-}^{t'} \quad \quad (
C_{3 \pm}^{t'}: w_2=e_1w_3^2x^2+{\widetilde{a}_2}tw_3x+st^2=0).
\end{align}
From this specialization, we can compute the weights of the curves to see what charged matter we have in the five-dimensional theory. 
\begin{equation}
\begin{tabular}{c|cccccc}
  & $D_0^{\text{s}}$  & $D_1^{\text{s}}$ & $D_0^t$ & $D_1^t$ & $D_2^t$ & $D_3^t$ \\
 \hline
$C_0^t$ & 0 & 0 & -2 & 2 & 0 & 0 \\
$2 \, C_{13}^t$ & 0 & 0 & 2 & -4 & 2 & 0 \\
$C_2^t$ & 0 & 0 & 0 & 0 & -2 & 2 \\
$C_{3+}^{t'}+C_{3-}^{t'}$ & 0 & 0 & 0 & 2 & 0 & -2
\end{tabular}
\quad \raisebox{-30pt}{.}
\end{equation}
For the curve $C_{13}^t$, the weight is computed to be $[2,-1,0]$, and this corresponds to the representation $\bf{15}$, which is the adjoint representation of $\frak{su}(4)$. The curves $C_{3\pm}^{t'}$ each carries the weight $[-1,0,1]$ that corresponds to the representation $\bf{6}$, which is the antisymmetric representation of $\frak{su}(4)$. 
{ Since these two curves are nonsplit, when we compute the weight of them together as $C_{3}^{t'}=C_{3+}^{t'}+C_{3-}^{t'}$, the representation is  $\bf{20'}$.} Since the locus is away from $se_1=0$, it is uncharged on the side for $\frak{su}(2)$ and hence the representation for the whole product group $\frak{su}(2)\times\frak{su}(4)$ is $\bf{(1,6)}$.

At the collision of both $S$ and $T$, 
\begin{align}
\begin{cases}
C_0^{\text{s}}\cap C_0^t :& s=t=w_2y^2+a_1xy-w_1w_3^2e_1x^3=0 \rightarrow \eta^{00} , \\
C_1^{\text{s}}\cap C_0^t :& e_1=t=w_2y+a_1x=0 \rightarrow \eta^{10} , \\
& e_1=t=y=0 \rightarrow \eta^{10y} , \\
C_1^{\text{s}}\cap C_1^t :& e_1=w_1=w_2y+a_1x=0 \rightarrow \eta^{11} , \\
C_1^{\text{s}}\cap C_2^t :& e_1=w_3=w_2y^2+ x(a_1y-w_1 st^2)=0 \rightarrow \eta^{12} , \\
C_1^{\text{s}}\cap C_3^t :& e_1=w_2=a_1y-w_1t({\widetilde{a}_2} w_3 x+st)=0 \rightarrow \eta^{13} .
\end{cases}
\end{align}
The fiber structure for this collision is an I$_6^{\text{s}}$ fiber, as depicted in Figure \ref{I2sI4s}.

In order to compute the weights for this collision of $\frak{su}(2)\times\frak{su}(4)$, we need to investigate the splittings of the curves from $C_i^{\text{s}}$ ($i=0,1$) and $C_i^t$ ($i=0,1,2,3$). We find that
\begin{align}
\begin{cases}
C_0^{\text{s}} & \rightarrow \eta^{00} ,\\
C_1^{\text{s}} & \rightarrow \eta^{10}+\eta^{11}+\eta^{12}+\eta^{13} ,\\
C_0^t & \rightarrow \eta^{00}+\eta^{10}+\eta^{10y} ,\\
C_1^t & \rightarrow \eta^{11} ,\\
C_2^t & \rightarrow \eta^{12} ,\\
C_3^t & \rightarrow \eta^{13} .
\end{cases}
\end{align}
Using linear relations, the intersection numbers between the curves and the Cartan divisors are computed below. Since $\eta^{00}$, $\eta^{11}$, $\eta^{12}$, and $\eta^{13}$ are obtained directly from $C_0^{\text{s}}$, $C_1^t$, $C_2^t$, and $C_3^t$ without modifications, the only curves in consideration are $\eta^{10}$ and $\eta^{10y}$ and we get
\begin{equation}
\begin{array}{c|rrrrrr}
 & D_0^{\text{s}} & D_1^{\text{s}} & D_0^t & D_1^t & D_2^t & D_3^t \\
 \hline
\eta^{00} & -2 & 2 & 0 & 0 & 0 & 0 \\
\eta^{10} & 1 & -1 & -1 & 1 & 0 & 0 \\
\eta^{10y} & 1 & -1 & -1 & 0 & 0 & 1 \\
\eta^{11} & 0 & 0 & 1 & -2 & 1 & 0 \\
\eta^{12} & 0 & 0 & 0 & 1 & -2 & 1 \\
\eta^{13} & 0 & 0 & 1 & 0 & 1 & -2 
\end{array}
\end{equation}

Note that for $\frak{su}(2)$, we only get the weight $[1]$, which is in the representation $\bf{2}$. On the other hand, for $\frak{su}(4)$, $\eta^{10}$ gives the weight $[-1,0,0]$ and $\eta^{10y}$ gives the weight $[0,0,-1]$. These are in representations $\bf{\bar{4}}$ and $\bf{4}$ respectively. Thus, we get the bifundamentals $(\bf{2},\bf{4})\oplus (\bf{2},\bf{\bar{4}})$ for our product group $\frak{su}(2)\times\frak{su}(4)$.

The representations with respect to $(\frak{su}(2),\frak{su}(4))$ from this I$_2^{\text{ns}}+$I$_4^{\text{s}}$-model with the Mordell-Weil group $\mathbb{Z}_2$ are summarized in Table \ref{Rep.SplitModel} below. Here we denoted weights as $[\psi ;\varphi_1,\varphi_2,\varphi_3]$, where $[\psi]$ is the weight for the $\frak{su}(2)$ and $[\varphi_1,\varphi_2,\varphi_3]$ is the weight for the $\frak{su}(4)$.

\begin{table}[H]
\begin{center}
\begin{tabular}{|c|c|c|c|c|}
\hline
Locus & \multicolumn{2}{c|}{$tw_1w_2=0$} & \multicolumn{2}{c|}{$se_1=tw_1w_2=0$}\\
\hline
Curves & $C_{13}^t$ & $C_{3\pm}^{t'}$ & $\eta^{10}$ & $\eta^{10y}$ \\
\hline
Weights & $[0;2,-1,0]$ & $[0;-1,0,1]$ & $[1;-1,0,0]$ & $[1;0,0,-1]$ \\
\hline
Representations & $(\bf{1},\bf{15})$ & $(\bf{1},\bf{6})$ & $(\bf{2},\bf{\bar{4}})$ & $(\bf{2},\bf{4})$ \\
\hline
\end{tabular}
\end{center}
\caption{Weights and representations for the I$_2^{\text{ns}}+$I$_4^{\text{s}}$-model with the Mordell-Weil group $\mathbb{Z}_2$.}
\label{Rep.SplitModel}
\end{table}

We have identified the charged matters for the product group $\frak{g}=\frak{su}(2)\times\frak{su}(4)$ to be 
\begin{equation}
\bf{R}= (\bf{3},\bf{1})\oplus(\bf{1},\bf{15})\oplus (\bf{2},\bf{4})\oplus (\bf{2},\bf{\bar{4}})\oplus(\bf{1},\bf{6}).
\end{equation} 

{  When $a_1=0$, $\eta^{10}$ and $\eta^{13}$ split and $\eta^{11}$ produces a curve that intersects with three curves. The splittings of the curves when $a_1=0$ is
\begin{align}
\label{I6sSplit}
\begin{split}
\begin{cases}
\eta^{10} \rightarrow & \ \eta^{103}+\eta^{10y} , \\
\eta^{11} \rightarrow & \ \eta^{113} , \\
\eta^{13} \rightarrow & \ \eta^{103}+\eta^{113}+\eta^{13'},
\end{cases}
\quad
\begin{cases}
\eta^{103} &: e_1=t=w_2=0, \quad
\eta^{10y} : e_1=t=y=0, \\
\eta^{113} &: e_1=w_1=w_2=0, \quad 
\eta^{13'}  : e_1=w_2={\widetilde{a}_2}w_3x+st=0.
\end{cases}
\end{split}
\end{align}
This new fiber structure when $a_1=0$ is the diagram in the third row of Figure \ref{I2sI4s}.  The intersection numbers between the new curves of this enhancement and the fibral divisors of I$_2^{\text{ns}}$ and I$_4^{\text{s}}$ are
\begin{equation}
\begin{array}{c|cccccc}
 & D_0^{\text{s}} & D_1^{\text{s}} & D_0^t & D_1^t & D_2^t & D_3^t \\
 \hline
\eta^{00} & -2 & 2 & 0 & 0 & 0 & 0 \\
\eta^{103} & 0 & 0 & 0 & 1 & 0 & -1 \\
\eta^{10y} & 1 & -1 & -1 & 0 & 0 & 1 \\
\eta^{113} & 0 & 0 & 1 & -2 & 1 & 0 \\
\eta^{12} & 0 & 0 & 0 & 1 & -2 & 1 \\
\eta^{13'} & 0 & 0 & 0 & 1 & 0 & -1 
\end{array}
\end{equation}
The weight of the curve $\eta^{103}$ is $[0;-1,0,1]$, which corresponds to the representation $(\bf{1},\bf{6})$. The curve $\eta^{10y}$ is the same as before but with a degeneracy of two, with the weight $[1;0,0,-1]$ that corresponds to the representation $(\bf{2},\bf{4})$. The curve $\eta^{113}$ carries the same weight as the curve $\eta^{11}$ earlier due to the linear relation, yielding the weight $[0;2,-1,0]$, which corresponds to the representation $(\bf{1},\bf{15})$. The weight of the curve $\eta^{13'}$ is then computed as $[0;-1,0,1]$, which corresponds to the representation $(\bf{1},\bf{6})$. Interestingly, all the new curves are in the adjoint representation of $\frak{su}(4)$ and uncharged under $\frak{su}(2)$.
Consider when $a_1={\widetilde{a}_2}=0$, which produces a codimension-four specialization. The only curve that transforms is 
\begin{equation}
\eta^{13'} \rightarrow \eta^{103}: e_1=t=w_2=0,
\end{equation}
which changes  the geometry into the bottom diagram of Figure \ref{I2sI4s}. 

\begin{figure}[H]
\begin{center}
\includegraphics[scale=1]{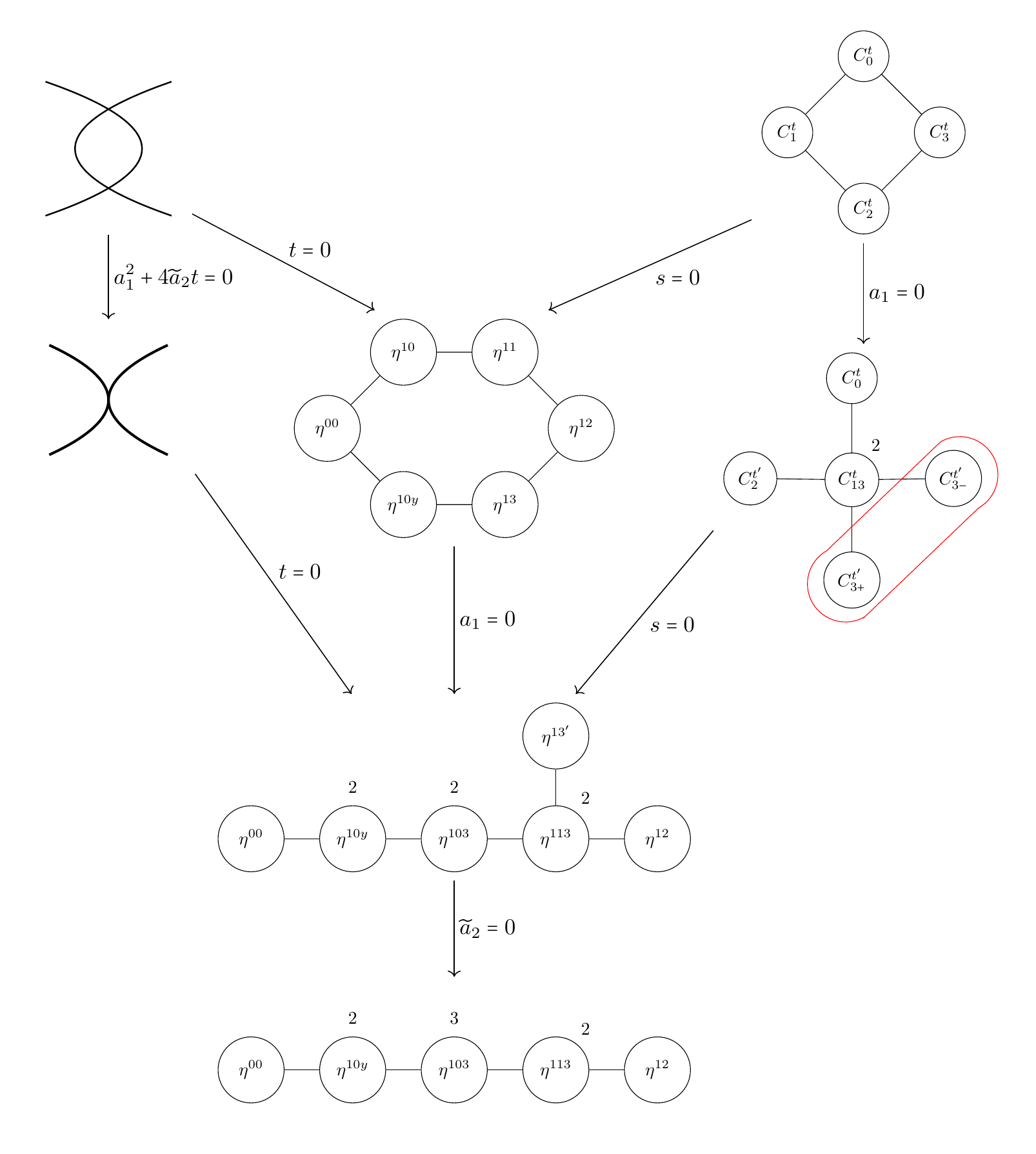}
\caption{Fiber structure of I$_2^{\text{ns}}+$I$_4^{\text{s}}$ with a Mordell-Weil group $\mathbb{Z}_2$.The diagrams on the top are the Kodaira fibers of type I$_2^{\text{ns}}$ and I$_4^{\text{s}}$. When they collide on the locus $se_1=tw_1w_2w_3=0$, we get the fiber structure of the collision of I$_2^{\text{ns}}$ and I$_4^{\text{s}}$ as a hexagon drawn in the middle. This enhancement has two newly split curves $\eta^{10}$ and $\eta^{10y}$ that give the representations $(\bf{2},\bf{4})\oplus (\bf{2},\bf{\bar{4}})$. When $a_1=0$, this hexagon further specializes to the diagram on the third row. When $a_1={\widetilde{a}_2}=0$, this further specializes into the diagram on the bottom.} 
\label{I2sI4s} 
\end{center}
\end{figure}

\subsection{Coulomb phases}
We determine the chamber structures by considering the weights that are sign indefinite. First we find the weights that constrain the chamber structures. Denote $\frak{su}(2)$ by $[\psi]$  and denote $\frak{su}(4)$ by $[\varphi_1,\varphi_2,\varphi_3]$. For $(\bf{1},\bf{6})$, there is nothing charged under $\frak{su}(2)$ and thus we only have one relation coming from the representation $\bf{6}$:
\begin{equation}
-\varphi_1+\varphi_3 .
\end{equation}
From the representation $ (\bf{2},\bf{4})$, we have the following relations:
\begin{align}
\psi-\varphi_1+\varphi_2, \quad \psi-\varphi_2+\varphi_3, \quad \psi-\varphi_3, \quad -\psi+\varphi_1, \quad -\psi-\varphi_1+\varphi_2, \quad -\psi-\varphi_2+\varphi_3.
\end{align}
Lastly, the representation $(\bf{2},\bf{\bar{4}})$ gives the following relations:
\begin{align}
\psi+\varphi_2-\varphi_3, \quad \psi+\varphi_1-\varphi_2, \quad \psi-\varphi_1, \quad -\psi+\varphi_3, \quad -\psi+\varphi_2-\varphi_3, \quad -\psi+\varphi_1-\varphi_2 .
\end{align}
Compositing all these relations, we have a total of seven independent relations that are given by
\begin{equation}
-\varphi_1+\varphi_3, \quad \psi-\varphi_1+\varphi_2, \quad \psi-\varphi_2+\varphi_3, \quad -\psi+\varphi_1, \quad -\psi-\varphi_1+\varphi_2, \quad -\psi-\varphi_2+\varphi_3, \quad -\psi+\varphi_3 .
\label{7rel}
\end{equation}
We can compute the region of the Weyl chamber from the Cartan Matrix:
\begin{center}
\begin{tabular}{c}
$\psi$ \\
\hline
2 \\
-2
\end{tabular} \quad \raisebox{-17pt}{,} \quad
\begin{tabular}{ccc}
$\varphi_1$ & $\varphi_2$ & $\varphi_3$ \\
 \hline
2 & -1 & 0 \\
-1 & 2 & -1 \\
0 & -1 & 2
\end{tabular}
\raisebox{-23pt}{.}
\end{center}
Thus the Weyl chamber is defined by the hyperplanes
\begin{equation}
\psi>0, \quad 2\varphi_1-\varphi_2>0, \quad -\varphi_1+2\varphi_2-\varphi_3>0, \quad -\varphi_2+2\varphi_3>0 .
\end{equation}

The weights corresponding to the seven entries in equation \eqref{7rel} are composited into a vector
\begin{equation}
v=(\varpi_3, \varpi_4, \varpi_5, \varpi_6, \varpi_7, \varpi_8, \varpi_9),
\end{equation}
where each weight is written in the form of a sign vector $[\psi_1;\phi_1 ,\phi_2, \phi_3]$:
\begin{align}
& \varpi_3=[0;-1,0,1], \quad \varpi_4=[1;-1,1,0], \quad \varpi_5=[1;0,-1,1], \quad \varpi_6=[-1;1,0,0], \\
&\varpi_7=[-1;-1,1,0], \quad \varpi_8=[-1;0,-1,1], \quad \varpi_9=[-1;0,0,1] ,
\end{align} 
with $1$ denoting the entry to be positive and $0$ denoting the entry to be negative. Then we get total twelve different chambers, which are denoted by  their weight vector:

\begin{equation}
\begin{tabular}{|c|c|c|c|}
\hline
 & $5ab^-\  (0110000)$ & $5ab^+\  (1110000)$ &   \\
 & $4ab^-\  (0111000)$& $4ab^+\  (1110001)$ &  \\
 & $3ab^-\  (0111001)$ &$3ab^+\  (1111001)$ &  \\
 & $2ab^-\  (0101001)$& $2ab^+\  (1111101)$ &  \\
 $1a^-\  (0001001)$ & \ $1b^-\  (0101101)$& \ $1b^+\  (1101101)$  &$1a^+\  (1111111)$   \\
 \hline
\end{tabular}.
\end{equation}\\

\begin{figure}[H]
\begin{center}
\scalebox{1}{
\begin{tikzpicture}[scale=2.2]
\coordinate (L1) at (0:0);
\coordinate (L2) at (0:5);
\coordinate (L3) at (60:5);
\coordinate (A1)  at  (barycentric cs:L1=0.7,L3=0.3);
\coordinate (A2)  at  (barycentric cs:L1=0.6,L3=0.4);
\coordinate (A3)  at  (barycentric cs:L1=0.4,L3=0.6);
\coordinate (A4)  at  (barycentric cs:L1=0.25,L3=0.75);
\coordinate (B1)  at  (barycentric cs:L2=0.7,L3=0.3);
\coordinate (B3)  at  (barycentric cs:L2=0.4,L3=0.6);
\coordinate (B4)  at  (barycentric cs:L2=0.25,L3=0.75);
\coordinate (C2)  at  (barycentric cs:L2=0.6,L1=0.4);
\coordinate (C3)  at  (barycentric cs:L2=0.4,L1=0.6);

\draw[thick]  (A1)--+(-75:1.45);
\draw[thick]  (B1)--+(255:1.45);

\node  at  (1.8,-.2) {\scalebox{1}{$\mathbf{\varpi_1}$}};
\node  at  (3.1,-.2) {\scalebox{1}{$\mathbf{\varpi_2}$}};

\draw[thick]  (L3)--+(0,-4.5);
\draw[thick]  (A3)--+(-18:2.8);
\draw[thick]  (B3)--+(198:2.8);
\draw[thick]  (C2)--+(145:2.8);
\draw[thick]  (C3)--+(35:2.8);
\draw[thick]  (L1)--(L2)--(L3)--cycle;
\node  at  (.6,.1) {\scalebox{1}{$1a^-$}};

\node  at  (2.35,.1) {\scalebox{1}{$1b^-$}};
\node  at  (2,.4) {\scalebox{1}{$2ab^-$}};
\node  at  (2,1.5) {\scalebox{1}{$3ab^-$}};
\node  at  (2,2.25) {\scalebox{1}{$4ab^-$}};
\node  at  (2,2.7) {\scalebox{1}{$5ab^-$}};

\node  at  (2.65,.1) {\scalebox{1}{$1b^+$}};
\node  at  (3.1,.4) {\scalebox{1}{$2ab^+$}};
\node  at  (3.1,1.5) {\scalebox{1}{$3ab^+$}};
\node  at  (3.1,2.25) {\scalebox{1}{$4ab^+$}};
\node  at  (3.1,2.7) {\scalebox{1}{$5ab^+$}};

\node  at  (4.4,.1) {\scalebox{1}{$1a^+$}};

\node  at  (1.2,-.2) {\scalebox{1}{$\mathbf{\varpi_4}$}};
\node  at  (2.5,-.2) {\scalebox{1}{$\mathbf{\varpi_3}$}};
\node  at  (3.8,-.2) {\scalebox{1}{$\mathbf{\varpi_8}$}};

\node  at  (.7,1.8) {\scalebox{1}{$\mathbf{\varpi_9}$}};
\node  at  (4.2,1.8) {\scalebox{1}{$\mathbf{\varpi_6}$}};
\node  at  (4.4,1.5) {\scalebox{1}{$\mathbf{\varpi_7}$}};
\node  at  (.6,1.5) {\scalebox{1}{$\mathbf{\varpi_5}$}};

\end{tikzpicture}
}
\end{center}
\label{ChamberPlaneReduced}
\caption{This is the chamber structure of the I$_2^{\text{ns}}+$I$_4^{\text{s}}$-model with a Mordell-Weil  group $\mathbb{Z}_2$ in a planar diagram.}
\end{figure}
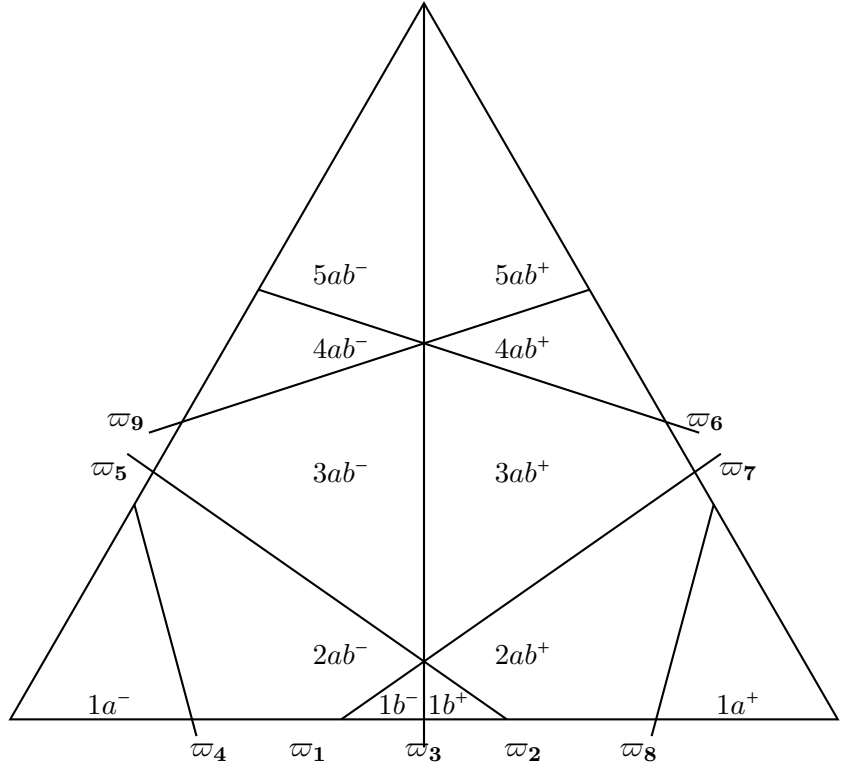

\subsection{$5d$  $\mathcal{N}=1$ prepotentials and the triple intersection polynomials}
The triple intersection polynomial is computed for the I$_2^{\text{ns}}+$I$_4^{\text{s}}$-model with the Mordell-Weil group $\mathbb{Z}_2$ in the crepant resolution that we geometrically obtained earlier. We note that the triple intersection is identical to the one we computed for the I$_2^{\text{ns}}+$I$_4^{\text{s}}$-model with a trivial Mordell-Weil group: 
\begin{align}
\begin{split}
\mathscr{F}_{trip}=&-8 \left(6 L^2-5 L T+T^2\right) \psi _1^3+4 T(L-T)\left(\phi _1^3+\phi _2^3\right) +2 T (L-2 T) \phi _3^3 \\
&+3 T (T-2 L)\phi _2^2 \left(\phi _1+\phi _3\right)+3LT\phi _2 \left(\phi _1^2+\phi _3^2\right) \\
&-6 L T \phi _1^2 \phi _3+6 L T \phi _1 \phi _2 \phi _3+12T(T-2L) \psi _1 \left(\phi _1^2-\phi _2 \phi _1+\phi _2^2+\phi _3^2-\phi _2 \phi _3\right) \\
&-8(2L-T)(3L-2T)\psi _0^3 -24(2L-T)(T-L)\psi _0^2 \psi _1 +24L(2L-T) \psi _0 \psi _1^2 
-4LT \phi _0^3  \\
&+3 T \phi _0 \left(2 \phi _1 \left(L \phi _3+2(2 L-T)\psi _1\right)+(2 T-3 L)\phi _1^2+(2 T-3 L)\phi _3^2 -4(2 L-T)\left(\psi _0-\psi _1\right)^2\right)\\
&+12T(2L-T)\phi _0\psi _1 \phi _3+(2 L-T) \phi _0^2 \left(-4 \psi _1+\phi _1+\phi _3\right).
\end{split}
\end{align}
For this split model with $\mathbb{Z}_2$, the representation is
\begin{equation}
\bf{R}=(\bf{3},\bf{1})\oplus(\bf{2},\bf{4})\oplus(\bf{2},\bf{\bar{4}})\oplus(\bf{1},\bf{6})\oplus(\bf{1},\bf{15}).
\end{equation}
Then the prepotential for this model in the chamber $5ab+$ is given as
\begin{align}
\begin{split}
6\mathscr{F}_{\text{IMS}}=& \ -8 (n_{\bf{1},\bf{15}}-1) \phi _1^3-(8 n_{\bf{1},\bf{15}}-8)\phi _2^3 -2 (4 n_{\bf{1},\bf{15}}+n_{\bf{1},\bf{6}}-4) \phi _3^3 \\
& -4 (n_{\bf{2},\bf{4}}+n_{\bf{2},\bf{\bar{4}}}+2 n_{\bf{3},\bf{1}}-2) \psi _1^3 +\frac{3}{2} (4 n_{\bf{1},\bf{15}}-2 n_{\bf{1},\bf{6}}-4) \phi _2^2 (\phi _1+\phi _3) \\
& -\frac{3}{2} (-2 n_{\bf{1},\bf{6}}) \phi _2 (\phi _1^2 + \phi _3^2) -6 n_{\bf{1},\bf{6}} \phi _1^2\phi _3 +6 n_{\bf{1},\bf{6}} \phi _1\phi _2 \phi _3 \\
&-6 (n_{\bf{2},\bf{4}}+n_{\bf{2},\bf{\bar{4}}}) \psi _1 \left(\phi _1^2 -\phi _1 \phi _2 +\phi _2^2 +\phi _3^2 -\phi _2 \phi _3 \right) .
\end{split}
\end{align}
In order to compute the number of representations $n_R$, we take the triple intersection number to be the same as the prepotential. This fixes every $n_{\bf{R}}$ except $n_{\bf{2},\bf{4}}$ and $n_{\bf{2},\bf{\bar{4}}}$, as only their sum is fixed by 
\begin{equation}
n_{\bf{2},\bf{4}}+n_{\bf{2},\bf{\bar{4}}}=ST .
\end{equation}
Using the fact that $n_{\bf{2},\bf{4}}=n_{\bf{2},\bf{\bar{4}}}$, we conclude the number of representations to be
\begin{align}
\begin{split}
&n_{\bf{3},\bf{1}}=\frac{1}{4} \left(2 L S+S^2-2 S T+4\right)=6 L^2-7 L T+2 T^2+1, \quad n_{\bf{2},\bf{4}}=n_{\bf{2},\bf{\bar{4}}}=\frac{1}{2} S T=T (2 L-T), \\
&n_{\bf{1},\bf{6}}=-3 L T+S T+2 T^2=L T, \quad n_{\bf{1},\bf{15}}=\frac{1}{2} \left(-L T+T^2+2\right) .
\end{split}
\end{align}
When the Calabi-Yau condition, $L=-K$,
\begin{equation}
n_{\bf{3},\bf{1}}=g_S, \quad n_{\bf{1},\bf{15}}=g_T.
\end{equation}

\subsection{$6d$ $\mathcal{N}=(1,0)$ anomaly cancellation}
In this section, we consider an I$_2^{\text{ns}}+$I$_4^{\text{s}}$-model with the Mordell-Weil group $\mathbb{Z}_2$. Then, the gauge algebra is given by
\begin{equation}
\frak{g}=A_1+A_3 ,
\end{equation}
and the representation is geometrically computed in section \ref{sec:splitZ2} to be 
\begin{equation}
\bf{R}= (\bf{3},\bf{1})\oplus(\bf{1},\bf{15}) \oplus (\bf{2},\bf{4})\oplus (\bf{2},\bf{\bar{4}})\oplus(\bf{1},\bf{6}).
\end{equation}
Then, the number of vector multiplets $n_V^{(6)}$, tensor multiplets $n_T$, and hypermultiplets $n_H$ are 
\begin{align}
\begin{split}
n_V^{(6)}&=18, \quad n_T=9-K^2, \\
n_H&=h^{2,1}(Y)+1+n_{\bf{3},\bf{1}}(3-1)+n_{\bf{2},\bf{4}} (8-0)+n_{\bf{2},\bf{\bar{4}}} (8-0) + n_{\bf{1},\bf{6}} (6-0) + n_{\bf{1},\bf{15}} (15-3)\\
&=17 K^2+18 K T+6 T^2+15 .
\end{split}
\end{align}
 We see that
\begin{equation}
n_H-n_V^{(6)}+29n_T-273=0, 
\end{equation}
so we can conclude that the pure gravitational anomalies are canceled.
We recall that the number of representations are
\begin{align}
\begin{split}
&n_{\bf{3},\bf{1}}=(2 K+T) (3 K+2 T)+1, \quad n_{\bf{2},\bf{4}}=n_{\bf{2},\bf{\bar{4}}}=-T (2K+T), \\
&n_{\bf{1},\bf{6}}=-K T, \quad n_{\bf{1},\bf{15}}=\frac{1}{2} (KT+T^2+2) .
\end{split}
\end{align}
We use the trace identities for $\text{SU($2$)}$ and $\text{SU($4$)}$, given by equations \eqref{eq:SU2trace} and \eqref{eq:SU4trace} to compute the remainder terms of the anomaly polynomial. We first compute the $\text{SU($2$)}$ side contribution of the anomaly polynomials. The number of representations $n_{\bf{R}}$ are then identified as
\begin{equation}
n_{\bf{3}}=n_{\bf{3},\bf{1}} , \quad n_{\bf{2}}=4 (n_{\bf{2},\bf{4}}+n_{\bf{2},\bf{\bar{4}}}).
\end{equation}
Hence, $X^{(2)}_1$ and $X^{(4)}_1$ are given by
\begin{align}
X^{(2)}_{1}&=\left(A_{\bf{3}}(1-n_{\bf{3}})-n_{\bf{2}}A_{\bf{2}}\right)\tr_{\bf{2}}F^2_1 =-12K(2K+T)\tr_{\bf{2}}F^2_1\\
\begin{split}
X^{(4)}_{1}&=\left(B_{\bf{3}}(1-n_{\bf{3}})-n_{\bf{2}}B_{\bf{2}}\right)\tr_{\bf{2}}F^4_1+\left(C_{\bf{3}}(1-n_{\bf{3}})-n_{\bf{2}}C_{\bf{2}}\right) (\tr_{\bf{2}}F^2_1)^2 \\
&=-12(2K+T)^2 (\tr_{\bf{2}}F^2_1)^2.
\end{split}
\end{align}

Now consider the contribution from the $\text{SU($4$)}$ side of the anomaly cancellation. We determine the number of representations to be
\begin{equation}
n_{\bf{4}}=2 (n_{\bf{2},\bf{4}}+n_{\bf{2},\bf{\bar{4}}}) , \quad n_{\bf{6}}=n_{\bf{6},\bf{1}}, \quad n_{\bf{15}}=n_{\bf{15},\bf{1}} .
\end{equation}
Hence, $X^{(2)}_2$ and $X^{(4)}_2$ are given by
\begin{align}
X^{(2)}_{2}&=\left(A_{\bf{15}}(1-n_{\bf{15}})-n_{\bf{6}}A_{\bf{6}}-n_{\bf{4}}A_{\bf{4}}\right)\tr_{\bf{4}}F^2_2 =6K T \tr_{\bf{4}}F^2_2\\
\begin{split}
X^{(4)}_{2}&=\left(B_{\bf{15}}(1-n_{\bf{15}})-n_{\bf{6}}B_{\bf{6}}-n_{\bf{4}}B_{\bf{4}}\right)\tr_{\bf{4}}F^4_2+\left(C_{\bf{15}}(1-n_{\bf{15}})-n_{\bf{6}}C_{\bf{6}}-n_{\bf{4}}C_{\bf{4}}\right) (\tr_{\bf{4}}F^2_2)^2 \\ &=-3T^2 (\tr_{\bf{4}}F^2_2)^2
\end{split}
\end{align}
Since we have a semisimple group with two simple components, we must include  the additional mixed term
\begin{equation}
Y_{12}=(n_{\bf{2},\bf{4}}+n_{\bf{2},\bf{\bar{4}}}) \tr_{\bf{2}}F^2_1 \tr_{\bf{4}}F^2_2
\end{equation}
to fully consider the bifundamental representation $(\bf{2},\bf{4})$. As a result, the full anomaly polynomial is given by
\begin{align}
\begin{split}
I_8&=\frac{9-n_T}{8} (\tr R^2)^2 +\frac{1}{6} (X^{(2)}_{1} +X^{(2)}_{2}) \tr R^2-\frac{2}{3} (X^{(4)}_{1}+X^{(4)}_{2})+4Y_{12} \\
&=\frac{1}{8} \left(K R-16 K \tr_{\bf{4}}F^2_2-8 T\tr_{\bf{4}}F^2_2+4 T \tr_{\bf{2}}F^2_1\right)^2,
\end{split}
\end{align}
which is a perfect square. Hence, we can conclude that the anomalies are all canceled when uplifted to a six-dimensional $\mathcal{N}=(1,0)$ theory.

\section{$\text{SU($2$)}\times \text{SU($4$)}$-model} \label{sec:splitNoZ2}

We consider an $\text{SU($2$)}\times \text{SU($4$)}$-model with a trivial Mordell-Weil group.  The Weierstrass equation for I$_2^{\text{ns}}+$I$_4^{\text{s}}$ is
\begin{equation}
y^2z+a_1xyz=x^3+{\widetilde{a}_2}tx^2z+{\widetilde{a}_4}st^2xz^2+{\widetilde{a}_6}s^2t^4 z^3 .
\end{equation}
The discriminant for this model is
\begin{align}
\begin{split}
\Delta=-s^2 t^4 &\left[
(a_1^2+4 {\widetilde{a}_2} t)^2 (a_1^2 {\widetilde{a}_6}+4 {\widetilde{a}_2} {\widetilde{a}_6} t-{\widetilde{a}_4}^2)
-8 s t^2 (9 a_1^2 {\widetilde{a}_4} {\widetilde{a}_6}+36 {\widetilde{a}_2} {\widetilde{a}_4} {\widetilde{a}_6} t-8 {\widetilde{a}_4}^3-54 {\widetilde{a}_6}^2 s t^2)\right].
\end{split}
\end{align}
The corresponding gauge group $G$ and the representation $\mathbf{R}$ are respectively 
\begin{align}
&G=\text{SU($2$)}\times \text{SU($4$)} , \quad 
\bf{R}=
(\bf{3},\bf{1})\oplus(\bf{1},\bf{15})\oplus(\bf{2},\bf{4})\oplus(\bf{2},\bf{\bar{4}})\oplus(\bf{1},\bf{6})\oplus 
(\bf{2},\bf{1})\oplus(\bf{1},\bf{4})\oplus(\bf{1},\bf{\bar{4}}).
\end{align}
The representation $\mathbf{R}$ is derived geometrically in the next subsection. 
The following sequence of blowups gives a crepant resolution:
\begin{equation}
  \begin{tikzcd}[column sep=huge] 
  X_0  \arrow[leftarrow]{r} {(x,y,s|e_1)} & \arrow[leftarrow]{r} {(x,y,t|w_1)}  X_1 &X_2  \arrow[leftarrow]{r}{(y,w_1|w_2)} &X_3 \arrow[leftarrow]{r}{(x,w_2|w_3)} & X_4.
  \end{tikzcd}
\end{equation}
Note that unlike other models, this required one more blowup to fully resolve the singularities. The proper transform is
\begin{equation}
w_2y^2+a_1xy=w_1 (e_1w_3^2x^3+{\widetilde{a}_2}w_3tx^2+{\widetilde{a}_4}st^2x+{\widetilde{a}_6}w_1w_2s^2t^4),
\end{equation}
where the relative projective coordinates are given by
\begin{equation}
[e_1w_1w_2w_3^2x : e_1w_1w_2^2w_3^2y : z=1][w_1w_2w_3^2x : w_1w_2^2w_3^2y : s][w_3x : w_2w_3y : t][y : w_1][x : w_2].
\end{equation}

\subsection{Fiber structure}
This model has the following fibral divisors that correspond to their curves: 
\begin{align}
\begin{cases}
 D_0^{\text{s}} &: s=w_2y^2+a_1xy-w_1w_3x^2(e_1w_3x+{\widetilde{a}_2}t)=0 , \\
 D_1^{\text{s}} &: e_1=w_2y^2+a_1xy-w_1t({\widetilde{a}_2}w_3x^2+{\widetilde{a}_4}stx+{\widetilde{a}_6}w_1w_2s^2t^3)=0 , \\
 D_0^t &: t=w_2y^2+a_1xy-e_1w_1w_3^2x^3=0 , \\
 D_1^t &: w_1=w_2y+a_1x=0, \\
 D_2^t &: w_3=w_2y^2+a_1xy-w_1st^2({\widetilde{a}_4}x+{\widetilde{a}_6}w_1w_2st^2)=0, \\
 D_3^t &: w_2=a_1y-w_1(e_1w_3^2x^2+{\widetilde{a}_2}w_3tx+{\widetilde{a}_4}st^2)=0 .
\end{cases}
\end{align}
 The generic fiber of $D_1^2$ is a conic that degenerates into two lines over its discriminant locus $V(a_1^2{\widetilde{a}_6}+4{\widetilde{a}_2}{\widetilde{a}_6} t-{\widetilde{a}_4}^2)$. The the resulting fiber is of type I$_3^{\text{ns}}$ and the curves 
 $C_1^{\text{s}}\to C_{1+}^{\text{s}}+C_{1-}^{\text{s}}$ yield weights of the fundamental representation of A$_1$. 
 The intersection numbers between the curves and the fibral divisors of I$_2^{\text{ns}}$ are 
\begin{equation}
C_1^{\text{s}} \rightarrow C_{1+}^{\text{s}}+C_{1-}^{\text{s}}, \hspace{2cm}
\begin{tabular}{c|cccccc}
 & $D_0^{\text{s}}$ & $D_1^{\text{s}}$ & $D_0^t$ & $D_1^t$ & $D_2^t$ & $D_3^t$ \\
 \hline
$C_0^{\text{s}}$ & -2 & 2 & 0 & 0 & 0 & 0 \\
$C_{1\pm}^{\text{s}}$ & 1 & -1 & 0 & 0 & 0 & 0 \\
$C_{1+}^{\text{s}}+C_{1-}^{\text{s}}$ & 2 & -2 & 0 & 0 & 0 & 0
\end{tabular}
\quad \raisebox{-30pt}{.}
\end{equation}
At the intersection of $S$ and $T$, we get a fiber of type I$_6^{\text{s}}$. Its components are 
\begin{align}
\begin{cases}
 C_0^t\cap C_1^t &: t=w_1=w_2y+a_1x=0, \\
 C_1^t\cap C_2^t &: w_1=w_3=w_2y+a_1x=0, \\
 C_2^t\cap C_3^t &: w_2=w_3=a_1y-{\widetilde{a}_4}w_1st^2=0, \\
 C_3^t\cap C_0^t &: t=w_2=a_1y-e_1w_1w_3^2x^2=0 .
\end{cases}
\end{align}
This specializes to an I$_0^{ss}$ when $a_1=0$, just like  the other model (with the Mordell-Weil group $\mathbb{Z}_2$) in the previous section.
The curve $C_1^t$ specializes into the central node $C_{13}^t$ where
\begin{equation}
C_{13}^t : w_1=w_2=0 .
\end{equation}
The curve $C_3^t$ splits into the three curves $C_{13}^t$, which is the central node, and $C_{3 \pm}^{t'}$, which is given by the two roots of the curve
\begin{equation}
C_{3 \pm}^{t'} : w_2=e_1w_3^2x^2+{\widetilde{a}_2}w_3tx+{\widetilde{a}_4}st^2=0 .
\end{equation}
Thus we establish a Kodaira fiber of type I$_0^{*ss}$ when $a_1=0$. This is expected as this is exactly the specialization under the same condition from the other model.

Since the locus of this specialization happens when $tw_1w_2w_3=0$ but away from the locus of $se_1=0$, there is no charged matter under $\frak{su}(2)$. Then we can compute the intersection numbers between the curves and the fibral divisors of I$_4^{\text{s}}$, from the splitting of the curves 
\begin{align}
\begin{cases}
C_1^t &\rightarrow C_{13}^t, \\
C_3^t &\rightarrow C_{13}^t+C_{3+}^{t'}+C_{3-}^{t'},
\end{cases}
\end{align}
where $C_{3\pm}^{t'}$ are of a non-split type:

\begin{center}
\begin{tabular}{c|cccccc}
 & $D_0^{\text{s}}$ & $D_1^{\text{s}}$  & $D_0^t$ & $D_1^t$ & $D_2^t$ & $D_3^t$ \\
 \hline
$C_0^t$ & 0 & 0 & -2 & 1 & 0 & 1 \\
$C_{13}^t$ & 0 & 0 & 1 & -2 & 1 & 0 \\
$C_2^t$ & 0 & 0 & 0 & 1 & -2 & 1 \\
$C_{3+}^{t'}+C_{3-}^{t'}$ & 0 & 0 & 0 & 2 & 0 & -2
\end{tabular}
\quad \raisebox{-30pt}{.}
\end{center}

For the curve $C_{13}^t$, the weight is computed to be $[2,-1,0]$, which corresponds to the representation $\bf{15}$. The curves $C_{3\pm}^{t'}$ each carries the weight $[-1,0,1]$ that yields the representation $\bf{6}$, which is the fundamental representation of $\frak{so}(6)$. Since these two curves are nonsplit, when we compute the weight of them together as $C_{3}^{t'}=C_{3+}^{t'}+C_{3-}^{t'}$, the weight is given by $[-2,0,2]$; the representation is then computed as $\bf{20'}$. Since the locus is away from $se_1=0$, it is uncharged on the side for $\frak{su}(2)$ and hence the representation for the whole product group $\frak{su}(2)\times\frak{su}(4)$ is $\bf{(1,6)}$. It is important to note that the representation we get from the specialization to an $I_0^{*ss}$ is the same for all $\bf{6}$, $\bf{15}$, and $\bf{20'}$ with the previous model (with the Mordell-Weil group $\mathbb{Z}_2$).
This has a codimension-three specialization when $a_1={\widetilde{a}_4}=0$. The curves split as
\begin{align}
\begin{cases}
C_{2}^{t} &\rightarrow C_{23}^t+C_{2\pm}^{t'}, \\
C_{3\pm}^{t'} &\rightarrow C_{23}^t+C_3^{t'},
\end{cases}
\end{align}
where the new curves are given by
\begin{align}
\begin{cases}
C_{23}^t &: w_2=w_3=0, \\
C_{2\pm}^{t'} &: w_3=y^2-{\widetilde{a}_6}w_1^2s^2t^4=0, \\
C_3^{t'} &: w_2=e_1w_3x+{\widetilde{a}_2}t=0.
\end{cases}
\end{align}
This gives the codimension three enhancement, which has a different fiber structure as presented in Figure \ref{I2sI4sNoZ2}. Using these splittings of the curve, we compute the intersection numbers between the curves and the fibral divisors of I$_4^{\text{s}}$ to be
\begin{equation}
\begin{tabular}{c|cccccc}
 & $D_0^{\text{s}}$ & $D_1^{\text{s}}$ & $D_0^t$ & $D_1^t$ & $D_2^t$ & $D_3^t$ \\
 \hline
$C_0^{t'}$ & 0 & 0 & -2 & 2 & 0 & 0 \\
$C_3^{t'}$ & 0 & 0 & 0 & 2 & 0 & -2 \\
$2 \, C_{13}^t$ & 0 & 0 & 2 & -4 & 0 & 2 \\
$2 \, C_{23}^t$ & 0 & 0 & 0 & 0 & 2 & -2 \\
$C_{2+}^{t'}+C_{2-}^{t'}$ & 0 & 0 & 0 & 0 & -2 & 2 \\
\end{tabular}
\quad \raisebox{-40pt}{.}
\end{equation}
The new curves are $C_3^{t'}$, which gives the weight $[-2,0,2]$, $C_{23}^{t}$, which gives the weight $[0,-1,1]$, and $C_{2\pm}^{t'}$, which individually gives the weight $[0,1,-1]$. They respectively correspond to the representations $\bf{20'}$, $\bf{\bar{4}}$, and $\bf{4}$. When we consider the two non-split type curves together as $C_2^{t'}=C_{2+}^{t'}+C_{2-}^{t'}$, the weight is then $[0,2,-2]$ which yields the representation $\bf{\bar{10}}$.

This model with a trivial Mordell-Weil group has an additional specialization from I$_4^{\text{s}}$ to an I$_5^{\text{s}}$ when $a_1^2 {\widetilde{a}_6}-{\widetilde{a}_4}^2=0$. Under this condition, the curve $C_2^t$ splits into two curves $C_2^{t'}$ and $C_2^{t''}$ as
\begin{align}
\begin{cases}
 C_2^t : w_3 &=w_2y^2+a_1xy-w_1st^2({\widetilde{a}_4}x+{\widetilde{a}_6}w_1w_2st^2) \\
 &=\frac{1}{a_1^2}(a_1y-{\widetilde{a}_4}w_1st^2)(a_1w_2y+a_1^2x-w_1w_2st^2)=0; \\
 C_2^{t'} : w_3 &=a_1y-{\widetilde{a}_4}w_1st^2=0, \\
 C_2^{t''} :w_3 &=a_1w_2y+a_1^2x-w_1w_2st^2=0.
\end{cases}
\end{align}
Thus, the fiber structure becomes an I$_5^{\text{s}}$, which is represented in Figure \ref{I2sI4sNoZ2}.

Based on the splitting of the curve $C_2^t$ where $C_0^t$, $C_1^t$, and $C_3^t$ remain the same, the intersection numbers between the curves and the fibral divisors of I$_4^{\text{s}}$ are computed to determine the weights of the new curves:
\begin{center}
\begin{tabular}{c|cccccc}
 & $D_0^{\text{s}}$ & $D_1^{\text{s}}$ & $D_0^t$ & $D_1^t$ & $D_2^t$ & $D_3^t$ \\
 \hline
$C_0^t$ & 0 & 0 & -2 & 1 & 0 & 1 \\
$C_1^t$ & 0 & 0 & 1 & -2 & 1 & 0 \\
$C_2^{t''}$ & 0 & 0 & 0 & 1 & -1 & 0 \\
$C_2^{t'}$ & 0 & 0 & 0 & 0 & -1 & 1 \\
$C_3^t$ & 0 & 0 & 1 & 0 & 1 & -2
\end{tabular}
\quad \raisebox{-40pt}{.}
\end{center}
The weights of the curves are  $[-1,1,0]$ for  the curve $C_2^{t'}$ and $[0,1,-1]$ for  the curve $C_2^{t''}$. They give the representations $\bf{4}$ and $\bf{\bar{4}}$ respectively. Since this collision is away from the locus of $se_1=0$, where the divisors of $\frak{su}(2)$ sit, there is no matter is charged under $\frak{su}(2)$. We therefore get $(\bf{1},\bf{4})\oplus(\bf{1},\bf{\bar{4}})$ for the charged matter.

In order to consider the collision of two types of Kodaira fibers, we need to enforce both $se_1=tw_1w_2=0$. Then we can get the following curves:
\begin{align}
\begin{cases}
C_0^{\text{s}}\cap C_0^t : & s=t=w_2y^2+a_1xy-e_1w_1w_3^2x^3=0 \rightarrow \eta^{00} , \\
C_1^{\text{s}}\cap C_0^t : & e_1=t=w_2y+a_1x=0 \rightarrow \eta^{10} , \\
& e_1=t=y=0 \rightarrow \eta^{10y} , \\
C_1^{\text{s}}\cap C_1^t : & e_1=w_1=w_2y+a_1x=0 \rightarrow \eta^{11} , \\
C_1^{\text{s}}\cap C_2^t : & e_1=w_3=w_2y^2+a_1xy-w_1st^2({\widetilde{a}_4}x+{\widetilde{a}_6}w_1w_2st^2)=0 \rightarrow \eta^{12} , \\
C_1^{\text{s}}\cap C_3^t : & e_1=w_2=a_1y-w_1t({\widetilde{a}_2}w_3x+{\widetilde{a}_4}st)=0 \rightarrow \eta^{13} .
\end{cases}
\end{align}

The fiber structure is given by the diagram in the second row of Figure \ref{I2sI4sNoZ2}. This has identical structure to the I$_2^{\text{ns}}\cap$I$_4^{\text{s}}$-model with the $\mathbb{Z}_2$  Mordell-Weil group.

The representations with respect to $(\frak{su}(2),\frak{su}(4))$ from this I$_2^{\text{ns}}+$I$_4^{\text{s}}$-model with a trivial Mordell-Weil group are summarized in Table \ref{Rep.noZ2model} below. Here we denote weights as $[\psi ;\varphi_1,\varphi_2,\varphi_3]$, where $[\psi]$ is the weight for the $\frak{su}(2)$ and $[\varphi_1,\varphi_2,\varphi_3]$ is the weight for the $\frak{su}(4)$.

\begin{table}[H]
\begin{center}
\scalebox{0.9}{\begin{tabular}{|c|c|c|c|c|c|c|c|}
\hline
Locus & \multicolumn{4}{c|}{$tw_1w_2=0$} & \multicolumn{2}{c|}{$se_1=tw_1w_2=0$} & $se_1=0$ \\
\hline
Curves & $C_{13}^t$ & $C_{3\pm}^{t'}$ &$C_2^{t'}$ &$C_2^{t''}$ & $\eta^{10}$ & $\eta^{10y}$ & $C_{1\pm}^{\text{s}}$ \\
\hline
Weights & $[0;2,-1,0]$ & $[0;-1,0,1]$ &$[0;0,1,-1]$ &$[0;-1,1,0]$ & $[1;-1,0,0]$ & $[1;0,0,-1]$ & $[-1;0,0,0]$\\
\hline
Rep & $(\bf{1},\bf{15})$ & $(\bf{1},\bf{6})$ &$(\bf{1},\bf{4})$ & $(\bf{1},\bf{\bar{4}})$ & $(\bf{2},\bf{\bar{4}})$ & $(\bf{2},\bf{4})$ & $(\bf{2},\bf{1})$ \\
\hline
\end{tabular}}
\end{center}
\caption{Weights and representations for the I$_2^{\text{ns}}+$I$_4^{\text{s}}$-model with a trivial Mordell-Weil group.}
\label{Rep.noZ2model}
\end{table}

This model has an additional $(\bf{2},\bf{1})\oplus(\bf{1},\bf{4})\oplus(\bf{1},\bf{\bar{4}})$ compared to the representations of the I$_2^{\text{ns}}+$I$_4^{\text{s}}$-model with the $\mathbb{Z}_2$ Mordell-Weil group:  
\begin{equation}
\bf{R}=
(\bf{3},\bf{1})\oplus(\bf{1},\bf{15})\oplus(\bf{2},\bf{4})\oplus(\bf{2},\bf{\bar{4}})\oplus(\bf{1},\bf{6})\oplus 
(\bf{2},\bf{1})\oplus(\bf{1},\bf{4})\oplus(\bf{1},\bf{\bar{4}}).
\end{equation}

{  When $a_1=0$, $\eta^{10}$ and $\eta^{13}$ split and $\eta^{11}$ produces a curve that intersects with three curves. The splittings of the curves when $a_1=0$ are 
\begin{align}
\label{I6sSplit}
\begin{split}
\begin{cases}
\eta^{10} \rightarrow & \ \eta^{103}+\eta^{10y} , \\
\eta^{11} \rightarrow & \ \eta^{113} , \\
\eta^{13} \rightarrow & \ \eta^{103}+\eta^{113}+\eta^{13'},
\end{cases}
\end{split}
\end{align}
where the new fibers are given by
\begin{align}
\begin{cases}
\eta^{103} &: e_1=t=w_2=0, \\
\eta^{10y} &: e_1=t=y=0, \\
\eta^{113} &: e_1=w_1=w_2=0, \\
\eta^{13'}  &: e_1=w_2={\widetilde{a}_2}w_3x+{\widetilde{a}_4}st=0.
\end{cases}
\end{align}
This new fiber structure when $a_1=0$ is the diagram in the third row of Figure \ref{I2sI4sNoZ2} as it is the same with the I$_2^{\text{ns}}+$I$_4^{\text{s}}$-model with the Mordell-Weil group $\mathbb{Z}_2$. 

Consider when $a_1={\widetilde{a}_2}=0$, which produces a codimension-four specialization. The only curve that transforms is 
\begin{equation}
\eta^{13'} \rightarrow \eta^{103}: e_1=t=w_2=0,
\end{equation}
thus changing the geometry into the bottom left diagram of Figure \ref{I2sI4sNoZ2}. 

Consider when $a_1={\widetilde{a}_4}=0$, which produces a codimension-four specialization. The only curve that transforms is 
\begin{align}
&\eta^{13'} \rightarrow \eta^{123}, \ \text{and} \ \eta^{13'} \rightarrow \eta^{123}+\eta^{12'}+\eta^{12''},
\end{align}
where the three new curves are given by
\begin{align}
\eta^{123} : e_1=w_2=w_3=0, \quad
\eta^{12'} : e_1=w_3=y+w_1st^2=0, \quad
\eta^{12''} : e_1=w_3=y-w_1st^2=0 .
\end{align}
thus changing the geometry into the bottom right diagram of Figure \ref{I2sI4sNoZ2}. }

\begin{figure}[H]
\begin{center}
\vspace{-1cm}
\lapbox{-1.2cm}
{\includegraphics[scale=1]{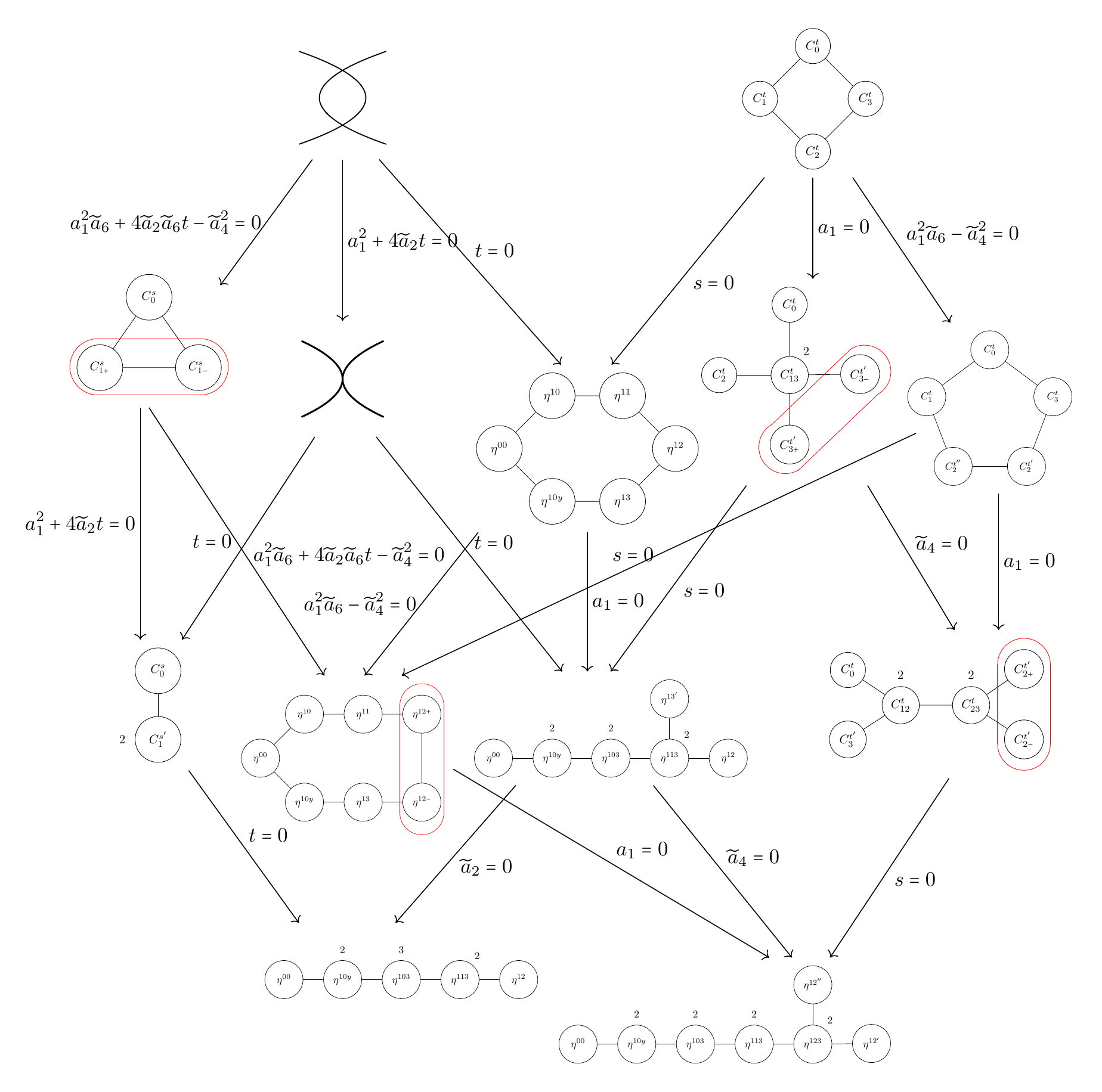}}
\caption{  Fiber structure of I$_2^{\text{ns}}+$I$_4^{\text{ns}}$ with a trivial Mordell-Weil group. This is the fiber structure of the collision of  I$_2^{\text{ns}}$ and I$_4^{\text{s}}$, which are drawn on the top, with a trivial Mordell-Weil group. The hexagon on the locus $se_1=tw_1w_2w_3=0$ from the I$_2^{\text{ns}}+$I$_4^{\text{s}}$-model with the Mordell-Weil group $\mathbb{Z}_2$ is identical to the hexagon for this model. This enhancement has two newly split curves $\eta^{10}$ and $\eta^{10y}$, giving the representations $(\bf{2},\bf{4})\oplus (\bf{2},\bf{\bar{4}})$. When $a_1=0$, this hexagon specializes to the diagram on the third row. When $a_1={\widetilde{a}_2}=0$, this further specializes into the diagram on the bottom left. When $a_1={\widetilde{a}_4}=0$, we get the diagram on the bottom right.} 
\label{I2sI4sNoZ2}
\end{center}
\end{figure}

\subsection{Coulomb phases}
Using the notation of $[\psi; \varphi_1,\varphi_2,\varphi_3]$, the nine independent weights for these representations are given by
\begin{align}
\begin{split}
&\varpi_1=[0;-1,1,0], \quad \varpi_2=[0;0,1,-1], \quad \varpi_3=[0;-1,0,1], \\
&\varpi_4=[1;-1,1,0], \quad \varpi_5=[1;0,-1,1], \quad \varpi_6=[-1;1,0,0], \\
&\varpi_7=[-1;-1,1,0], \quad \varpi_8=[-1;0,-1,1], \quad \varpi_9=[-1;0,0,1] . 
\end{split}
\end{align}
The first two weights are from  $(\bf{1},\bf{4})\oplus(\bf{1},\bf{\bar{4}})$, the third weight is from  $(\bf{1},\bf{6})$, and the remaining  six weights are from $(\bf{2},\bf{4})\oplus(\bf{2},\bf{\bar{4}})$. More specifically, 
the first weight is $(\bf{1},\bf{4})$, the second weight is $(\bf{1},\bf{\bar{4}})$,the third weight is $(\bf{1},\bf{6})$, the fourth to eighth weights are $(\bf{2},\bf{4})$, and the ninth weight is $(\bf{2},\bf{\bar{4}})$. 

For convenience, we denote these relations as a vector of weights $v$,
\begin{align}
\begin{split}
v
=&(-\varphi_1+\varphi_2, \ \varphi_2-\varphi_3, \ -\varphi_1+\varphi_3, \ \psi_1 -\varphi_1+\varphi_2, \ \psi_1 -\varphi_2+\varphi_3,  \\
& -\psi_1 +\varphi_1, \ -\psi_1 -\varphi_1+\varphi_2, \ -\psi_1 -\varphi_2+\varphi_3, \ -\psi_1+\varphi_3) .
\end{split}
\end{align}

The chambers structures of these are computed, which is drawn in Figure \ref{ChambersNoZ2}. There are in total $20$ chambers. The chambers are denoted as a set of signs of the nine relations that determine the chamber, where $1$ denotes the relation to be positive and $-1$ denotes the relation to be negative.

\begin{figure}[hbt]
\begin{center}
\scalebox{.95}{
\begin{tikzpicture}[scale=2.2]
\coordinate (L1) at (0:0);
\coordinate (L2) at (0:5);
\coordinate (L3) at (60:5);
\coordinate (A1)  at  (barycentric cs:L1=0.7,L3=0.3);
\coordinate (A2)  at  (barycentric cs:L1=0.6,L3=0.4);
\coordinate (A3)  at  (barycentric cs:L1=0.4,L3=0.6);
\coordinate (A4)  at  (barycentric cs:L1=0.25,L3=0.75);
\coordinate (B1)  at  (barycentric cs:L2=0.7,L3=0.3);
\coordinate (B3)  at  (barycentric cs:L2=0.4,L3=0.6);
\coordinate (B4)  at  (barycentric cs:L2=0.25,L3=0.75);
\coordinate (C2)  at  (barycentric cs:L2=0.6,L1=0.4);
\coordinate (C3)  at  (barycentric cs:L2=0.4,L1=0.6);

\draw[thick]  (A1)--+(-75:1.45);
\draw[thick]  (B1)--+(255:1.45);

\draw[thick, color=blue]  (A4)--+(0,-3.4);
\draw[thick, color=blue]  (B4)--+(0,-3.4);
\node  at  (1.8,-.2) {\scalebox{1}{$\mathbf{\varpi_1}$}};
\node  at  (3.1,-.2) {\scalebox{1}{$\mathbf{\varpi_2}$}};

\draw[thick]  (L3)--+(0,-4.5);
\draw[thick]  (A3)--+(-18:2.8);
\draw[thick]  (B3)--+(198:2.8);
\draw[thick]  (C2)--+(145:2.8);
\draw[thick]  (C3)--+(35:2.8);
\draw[thick]  (L1)--(L2)--(L3)--cycle;
\node  at  (.6,.1) {\scalebox{1}{$1a^-$}};
\node  at  (1.7,.4) {\scalebox{1}{$2a^-$}};
\node  at  (1.7,1.5) {\scalebox{1}{$3a^-$}};
\node  at  (1.7,2.25) {\scalebox{1}{$4a^-$}};
\node  at  (1.75,2.65) {\scalebox{1}{$5a^-$}};

\node  at  (2.35,.1) {\scalebox{1}{$1b^-$}};
\node  at  (2.2,.4) {\scalebox{1}{$2b^-$}};
\node  at  (2.2,1.5) {\scalebox{1}{$3b^-$}};
\node  at  (2.1,2.25) {\scalebox{1}{$4b^-$}};
\node  at  (2.2,2.7) {\scalebox{1}{$5b^-$}};

\node  at  (2.65,.1) {\scalebox{1}{$1b^+$}};
\node  at  (2.9,.4) {\scalebox{1}{$2b^+$}};
\node  at  (2.9,1.5) {\scalebox{1}{$3b^+$}};
\node  at  (2.9,2.25) {\scalebox{1}{$4b^+$}};
\node  at  (2.9,2.7) {\scalebox{1}{$5b^+$}};

\node  at  (4.4,.1) {\scalebox{1}{$1a^+$}};
\node  at  (3.4,.4) {\scalebox{1}{$2a^+$}};
\node  at  (3.4,1.5) {\scalebox{1}{$3a^+$}};
\node  at  (3.4,2.25) {\scalebox{1}{$4a^+$}};
\node  at  (3.3,2.65) {\scalebox{1}{$5a^+$}};

\node  at  (1.2,-.2) {\scalebox{1}{$\mathbf{\varpi_4}$}};
\node  at  (2.5,-.2) {\scalebox{1}{$\mathbf{\varpi_3}$}};
\node  at  (3.8,-.2) {\scalebox{1}{$\mathbf{\varpi_8}$}};

\node  at  (.7,1.8) {\scalebox{1}{$\mathbf{\varpi_9}$}};
\node  at  (4.2,1.8) {\scalebox{1}{$\mathbf{\varpi_6}$}};
\node  at  (4.4,1.5) {\scalebox{1}{$\mathbf{\varpi_7}$}};
\node  at  (.6,1.5) {\scalebox{1}{$\mathbf{\varpi_5}$}};

\end{tikzpicture}
}
\end{center}
\label{ChamberPlane}
\caption{{ This is the chamber structure of the I$_2^{\text{ns}}+$I$_4^{\text{s}}$-model with a trivial Mordell-Weil group in a planar diagram.}}
\end{figure}
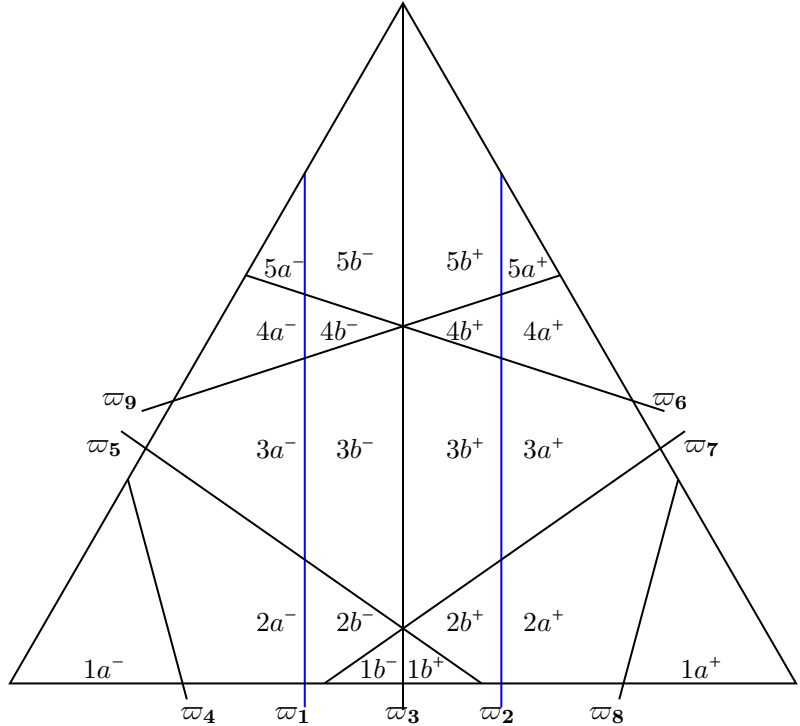

\clearpage
\newpage

\subsection{$5d$  $\mathcal{N}=1$ prepotentials and the triple intersection polynomials}
The triple intersection polynomial is computed for the I$_2^{\text{ns}}+$I$_4^{\text{s}}$-model with a trivial Mordell-Weil group defined by the crepant resolution we determined  earlier  [REF]:
\begin{align}
\begin{split}
\mathscr{F}_{trip}=&-2 S (2 L+S)\psi _1^3 +4 T (L-T)\phi _1^3 +2 T (S-2 L)\phi _2^3 +2 T (L-2 T)\phi _3^3  \\
&+3 T (2 L-S-T) \phi _2^2 \left(\phi _1+\phi _3\right)+3 T (-3 L+S+2 T) \phi _2 \left(\phi _1^2+\phi _3^2\right) \\
&-6 L T \phi _1^2 \phi _3+6 L T \phi _1 \phi _2 \phi _3-6 S T \psi _1 \left(\phi _1^2-\phi _2 \phi _1+\phi _2^2+\phi _3^2-\phi _2 \phi _3\right) \\
&+4 S(L-S) \psi _0^3 +6 S(S-2 L) \psi _0^2 \psi _1 +12 L S \psi _0 \psi _1^2 -2 T(-2 L+S+2 T) \phi _0^3  \\
&+3 T \phi _0 \left(2 \phi _1 \left(L \phi _3+S \psi _1\right)+\phi _1^2 (L-S)+\phi _3^2 (L-S)-2 S \left(\psi _0-\psi _1\right)^2+2 S \psi _1 \phi _3\right) \\
&+\phi _0^2 \left(3 T \left(\phi _1+\phi _3\right) (-2 L+S+T)-6 S T \psi _1\right) .
\end{split}
\end{align}
The bottom three lines have terms depending on $\psi_0$ and $\phi_0$, which will not be contributing to the prepotential term that is computed below, and the prepotential is compared with the triple intersection polynomial explicitly. The triple intersection number is different for different crepant resolutions, so the correct chamber has to be identified.

The Intrilligator-Morrison-Seiberg prepotential is computed in the same chamber that corresponds to $5b+$ in Figure \ref{ChambersNoZ2}. This is determined from the weights computed from the curves, which is summarized in Table \ref{Rep.noZ2model}. The chamber that we geometrically computed has positive relations for the first three entries of the sign vector $v$ in Figure \ref{ChambersNoZ2}. Also, the sixth entry is negative and the ninth entry is negative in the sign vector $v$ as well. The only chamber that satisfies $v=(111xx0xx0)$, where $x$ denotes an unknown sign, is the chamber $5b+$, which has $v=(111110000)$. Using the representation for this split model with $\mathbb{Z}_2$, 
\begin{equation}
\bf{R}=(\bf{3},\bf{1}\oplus(\bf{1},\bf{15}))\oplus(\bf{2},\bf{4})\oplus(\bf{2},\bf{\bar{4}})\oplus(\bf{1},\bf{6})\oplus(\bf{2},\bf{1})\oplus(\bf{1},\bf{4})\oplus(\bf{1},\bf{\bar{4}}) ,
\end{equation}
the prepotential in the chamber $5b+$ is 
\begin{align}
\begin{split}
6 \mathscr{F}_{\text{IMS}}=& \ -8 (n_{\bf{1},\bf{15}}-1) \phi _1^3-(8 n_{\bf{1},\bf{15}}+n_{\bf{1},\bf{4}}+n_{\bf{1},\bf{\bar{4}}}-8)\phi _2^3 -2 (4 n_{\bf{1},\bf{15}}+n_{\bf{1},\bf{6}}-4) \phi _3^3 \\
& -(n_{\bf{2},\bf{1}}+4 (n_{\bf{2},\bf{4}}+n_{\bf{2},\bf{\bar{4}}}+2 n_{\bf{3},\bf{1}}-2)) \psi _1^3 +\frac{3}{2} (4 n_{\bf{1},\bf{15}}+n_{\bf{1},\bf{4}}+n_{\bf{1},\bf{\bar{4}}}-2 n_{\bf{1},\bf{6}}-4) \phi _2^2 (\phi _1+\phi _3) \\
& -\frac{3}{2} (n_{\bf{1},\bf{4}}+n_{\bf{1},\bf{\bar{4}}}-2 n_{\bf{1},\bf{6}}) \phi _2 (\phi _1^2 + \phi _3^2) -6 n_{\bf{1},\bf{6}} \phi _1^2\phi _3 +6 n_{\bf{1},\bf{6}} \phi _1\phi _2 \phi _3 \\
&-6 (n_{\bf{2},\bf{4}}+n_{\bf{2},\bf{\bar{4}}}) \psi _1 \left(\phi _1^2 -\phi _1 \phi _2 +\phi _2^2 +\phi _3^2 -\phi _2 \phi _3 \right) .
\end{split}
\end{align}
Using the dictionary that the triple intersection polynomial is identical to the prepotential, we can fix $n_{\bf{1},\bf{6}}$ and $n_{\bf{1},\bf{15}}$ completely, and determine the following linear relations for  the remainder $n_{\bf{R}}$:
\begin{align}
\begin{split}
&n_{\bf{1},\bf{4}}+n_{\bf{1},\bf{\bar{4}}}=2T(4 L-S-2 T), \  n_{\bf{2},\bf{4}}+n_{\bf{2},\bf{\bar{4}}}=ST, \  n_{\bf{2},\bf{1}}+8 n_{\bf{3},\bf{1}}=2 S (2 L + S - 2 T)+8 .
\end{split}
\end{align}
In order to compute individual $n_{\bf{R}}$, we recall that the representation $(\bf{2},\bf{1})$ was computed from the splittings of the curve $C_1^{\text{s}} \rightarrow C_{1+}^{\text{s}}+C_{1-}^{\text{s}}$ in section \ref{sec:splitNoZ2}. These  splittings are from the condition of the coefficients such that $a_1^2 {\widetilde{a}_6}+4{\widetilde{a}_2}{\widetilde{a}_6}t-{\widetilde{a}_4}^2=0$. The class of this condition is then twice that of the class of ${\widetilde{a}_4}$, which is $2 (4L-S-2T)$. Then these  splittings happen when $s=a_1^2 {\widetilde{a}_6}+4{\widetilde{a}_2}{\widetilde{a}_6}t-{\widetilde{a}_4}^2=0$, which yields the class $2S (4L-S-2T)$. Hence, $n_{\bf{2},\bf{1}}=2S (4L-S-2T)$.  Moreover, we also use $n_{\bf{1},\bf{4}}=n_{\bf{1},\bf{\bar{4}}}$ and $n_{\bf{2},\bf{4}}=n_{\bf{2},\bf{\bar{4}}}$. Thus, we can now fix all the $n_{\bf{R}}$ to be
\begin{align}
\begin{split}
&n_{\bf{3},\bf{1}}=\frac{1}{2} (-L S + S^2 + 2), \quad n_{\bf{2},\bf{1}}=2 S (4 L - S - 2 T), \quad n_{\bf{2},\bf{4}}=n_{\bf{2},\bf{\bar{4}}}=\frac{1}{2} S T, \\
&n_{\bf{1},\bf{6}}=L T, \quad n_{\bf{1},\bf{15}}=\frac{1}{2} \left(-L T+T^2+2\right), \quad n_{\bf{1},\bf{4}}=n_{\bf{1},\bf{\bar{4}}}=T (4 L-S-2 T) .
\end{split}
\end{align}
Using the genus of the curve, let $g_S$ and $g_T$ be the curves such that $2-2g_S=S(L-S)$ and $2-2g_T=T(L-T)$. Also we use the Calabi-Yau condition $L=-K$. Then we can interpret the prepotential in terms of these genus as
\begin{align}
\begin{split}
\mathscr{F}_{trip}=& -(8-8 g_S+6 S^2)\psi _1^3 +(8-8 g_T)\phi _1^3 +(-8+8 g_T+2 T (S-2 T))\phi _2^3 -2(2 g_T+T^2-2)\phi _3^3  \\
&+3 \left(-4 g_T-S T+T^2+4\right) \phi _2^2 \left(\phi _1+\phi _3\right)+3 \left(6 g_T+T (S-T)-6\right)\phi _2 \left(\phi _1^2+\phi _3^2\right) \\
&-6 \left(-2 g_T+T^2+2\right) \phi _1^2 \phi _3+6 \left(-2 g_T+T^2+2\right)\phi _1 \phi _2 \phi _3\\
&-6 S T \psi _1 \left(\phi _1^2-\phi _2 \phi _1+\phi _2^2+\phi _3^2-\phi _2 \phi _3\right)  .
\end{split}
\end{align}

\subsection{$6d$ $\mathcal{N}=(1,0)$ anomaly cancellation}
{ 
In this section, we consider an I$_2^{\text{ns}}+$I$_4^{\text{s}}$-model with the trivial Mordell-Weil group. Then, the gauge algebra and the representations, which are  computed geometrically in section \ref{sec:splitNoZ2}, are given by}
\begin{align}
\frak{g}=&A_1+A_3 , \quad \bf{R}=& (\bf{3},\bf{1})\oplus (\bf{2},\bf{4})\oplus (\bf{2},\bf{\bar{4}})\oplus (\bf{1},\bf{4})\oplus (\bf{1},\bf{\bar{4}})\oplus(\bf{1},\bf{6})\oplus(\bf{1},\bf{15}) .
\end{align}
The number of representations are computed above:
\begin{align}
\begin{split}
&n_{\bf{3},\bf{1}}=\frac{1}{2} (K S + S^2 + 2)=g_S, \quad n_{\bf{2},\bf{1}}=-2 S (4K+S+2 T), \quad n_{\bf{2},\bf{4}}=n_{\bf{2},\bf{\bar{4}}}=\frac{1}{2} S T\\
&n_{\bf{1},\bf{6}}=-KT, \quad n_{\bf{1},\bf{15}}=\frac{1}{2} \left(KT+T^2+2\right)=g_T, \quad n_{\bf{1},\bf{4}}=n_{\bf{1},\bf{\bar{4}}}=-T(4K+S+2T) .
\end{split}
\label{nRwithK}
\end{align}
Then, the number of vector multiplets $n_V^{(6)}$, tensor multiplets $n_T$, and hypermultiplets $n_H$ are 
\begin{align}
\begin{split}
n_V^{(6)}&=18, \quad n_T=9-K^2, \\
n_H&=h^{2,1}(Y)+1+n_{\bf{3},\bf{1}}(3-1)+n_{\bf{2},\bf{1}} (2-0)+n_{\bf{2},\bf{4}} (8-0)+n_{\bf{2},\bf{\bar{4}}} (8-0) + n_{\bf{1},\bf{6}} (6-0) \\
&\quad + n_{\bf{1},\bf{15}} (15-3)+n_{\bf{1},\bf{4}} (4-0)+n_{\bf{1},\bf{\bar{4}}} (4-0)  =30 + 29 K^2 .
\end{split}
\end{align}
Then we see that
\begin{equation}
n_H-n_V^{(6)}+29n_T-273=0,
\end{equation}
which vanishes. The pure gravitational anomalies are thus canceled. 

In order to check the remainder terms of the anomaly polynomial, the number of representations are then identified as
\begin{align}
\begin{split}
n_{\bf{3}}=n_{\bf{3},\bf{1}} , \quad n_{\bf{2}}=n_{\bf{2},\bf{1}}+4 & (n_{\bf{2},\bf{4}}+n_{\bf{2},\bf{\bar{4}}}), \\
n_{\bf{4}}=n_{\bf{1},\bf{4}}+n_{\bf{1},\bf{\bar{4}}}+2(n_{\bf{2},\bf{4}}+n_{\bf{2},\bf{\bar{4}}}) , &\quad n_{\bf{6}}=n_{\bf{6},\bf{1}}, \quad n_{\bf{15}}=n_{\bf{15},\bf{1}} .
\end{split}
\end{align}
Hence, using the trace identities os $\text{SU($2$)}$ and $\text{SU($4$)}$ given by the equations \eqref{eq:SU2trace} and \eqref{eq:SU4trace}, $X^{(2)}_1$, $X^{(4)}_1$, $X^{(2)}_2$, $X^{(4)}_2$, and $Y_{ab}$ are given by
\begin{align}
X^{(2)}_{1}&=\left(A_{\bf{3}}(1-n_{\bf{3}})-n_{\bf{2}}A_{\bf{2}}\right)\tr_{\bf{2}}F^2_1 =6KT\tr_{\bf{2}}F^2_1\\
\begin{split}
X^{(4)}_{1}&=\left(B_{\bf{3}}(1-n_{\bf{3}})-n_{\bf{2}}B_{\bf{2}}\right)\tr_{\bf{2}}F^4_1+\left(C_{\bf{3}}(1-n_{\bf{3}})-n_{\bf{2}}C_{\bf{2}}\right) (\tr_{\bf{2}}F^2_1)^2 \\
&=-3T^2 (\tr_{\bf{2}}F^2_1)^2
\end{split} \\
X^{(2)}_{2}&=\left(A_{\bf{15}}(1-n_{\bf{15}})-n_{\bf{6}}A_{\bf{6}}-n_{\bf{4}}A_{\bf{4}}\right)\tr_{\bf{4}}F^2_2 =6KS\tr_{\bf{4}}F^2_2 \\
\begin{split}
X^{(4)}_{2}&=\left(B_{\bf{15}}(1-n_{\bf{15}})-n_{\bf{6}}B_{\bf{6}}-n_{\bf{4}}B_{\bf{4}}\right)\tr_{\bf{4}}F^4_2+\left(C_{\bf{15}}(1-n_{\bf{15}})-n_{\bf{6}}C_{\bf{6}}-n_{\bf{4}}C_{\bf{4}}\right) (\tr_{\bf{4}}F^2_2)^2 \\ &=-3S^2 (\tr_{\bf{4}}F^2_2)^2
\end{split} \\
Y_{12}&=(n_{\bf{2},\bf{4}}+n_{\bf{2},\bf{\bar{4}}}) \tr_{\bf{2}}F^2_1 \tr_{\bf{4}}F^2_2 =ST\tr_{\bf{2}}F^2_1 \tr_{\bf{4}}F^2_2 .
\end{align}
The remainder terms of anomaly polynomial are  then given by
\begin{align}
\begin{split}
I_8&=\frac{9-n_T}{8} (\tr R^2)^2 +\frac{1}{6} (X^{(2)}_{1} +X^{(2)}_{2}) \tr R^2-\frac{2}{3} (X^{(4)}_{1}+X^{(4)}_{2})+4Y_{12} \\
&= \frac{1}{8} \left(K\tr R^2+4 S\tr_{\bf{2}}F^2_1+4 T\tr_{\bf{4}}F^2_2\right)^2 ,
\end{split}
\end{align}
which is a perfect square. Hence, all the six-dimensional anomalies are canceled.

\section*{Acknowledgements}
The authors are grateful to  Lara Anderson,  Richard Derryberry, Sergei Gukov, Ravi Jagadeesan,  Patrick Jefferson, Julian Salazar, Washington Taylor, and Shu-Heng Shao  for helpful discussions. 
M.E. and M.J.K.  would like to thank the American Institute of Mathematics  (AIM) for the support and hospitality during the  2017 workshop on {\em Singular Geometry and Higgs Bundles in String Theory}.
M.E. would like to thank Imane Esole for her understanding and cooperation during the final stage of this work. 
M.E. is supported in part by the National Science Foundation (NSF) grant DMS-1701635  ``Elliptic Fibrations and String Theory''.
M.J.K. would like to acknowledge a partial support from NSF grant PHY-1352084. 
M.J.K. is thankful to Daniel Jafferis for his support and guidance.



\begin{thebibliography}{10}

\bibitem{Aluffi_CBU}
P.~Aluffi.
\newblock Chern classes of blow-ups.
\newblock {\em Math. Proc. Cambridge Philos. Soc.}, 148(2):227--242, 2010.

\bibitem{AE1}
P.~Aluffi and M.~Esole.
\newblock {Chern class identities from tadpole matching in type IIB and
  F-theory}.
\newblock {\em JHEP}, 03:032, 2009.

\bibitem{AE2}
P.~Aluffi and M.~Esole.
\newblock {New Orientifold Weak Coupling Limits in F-theory}.
\newblock {\em JHEP}, 02:020, 2010.

\bibitem{Anderson:2017zfm} 
  L.~B.~Anderson,  M.~Esole, L.~Fredrickson, and L.~P.~Schaposnik,
  Singular geometry and Higgs bundles in string theory,
  arXiv:1710.08453 [math.DG].


\bibitem{Anderson:2015cqy} 
  L.~B.~Anderson, J.~Gray, N.~Raghuram and W.~Taylor,
  Matter in transition,
  JHEP {\bf 1604}, 080 (2016)
  doi:10.1007/JHEP04(2016)080
  [arXiv:1512.05791 [hep-th]].

\bibitem{Anderson:2017rpr} 
  L.~B.~Anderson, J.~J.~Heckman, S.~Katz and L.~Schaposnik,
  T-Branes at the Limits of Geometry,
  JHEP {\bf 1710}, 058 (2017)


\bibitem{Aspinwall:1996nk} 
  P.~S.~Aspinwall and M.~Gross, 
  The SO(32) heterotic string on a K3 surface,
  Phys.\ Lett.\ B {\bf 387}, 735 (1996)
  doi:10.1016/0370-2693(96)01095-7
  [hep-th/9605131].
  
  \bibitem{Aspinwall:1998xj} 
  P.~S.~Aspinwall and D.~R.~Morrison,
   Nonsimply connected gauge groups and rational points on elliptic curves, 
  JHEP {\bf 9807}, 012 (1998)
  [hep-th/9805206].
  
\bibitem{Aspinwall:2000kf} 
  P.~S.~Aspinwall, S.~H.~Katz and D.~R.~Morrison,
   Lie groups, Calabi-Yau threefolds, and F theory,
  Adv.\ Theor.\ Math.\ Phys.\  {\bf 4}, 95 (2000)
  [hep-th/0002012].



  
  \bibitem{Avramis:2005hc} 
  S.~D.~Avramis and A.~Kehagias,
  A Systematic search for anomaly-free supergravities in six dimensions, 
  JHEP {\bf 0510}, 052 (2005)
  doi:10.1088/1126-6708/2005/10/052
  [hep-th/0508172].



\bibitem{Batyrev.Betti}
V.~V. Batyrev.
\newblock Birational {C}alabi-{Y}au {$n$}-folds have equal {B}etti numbers.
\newblock In {\em New trends in algebraic geometry ({W}arwick, 1996)}, volume
  264 of {\em London Math. Soc. Lecture Note Ser.}, pages 1--11. Cambridge
  Univ. Press, Cambridge, 1999.
%
\bibitem{Baume:2017hxm} 
  F.~Baume, M.~Cvetic, C.~Lawrie and L.~Lin,
  When Rational Sections Become Cyclic: Gauge Enhancement in F-theory via Mordell--Weil Torsion,
  arXiv:1709.07453 [hep-th].


\bibitem{Bershadsky:1996nu} 
  M.~Bershadsky and A.~Johansen,
  Colliding singularities in F theory and phase transitions,
  Nucl.\ Phys.\ B {\bf 489}, 122 (1997)
  doi:10.1016/S0550-3213(97)00027-8
  [hep-th/9610111].
  
  \bibitem{Bonetti:2011mw} 
  F.~Bonetti and T.~W.~Grimm,
  ``Six-dimensional (1,0) effective action of F-theory via M-theory on Calabi-Yau threefolds,''
  JHEP {\bf 1205}, 019 (2012)

\bibitem{Bourbaki.GLA79} N.~Bourbaki, {\it Groups and Lie Algebras. Chap. 7-9.}, 
Translated from the 1975 and 1982 French originals by A.~Pressley. Elements of Mathematics (Berlin).
Springer-Verlag, Berlin, 2005.   


\bibitem{Cadavid:1995bk} 
  A.~C.~Cadavid, A.~Ceresole, R.~D'Auria and S.~Ferrara,
  Eleven-dimensional supergravity compactified on Calabi-Yau threefolds,
  Phys.\ Lett.\ B {\bf 357}, 76 (1995)
  

%
\bibitem{Cremmer:1978km} 
  E.~Cremmer, B.~Julia and J.~Scherk,
   Supergravity Theory in Eleven-Dimensions,
  Phys.\ Lett.\  {\bf 76B}, 409 (1978).

	
\bibitem{Formulaire}
P.~Deligne.
\newblock Courbes elliptiques: formulaire d'apr{\`e}s {J}. {T}ate.
\newblock In {\em Modular functions of one variable, {IV} ({P}roc. {I}nternat.
  {\text{s}}ummer {\text{s}}chool, {U}niv. {A}ntwerp, {A}ntwerp, 1972)}, pages 53--73.
  Lecture Notes in Math., Vol. 476. Springer, Berlin, 1975.
  
  \bibitem{Dixon:1986jc}
L.~J. Dixon, J.~A. Harvey, C.~Vafa, and E.~Witten.
\newblock {Strings on Orbifolds. 2.}
\newblock {\em Nucl. Phys.}, B274:285--314, 1986.

  \bibitem{Erler:1993zy} 
  J.~Erler,
   Anomaly cancellation in six-dimensions, 
  J.\ Math.\ Phys.\  {\bf 35}, 1819 (1994)
  doi:10.1063/1.530885
  [hep-th/9304104].

  \bibitem{Esole.Elliptic} 
  M.~Esole,
   Introduction to Elliptic Fibrations, 
  Math.\ Phys.\ Stud.\  {\bf 9783319654270}, 247 (2017).

  \bibitem{EJJN1}
M.~Esole, S.~G. Jackson, R.~Jagadeesan, and A.~G. No{\"e}l.
\newblock {Incidence Geometry in a Weyl Chamber I: GL$_n$}, 
\newblock arXiv:1508.03038 [math.RT].



\bibitem{EJJN2} 
  M.~Esole, S.~G.~Jackson, R.~Jagadeesan and A.~G.~No{\"e}l,
  Incidence Geometry in a Weyl Chamber II: $SL_n$,
  arXiv:1601.05070 [math.RT].
  
  
  
\bibitem{G2} 
  M.~Esole, R.~Jagadeesan and M.~J.~Kang,
  The Geometry of G$_2$, Spin(7), and Spin(8)-models,
  arXiv:1709.04913 [hep-th].


\bibitem{Euler} 
  M.~Esole, P.~Jefferson and M.~J.~Kang,
  Euler Characteristics of Crepant Resolutions of Weierstrass Models,
  arXiv:1703.00905 [math.AG].

\bibitem{Esole:2014dea} 
  M.~Esole, M.~J.~Kang and S.~T.~Yau,
  A New Model for Elliptic Fibrations with a Rank One Mordell-Weil Group: I. Singular Fibers and Semi-Stable Degenerations,
  arXiv:1410.0003 [hep-th].

\bibitem{F4} 
  M.~Esole, P.~Jefferson and M.~J.~Kang,
  The Geometry of F$_4$-Models,
  arXiv:1704.08251 [hep-th].
  
  \bibitem{Esole:2012tf} 
  M.~Esole and R.~Savelli,
  Tate Form and Weak Coupling Limits in F-theory,
  JHEP {\bf 1306}, 027 (2013)

\bibitem{ES} 
  M.~Esole and S.~H.~Shao,
   M-theory on Elliptic Calabi-Yau Threefolds and $6d$ Anomalies, 
  arXiv:1504.01387 [hep-th].


\bibitem{ESY1}
M.~Esole, S.-H. Shao, and S.-T. Yau.
\newblock {Singularities and Gauge Theory Phases}.
\newblock {\em Adv. Theor. Math. Phys.}, 19:1183--1247, 2015.

\bibitem{ESY2} 
  M.~Esole, S.~H.~Shao and S.~T.~Yau,
  Singularities and Gauge Theory Phases II,
  Adv.\ Theor.\ Math.\ Phys.\  {20}, 683 (2016)



\bibitem{EY} 
  M.~Esole and S.~T.~Yau,
  Small resolutions of SU(5)-models in F-theory,
  Adv.\ Theor.\ Math.\ Phys.\   {17}, no. 6, 1195 (2013)
  
  \bibitem{SO4}
  M.~Esole and M.~J.~Kang, SO(4)-models, To appear. 
  
  \bibitem{MP}
M.~Esole and P.~Jefferson.
\newblock To appear.
  
\bibitem{Ferrara:1996wv} 
  S.~Ferrara, R.~Minasian and A.~Sagnotti,
   Low-energy analysis of M and F theories on Calabi-Yau threefolds, 
  Nucl.\ Phys.\ B {\bf 474}, 323 (1996)
  doi:10.1016/0550-3213(96)00268-4
  [hep-th/9604097].

\bibitem{Fullwood:SVW}
J.~Fullwood.
\newblock {On generalized Sethi-Vafa-Witten formulas}.
\newblock {\em J. Math. Phys.}, 52:082304, 2011.



\bibitem{Green:1984bx} 
  M.~B.~Green, J.~H.~Schwarz and P.~C.~West,
   Anomaly Free Chiral Theories in Six-Dimensions, 
  Nucl.\ Phys.\ B {\bf 254}, 327 (1985).
  
  
  
\bibitem{Grimm:2011fx} 
  T.~W.~Grimm and H.~Hayashi,
  F-theory fluxes, Chirality and Chern-Simons theories,
  JHEP {\bf 1203}, 027 (2012)

  
  \bibitem{Grimm:2015zea} 
  T.~W.~Grimm and A.~Kapfer,
  Anomaly Cancelation in Field Theory and F-theory on a Circle,
  JHEP {\bf 1605}, 102 (2016)
  doi:10.1007/JHEP05(2016)102
  [arXiv:1502.05398 [hep-th]].
  

\bibitem{Hayashi:2014kca}
H.~Hayashi, C.~Lawrie, D.~R. Morrison, and S.~Sch\"afer-Nameki.
\newblock {Box Graphs and Singular Fibers}.
\newblock {\em JHEP}, 1405:048, 2014.


\bibitem{Humphreys}
J.~Humphreys, {\it Introduction to Lie Algebras and Representation Theory}, Graduate Texts in Mathematics 9, Springer-Verlag, New York, 1972. 

\bibitem{Husemoller} D.~Husem\"oller, {\it Elliptic curves}. Second edition. With appendices by Otto Forster, Ruth Lawrence and Stefan Theisen. Graduate Texts in Mathematics, 111. Springer-Verlag, New York, 2004. 


\bibitem{IMS}
K.~A. Intriligator, D.~R. Morrison, and N.~Seiberg.
\newblock {Five-dimensional supersymmetric gauge theories and degenerations of
  Calabi-Yau spaces}.
\newblock {\em Nucl.Phys.}, B497:56--100, 1997.



\bibitem{Kodaira}
K.~Kodaira.
\newblock On compact analytic surfaces. {II}, {III}.
\newblock {\em Ann. of Math. (2) 77 (1963), 563--626; ibid.}, 78:1--40, 1963.

\bibitem{Kontsevich.Orsay}
M.~Kontsevich.
\newblock String cohomology, December 1995.
\newblock Lecture at Orsay.



  


\bibitem{Marsano}
J.~Marsano and S.~Sch\"afer-Nameki.
\newblock {Yukawas, G-flux, and Spectral Covers from Resolved Calabi-Yau's}.
\newblock {\em JHEP}, 11:098, 2011.


\bibitem{Matsuki.Weyl}
K.~Matsuki,  Weyl groups and birational transformations among minimal
  models, \href{http://dx.doi.org/10.1090/memo/0557}{{\em Mem. Amer. Math.
  Soc.} {\bfseries 116} no.~557, (1995) vi+133}.
  \bibitem{Mayrhofer:2014opa} 
  C.~Mayrhofer, D.~R.~Morrison, O.~Till and T.~Weigand,
   Mordell-Weil Torsion and the Global Structure of Gauge Groups in F-theory, 
  JHEP {\bf 1410}, 16 (2014)
  [arXiv:1405.3656 [hep-th]].


  
\bibitem{Miranda.smooth}
R.~Miranda.
\newblock Smooth models for elliptic threefolds.
\newblock In {\em The birational geometry of degenerations ({C}ambridge,
  {M}ass., 1981)}, volume~29 of {\em Progr. Math.}, pages 85--133.
  Birkh{\"a}user Boston, Mass., 1983.



\bibitem{Morrison:2012ei} 
  D.~R.~Morrison and D.~S.~Park,
  F-Theory and the Mordell-Weil Group of Elliptically-Fibered Calabi-Yau Threefolds, 
  JHEP {\bf 1210}, 128 (2012)
  doi:10.1007/JHEP10(2012)128
  
  \bibitem{Morrison:2011mb} 
  D.~R.~Morrison and W.~Taylor,
   Matter and singularities, 
  JHEP {\bf 1201}, 022 (2012)
  doi:10.1007/JHEP01(2012)022

\bibitem{Neron}A. N\'eron, Mod\`eles Minimaux des Vari\'et\'es Abeliennes sur les Corps Locaux et
Globaux, Publ. Math. I.H.E.S. 21, 1964, pp. 361--482.




\bibitem{Romans:1986er} 
  L.~J.~Romans,
  Selfduality for Interacting Fields: Covariant Field Equations for Six-dimensional Chiral Supergravities,
  Nucl.\ Phys.\ B {\bf 276}, 71 (1986).

  \bibitem{Sadov:1996zm} 
  V.~Sadov,
  Generalized Green-Schwarz mechanism in F theory,
  Phys.\ Lett.\ B {\bf 388}, 45 (1996)
  doi:10.1016/0370-2693(96)01134-3
  [hep-th/9606008].
  
  \bibitem{Sagnotti:1992qw} 
  A.~Sagnotti,
  A Note on the Green-Schwarz mechanism in open string theories,
  Phys.\ Lett.\ B {\bf 294}, 196 (1992)
  doi:10.1016/0370-2693(92)90682-T
  [hep-th/9210127].
  
  \bibitem{Schwarz:1995zw} 
  J.~H.~Schwarz,
  Anomaly - free supersymmetric models in six-dimensions,
  Phys.\ Lett.\ B {\bf 371}, 223 (1996)
  doi:10.1016/0370-2693(95)01610-4
  [hep-th/9512053].
  
  \bibitem{Tate}J.~Tate, 
Algorithm for determining the type of a singular fiber in an elliptic pencil. {\it Modular functions of one variable}, IV ({\it Proc. Internat. Summer School}, Univ. Antwerp, Antwerp, 1972), pp. 33--52. Lecture Notes in Mathematics 476, Springer-Verlag, Berlin, 1975.

\bibitem{Wazir}
R.~Wazir.
\newblock Arithmetic on elliptic threefolds.
\newblock {\em Compositio Mathematica}, 140(03):567--580, 2004.
\bibitem{Witten}
E.~Witten, Phase transitions in M theory and F theory, Nucl. Phys. B {\bf 471}, 195 (1996).

  
\end{thebibliography}
\end{document}